\pdfoutput=1
\documentclass[11pt, letterpaper]{article}
\usepackage{blindtext}
\usepackage{hyperref}
\hypersetup{
	colorlinks=true,
	linkcolor=blue!70!black,
	citecolor=blue!70!black,
	urlcolor=blue!70!black
}

\usepackage[margin=1in]{geometry}

\usepackage{amsfonts}
\usepackage{amsmath}
\usepackage{amssymb}
\usepackage{amsthm}
\usepackage{booktabs}
\usepackage{algorithm}
\usepackage[lined,boxed,ruled,norelsize,algo2e,linesnumbered,noend]{algorithm2e}
\usepackage{bbm}
\usepackage{bbold}

\usepackage{graphicx}
\graphicspath{{./figs/}}

\usepackage{cleveref}
\usepackage{thm-restate}

\newtheorem{theorem}{Theorem}

\newtheorem{lemma}[theorem]{Lemma}
\newtheorem{corollary}[theorem]{Corollary}
\newtheorem{conjecture}[theorem]{Conjecture}

\newtheorem{definition}[theorem]{Definition}
\newtheorem{proposition}[theorem]{Proposition}

\renewcommand{\epsilon}{\varepsilon}

\newcommand{\defeq}{\stackrel{\textup{\tiny def}}{=}}
\newcommand{\norm}[1]{\left\lVert#1\right\rVert}
\newcommand{\norms}[1]{\lVert#1\rVert}
\newcommand{\normop}[1]{\left\lVert#1\right\rVert_{\textup{op}}}
\newcommand{\normsop}[1]{\lVert#1\rVert_{\textup{op}}}
\newcommand{\normf}[1]{\left\lVert#1\right\rVert_{\textup{F}}}
\newcommand{\normsf}[1]{\lVert#1\rVert_{\textup{F}}}
\newcommand{\inprod}[2]{\left\langle#1, #2\right\rangle}
\newcommand{\inprods}[2]{\langle#1, #2\rangle}
\newcommand{\eps}{\epsilon}
\newcommand{\lam}{\lambda}

\newcommand{\R}{\mathbb{R}}

\newcommand{\N}{\mathbb{N}}

\newcommand{\gK}{\mathcal{K}}

\newcommand{\gT}{\mathcal{T}}
\newcommand{\gH}{\mathcal{H}}

\newcommand{\tgK}{\widetilde{\mathcal{K}}}
\newcommand{\diag}[1]{\textbf{\textup{diag}}\left(#1\right)}
\newcommand{\half}{\frac{1}{2}}

\newcommand{\1}{\mathbbm{1}}

\newcommand{\E}{\mathbb{E}}

\newcommand{\Nor}{\mathcal{N}}

\newcommand{\Tr}{\textup{Tr}}

\newcommand{\ma}{\mathbf{A}}

\newcommand{\id}{\mathbf{I}}

\usepackage{xcolor}
\definecolor{burntorange}{rgb}{0.8, 0.33, 0.0}

\newcommand{\nnz}{\textup{nnz}}

\newcommand{\Par}[1]{\left(#1\right)}
\newcommand{\Brack}[1]{\left[#1\right]}
\newcommand{\Brace}[1]{\left\{#1\right\}}
\newcommand{\Abs}[1]{\left|#1\right|}

\newcommand{\set}{\mathcal{K}}
\newcommand{\dd}{\textup{d}}
\newcommand{\Sym}{\mathbb{S}}
\newcommand{\Uni}{\mathbb{U}}
\newcommand{\msig}{\boldsymbol{\Sigma}}
\newcommand{\mmu}{\boldsymbol{\mu}}
\newcommand{\poly}{\textup{poly}}
\newcommand{\polylog}{\textup{polylog}}
\newcommand{\polyloglog}{\textup{polyloglog}}
\newcommand{\bO}{\Breve{O}}
\newcommand{\event}{\mathcal{E}}

\newcommand{\supp}{\textup{supp}}
\newcommand{\vol}{\textup{Vol}}

\newcommand{\spn}{\textup{span}}

\def\pleq{\preccurlyeq}

\renewcommand\AA{\boldsymbol{\mathrm{{A}}}}
\newcommand\BB{\boldsymbol{\mathrm{{B}}}}
\newcommand\CC{\boldsymbol{\mathrm{{C}}}}
\newcommand\DD{\boldsymbol{\mathrm{{D}}}}

\newcommand\HH{\boldsymbol{\mathrm{{H}}}}
\newcommand\II{\boldsymbol{\mathrm{{I}}}}

\newcommand\MM{\boldsymbol{\mathrm{{M}}}}
\newcommand\NN{\boldsymbol{\mathrm{{N}}}}
\newcommand\LL{\boldsymbol{\mathrm{{L}}}}
\newcommand\PP{\boldsymbol{\mathrm{{P}}}}
\newcommand\QQ{\boldsymbol{\mathrm{{Q}}}}
\newcommand\RR{\boldsymbol{\mathrm{{R}}}}

\newcommand\TT{\boldsymbol{\mathrm{{T}}}}
\newcommand\UU{\boldsymbol{\mathrm{{U}}}}
\newcommand\WW{\boldsymbol{\mathrm{{W}}}}
\newcommand\VV{\boldsymbol{\mathrm{{V}}}}
\newcommand\XX{\boldsymbol{\mathrm{{X}}}}
\newcommand\YY{\boldsymbol{\mathrm{{Y}}}}
\newcommand\ZZ{\boldsymbol{\mathrm{{Z}}}}

\newcommand\PPi{\boldsymbol{\mathrm{\Pi}}}

\newcommand{\Atil}{\widetilde{\AA}}
\newcommand{\Ahat}{\widehat{\AA}}
\newcommand{\Acal}{\mathcal{A}}

\newcommand{\hLL}{\widehat{\LL}}
\newcommand{\hDD}{\widehat{\DD}}

\newcommand{\ha}{\widehat{a}}
\newcommand{\hz}{\widehat{z}}

\renewcommand\aa{\boldsymbol{\mathrm{a}}}
\newcommand\bb{\boldsymbol{\mathrm{b}}}
\newcommand\cc{\boldsymbol{\mathrm{c}}}
\renewcommand\dd{\boldsymbol{\mathrm{d}}}
\newcommand\ee{\boldsymbol{\mathrm{e}}}

\renewcommand\gg{\boldsymbol{\mathrm{g}}}
\newcommand\hh{\boldsymbol{\mathrm{h}}}

\renewcommand\ll{\boldsymbol{\mathrm{\ell}}}

\newcommand\rr{\boldsymbol{\mathrm{r}}}
\renewcommand\ss{\boldsymbol{\mathrm{s}}}
\def\tt{\boldsymbol{\mathrm{t}}}
\newcommand\uu{\boldsymbol{\mathrm{u}}}
\newcommand\vv{\boldsymbol{\mathrm{v}}}
\newcommand\ww{\boldsymbol{\mathrm{w}}}
\newcommand\yy{\boldsymbol{\mathrm{y}}}
\newcommand\zz{\boldsymbol{\mathrm{z}}}
\newcommand\xx{\boldsymbol{\mathrm{x}}}

\newcommand\vone{\boldsymbol{\mathrm{1}}}
\newcommand\vzero{\boldsymbol{\mathrm{0}}}
\newcommand\vLL{\vec{\LL}}
\newcommand{\vG}{\vec{G}}
\newcommand{\vH}{\vec{H}}

\newcommand{\vC}{\vec{C}}
\newcommand{\rev}{\textup{rev}}
\newcommand{\und}{\textup{und}}
\newcommand{\blift}{\textup{blift}}
\newcommand{\ER}{\textup{ER}}

\newcommand{\hm}{\hat{m}}
\newcommand{\hn}{\hat{n}}
\newcommand{\hE}{\hat{E}}

\newcommand{\wws}{\ww_{\star}}
\newcommand{\uus}{\uu_{\star}}

\newcommand{\aas}{\aa_{\star}}
\newcommand{\bbs}{\bb_{\star}}
\newcommand{\xxs}{\xx_{\star}}
\newcommand{\yys}{\yy_{\star}}
\newcommand{\vHs}{\vH_{\star}}
\newcommand{\vGs}{\vG_{\star}}
\newcommand{\vCs}{\vC_{\star}}
\newcommand{\Gs}{G_\star}
\newcommand{\bw}{\bar{w}}
\newcommand{\trr}{\tilde{\rr}}
\newcommand{\bww}{\bar{\ww}}

\newcommand{\hx}{\hat{x}}
\newcommand{\hxx}{\hat{\xx}}
\newcommand{\hyy}{\hat{\yy}}
\renewcommand{\deg}{\textbf{\textup{deg}}}

\newcommand\ctight{c_{\textup{tight}}}
\newcommand\cset{C_{\textup{set}}}
\newcommand\ccolor{C_{\textup{color}}}
\newcommand\ctro{C_{\textup{tro}}}
\newcommand{\ceso}{C_{\textup{ESO}}}

\newcommand{\csign}{C_{\textup{sign}}}

\newcommand{\cps}{C_{\textup{PS}}}
\newcommand{\cbfs}{C_{\textup{BFS}}}

\newcommand{\cadk}{C_{\textup{ADK}}}
\newcommand{\cess}{C_{\textup{ESS}}}
\newcommand{\bv}{\bar{v}}
\newcommand{\bV}{\bar{V}}
\newcommand{\tsolve}{\gT_{\textup{solve}}}

\newcommand{\calF}{\mathcal{F}}
\newcommand{\calE}{\mathcal{E}}

\DontPrintSemicolon
\SetFuncSty{textsc}
\SetKwProg{Pro}{Procedure}{}{}
\SetKwFunction{PS}{DecompSparsify}
\SetKwFunction{ES}{ExistentialSparsify}
\SetKwFunction{ESO}{ExistentialDecompSparsify}
\SetKwFunction{ESF}{ExistentialSparsifyFractional}
\SetKwFunction{FSO}{FastSparsifyOnce}
\SetKwFunction{CS}{ClusterSparsify}
\SetKwFunction{Simplify}{Simplify}
\SetKwFunction{RO}{Rounding}
\SetKwFunction{BFS}{BasicFastSparsify}
\SetKwFunction{AER}{ApproxER}
\SetKwFunction{Coarsen}{Coarsen}
\SetKwFunction{PMRO}{ProjMinusRankOne}
\SetKwFunction{FS}{FastSparsify}
\SetKwFunction{ERP}{ERDecomp}
\SetKwFunction{DPRO}{DegreePreservingRounding}
\SetKwFunction{EP}{ExpanderDecompADK}
\SetKwFunction{BSS}{BasicSpectralSketch}
\SetKwFunction{PSS}{DecompSpectralSketch}
\SetKwFunction{ESS}{ExpanderSpectralSketch}
\SetKwFunction{SS}{SpectralSketch}
\SetKwFunction{GESO}{GeneralEulerianSparsifyOnce}
\SetKwFunction{GES}{GeneralEulerianSparsify}
\newcommand{\ESOalgo}{\textsc{ExistentialDecompSparsify}}
\newcommand{\ESalgo}{\textsc{ExistentialSparsify}}

\newcommand{\PSalgo}{\textsc{DecompSparsify}}
\newcommand{\ROalgo}{\textsc{Rounding}}

\newcommand{\BFSalgo}{\textsc{BasicFastSparsify}}

\newcommand{\PMROalgo}{\textsc{ProjMinusRankOne}}
\newcommand{\FSalgo}{\textsc{FastSparsify}}
\newcommand{\ERPalgo}{\textsc{ERDecomp}}

\newcommand{\EPalgo}{\textsc{ExpanderDecompADK}}

\newcommand{\ESSalgo}{\textsc{ExpanderSpectralSketch}}
\newcommand{\SSalgo}{\textsc{SpectralSketch}}

\newcommand{\codeInput}{\textbf{Input:} }
 \SetKwComment{Comment}{$\triangleright$\ }{}
\SetCommentSty{mycommfont}

\renewcommand{\eps}{\varepsilon}
\newcommand{\ind}{\mathbb{I}}

\usepackage[backend=biber, isbn=false, style=alphabetic, backref=true, doi=false, url=false, maxcitenames=10, mincitenames=5, maxalphanames=10, maxbibnames=10, minbibnames=5, minalphanames=5, defernumbers=true, sortlocale=en_US]{biblatex}
\title{Eulerian Graph Sparsification \\ by Effective Resistance Decomposition}
\author{Arun Jambulapati \\ University of Michigan \\ \texttt{jmblpati@gmail.com} 
  	\and 
   	Sushant Sachdeva \\ University of Toronto \\ \texttt{sachdeva@cs.toronto.edu}
  	\and
  	Aaron Sidford \\ Stanford University\\ \texttt{sidford@stanford.edu}
  	\and
  	Kevin Tian \\ University of Texas at Austin\\ \texttt{kjtian@cs.utexas.edu}
          \and
          Yibin Zhao \\University of Toronto \\ \texttt{ybzhao@cs.toronto.edu}
}
\date{}

\begin{document}
\maketitle
\thispagestyle{empty}

\begin{abstract}
We provide an algorithm that, given an $n$-vertex $m$-edge Eulerian graph with polynomially bounded weights, computes an $\breve{O}(n\log^{2} n \cdot \varepsilon^{-2})$-edge $\varepsilon$-approximate Eulerian sparsifier with high probability in $\breve{O}(m\log^3 n)$ time (where $\breve{O}(\cdot)$ hides $\polyloglog(n)$ factors). Due to a reduction from [Peng-Song, STOC '22], this yields an $\breve{O}(m\log^3 n + n\log^6 n)$-time algorithm for solving $n$-vertex $m$-edge Eulerian Laplacian systems with polynomially-bounded weights with high probability, improving upon the previous state-of-the-art runtime of $\Omega(m\log^8 n + n\log^{23} n)$. We also give a polynomial-time algorithm that computes $O(\min(n\log n \cdot \varepsilon^{-2} + n\log^{5/3} n \cdot \varepsilon^{-4/3}, n\log^{3/2} n \cdot \varepsilon^{-2}))$-edge sparsifiers,
improving the best such sparsity bound of $O(n\log^2 n \cdot
\varepsilon^{-2} + n\log^{8/3} n \cdot \varepsilon^{-4/3})$ [Sachdeva-Thudi-Zhao, ICALP '24]. 
Finally, we show that our techniques extend to yield the first $O(m\cdot\text{polylog}(n))$ time algorithm for computing $O(n\varepsilon^{-1}\cdot\text{polylog}(n))$-edge graphical spectral sketches, as well as a natural Eulerian generalization we introduce. 

In contrast to prior Eulerian graph sparsification algorithms which used either short cycle or expander decompositions, our algorithms use  a simple efficient effective resistance decomposition scheme we introduce. Our algorithms apply a natural sampling scheme and electrical routing (to achieve degree balance) to such decompositions. Our analysis leverages new asymmetric variance bounds specialized to Eulerian Laplacians and tools from discrepancy theory.

\end{abstract}

\newpage
\thispagestyle{empty}
\tableofcontents

\newpage
\setcounter{page}{1}
\section{Introduction}\label{sec:intro}

Over the past decade, ideas from spectral graph theory have led to a revolution in graph algorithms.
A major frontier for such developments is the design of spectral algorithms for directed graphs. 
Such algorithms have wide-ranging applications from fast algorithms for processing Markov chains (see e.g., \cite{CohenKPPSV16,AhmadinejadJSS19}) to deterministic low-space computation (see e.g., \cite{AhmadinejadKMPS20}).
A fundamental challenge in this setting is the fairly involved machinery used in spectral directed graph algorithms, which include efficient constructions of expander decompositions~\cite{CohenKPPRSV17} and short cycle decompositions~\cite{ChuGPSSW18}. In this paper we focus on the central topic of \emph{spectral sparsification of directed graphs}, for which, this challenge is particularly manifest.

A sparsifier of an undirected graph $G = (V,E, \ww)$ or directed graph $\vG$ is another graph supported on the same set of vertices with fewer edges, that approximately preserves some property.
Several notions of sparsification for undirected graphs have been studied in the literature, e.g., spanners~\cite{BaswanaS03, ThorupZ05}, which approximately preserve shortest path distances, and cut sparsifiers~\cite{BenczurK96}, which approximately preserve cut sizes.
Spectral sparsification~\cite{SpielmanT04} has been particularly influential in the design of graph algorithms.
An $\eps$-approximate undirected spectral sparsifier (henceforth, $\eps$\emph{-approximate undirected sparsifier}) $H = (V, E', \ww')$ of undirected $G$ approximately preserves the quadratic form of $G$'s graph Laplacian, i.e., for all $\xx \in \R^V,$ 
\begin{equation}\label{eq:undir_lap_approx}
\begin{gathered}
(1 - \eps) \xx^\top\LL_G\xx \le \xx^\top\LL_H\xx \le (1 + \eps) \xx^\top\LL_G\xx,
\text{ where } \xx^{\top}\LL_G\xx = \sum_{e = (u,v) \in E} \ww_e(\xx_u - \xx_v)^2,
\end{gathered}
\end{equation}
where $\LL_G$ and $\LL_H$ are the undirected Laplacian matrices of $G$ and $H$ (see Section~\ref{sec:prelims} for notation), and \eqref{eq:undir_lap_approx} is equivalent to $(1 - \eps)\LL_G \preceq \LL_H \preceq (1 + \eps)\LL_G$.
Spectral sparsification generalizes cut sparsification and was key to the advent of nearly-linear time Laplacian systems solvers~\cite{SpielmanT04}. 

Simple and efficient algorithms for computing undirected spectral sparsifiers with nearly-optimal guarantees are known.
Spielman and Srivastava~\cite{SpielmanS08} showed that independently sampling (and reweighting) $O(n\varepsilon^{-2} \log n)$ edges of an $n$-vertex graph, with probability proportional to their \emph{effective resistances} (a graph-theoretic analog of leverage scores), produces a spectral sparsifier. 
All effective resistances can be estimated in $\bO(m\log n)$ time\footnote{When discussing a graph clear from context with $n$ vertices and edge weight ratio bounded by $U$, we use the $\bO$ notation to hide $\polyloglog(nU)$ factors for brevity (in runtimes only).} using fast Laplacian system solvers~\cite{JambulapatiS21} (see Lemma~\ref{lem:approx_er}) -- this step dominates the runtime for undirected spectral sparsification. 
Additionally, Batson, Spielman, and Srivastava~\cite{BatsonSS09} showed spectral sparsifiers with $O(n\varepsilon^{-2})$ edges exist, which is optimal \cite{BatsonSS09, CarlsonKST19} and constructible in near-linear time \cite{LeeS17, JambulapatiRT23}.

Obtaining correspondingly simple and fast sparsification algorithms and optimal sparsity bounds for directed graphs remains elusive.
Even proposing useful notions of directed sparsification was challenging; any sparsifier of the complete, directed, bipartite graph, i.e., the graph with a directed edge from every node in one side of the bipartition to the other, that approximately preserves all directed cuts cannot delete any edges.
The influential work~\cite{CohenKPPRSV17} overcame this bottleneck by restricting their attention to directed Eulerian graphs (where every vertex has equal weighted in-degree and out-degree). 
Further, \cite{CohenKPPRSV17} showed that their sparsification notion suffices for numerous applications, including fast solvers for all directed Laplacian linear systems (not necessarily corresponding to an Eulerian graph), overviewed in Section~\ref{sec:apps}. In this paper, we consider the following definition of Eulerian sparsification closely related to that of  \cite{CohenKPPRSV17}.\footnote{The key difference is that we add the $E(\vH) \subseteq E$ restriction.}

\begin{definition}[Eulerian sparsifier]\label{def:eulerian_sparsifier}
$\vH$ is an \emph{$\eps$-approximate Eulerian sparsifier} of $\vG = (V,E,\ww)$ if $\vH$ and $\vG$ are both Eulerian, $V(\vH) = V$, and for $G \defeq \und(\vG)$,  we have
\begin{equation}\label{eq:eulerian_sparsifier_def}\normop{\LL_G^{\frac \dagger 2}\Par{\vLL_{\vG} - \vLL_{\vH}}\LL_G^{\frac \dagger 2}} \le \eps,\text{ and } E(\vH) \subseteq E.\end{equation}
\end{definition}

\Cref{def:eulerian_sparsifier} generalizes the notion of undirected sparsification (\Cref{fact:dirclose_undirclose}).
While useful in applications, Definition~\ref{def:eulerian_sparsifier} poses computational challenges. 
Eulerian sparsifiers preserve exact degree balance, so in contrast to undirected sparsifiers, one cannot simply sample edges independently to compute sparsifiers. There have been two broad approaches for addressing this key challenge.

The first approach leverages \emph{expander decompositions} and is related to one used in \cite{SpielmanT04} to sparsify undirected graphs. \cite{CohenKPPRSV17} followed such an approach and their algorithm consists of decomposing the Eulerian graph $\vG$ into expanders, sampling edges independently inside the expanders, and then fixing the resulting degree imbalance by adding edges; this resulted in sparsifiers that did not necessarily satisfy the $E(\vH) \subseteq E$ property in \eqref{eq:eulerian_sparsifier_def}. This approach was refined in~\cite{AhmadinejadPPSV23} (using cycle decompositions as in the second approach below, but not necessarily short ones), resulting in an algorithm for constructing Eulerian sparsifiers with $O(n\varepsilon^{-2}  \log^{20} n)$ edges in $O(m\log^7 n)$ time.
Existing near-linear time expander decomposition methods \cite{SaranurakW19, AgassyDK23} incur several logarithmic factors in the running time and (inverse) expansion quality, leading to these large, difficult to improve, polylogarithmic factors in the running time and sparsity.

The second approach leverages that most the edges in $\vG$ can be decomposed into edge-disjoint short cycles, termed a \emph{short cycle decomposition}.  \cite{ChuGPSSW18} pioneered this approach and sampled the edges in a coordinated manner within each cycle to preserve degree balance.
Advances in short cycle decompositions~\cite{LiuSY19, ParterY19, SachdevaTZ23} resulted in an $m^{1 + o(1)}$-time algorithm for constructing Eulerian sparsifiers with $O(n\varepsilon^{-2}\log^3 n)$ edges. 
Short cycle decompositions yield Eulerian sparsifier constructions with significantly improved sparsity compared to the expander decomposition approach, at the cost of large $m^{o(1)}$ factors in running time.

In summary, all prior algorithms for constructing Eulerian sparsifiers use either expander decomposition or short cycle decomposition, which result in substantial polylogarithmic factors (or larger) in sparsities and runtimes. More broadly, large gaps seem to remain in our understanding of efficient algorithms for constructing Eulerian sparsifiers and the optimal sparsity achievable. 

\subsection{Our results}\label{ssec:results}

We present a new sparsification framework that allows one to preserve exact degree balance while sampling, as in Eulerian sparsification, and yet analyze the sampling error as if the edges were sampled independently.
Our framework is simple and intuitive, as it is based on randomly signing multiplicative reweightings to edges, and using electrical flows to fix the degree balance.
Combining our framework with a lightweight graph-theoretic construction, \emph{effective resistance decomposition} (Definition~\ref{def:er_partition}), we obtain the following Eulerian sparsification result.

\begin{restatable}{theorem}{restatefastsparsify}\label{thm:fastsparsify}
Given Eulerian $\vG = (V, E, \ww)$ with $|V| = n$, $|E| = m$, integral $\ww \in [1, \poly(n)]^E$ and $\eps \in (0, 1)$, $\FSalgo$ (Algorithm~\ref{alg:fastsparse}) in $\bO\Par{m\log^3 n}$ time returns Eulerian $\vH$ that w.h.p.,\footnote{In the introduction only, we use the abbreviation ``w.h.p.'' (``with high probability'') to mean that a statement holds with $n^{-C}$ failure probability for an arbitrarily large constant $C$ (which affects other constants in the statement). In the formal variants of theorem statements later in the paper, we state precise dependences on failure probabilities.} is an $\eps$-approximate Eulerian sparsifier of $\vG$ with $|E(\vH)| =
    O\Par{{n}{\eps^{-2}}\log^2(n)\log^2\log\Par{n}}$\,.
\end{restatable}

Theorem~\ref{thm:fastsparsify} constructs Eulerian sparsifiers with sparsity within a $\bO(\log^2 n)$ factor of optimal \cite{CarlsonKST19}, in time $\bO( m \log^3 n)$. Our algorithm simultaneously achieves a substantially faster runtime than prior Eulerian sparsification schemes and improves the state-of-the-art sparsity bound (see Table~\ref{tab:sota}). For instance, the prior state-of-the-art Eulerian sparsification algorithm with both $O(n\eps^{-2}\cdot \polylog(n))$ edges and a $O(m \cdot \polylog(n))$ runtime has (up to $O(\poly \log \log n))$ factors an extra $\Omega(\log^{18} n)$ factor in sparsity and an $\Omega(\log^4 n)$ factor in the runtime compared to \Cref{thm:fastsparsify}.

As a corollary of our fast sparsification algorithm (\Cref{thm:fastsparsify}), reductions due to Peng and Song \cite{PengS22} and earlier works on solving (variants of) directed Laplacian systems \cite{CohenKPPSV16, CohenKPPRSV17, AhmadinejadJSS19}, we obtain a host of additional results. The following is a straightforward corollary obtained by a direct reduction given in the main result of \cite{PengS22}.

\begin{restatable}[Eulerian Laplacian solver]{corollary}{restateeulsolver}
\label{cor:eulsolver}
There is  an algorithm which given input Eulerian $\vG = (V,E,\ww)$ with $|V|=n$, $|E|=m$,  $\ww \in [1,\poly(n)]^E$, and  $\bb \in \R^V$, in $\bO\Par{m\log^3\Par{n} + n\log^6\Par{n}}$ time returns $\xx \in \R^V$  satisfying,
 w.h.p.,
$\|\xx - \vLL_{\vG}^\dagger \bb\|_{\LL_G} \leq \eps \|\vLL_{\vG}^\dagger \bb\|_{\LL_G}$ for $G \defeq \und(\vG)$.
\end{restatable}

The runtime of Corollary~\ref{cor:eulsolver} improves upon the prior state-of-the-art claimed in the literature of $\bO( m\log^8 n + n\log^{23} n)$ (see Appendix C, \cite{PengS22}). Up to small polylogarithmic factor overheads in runtimes, our Eulerian Laplacian solver also implies a solver for all directed Laplacians 
(Corollary~\ref{cor:dirlapsolver}), and fast high-accuracy approximations for directed graph primitives such as computation of stationary distributions, mixing times, Personalized PageRank vectors, etc., as observed by \cite{CohenKPPSV16, AhmadinejadJSS19}. We state these additional applications in Section~\ref{sec:apps}.

We further ask: what is the optimal number of edges in an Eulerian sparsifier? By combining our new approach with recent advances in discrepancy theory due to Bansal, Jiang, and Meka \cite{BansalJM23}, we obtain the following improved sparsity bound over Theorem~\ref{thm:fastsparsify}. 

\begin{restatable}{theorem}{restateexistential}\label{thm:existential}
    Given Eulerian $\vec{G} = (V,E,\ww)$ with $|V| = n$, $|E| = m$, $\ww \in [1,\poly(n)]^E$ and $\epsilon \in (0,1)$, \ES (Algorithm~\ref{alg:existsparse_new}) in  $\poly(n,\eps^{-1})$ time returns Eulerian $\vec{H}$ such that w.h.p.\ $\vH$ is an $\eps$-approximate Eulerian sparsifier of $\vG$ with
    \begin{gather*}
    |E(\vH)| = O\left(\min\Brace{\frac{n\log n}{\epsilon^2} + \frac{n \log^{5/3} n}{\epsilon^{4/3}},\; \frac{n \log^{3/2} n}{\eps^{2}}}\right).
    \end{gather*}
\end{restatable}

\begin{table}[ht!]
    \centering
    \begin{tabular}{{c}{c}{c}{c}}
    \toprule
      Method  &  Sparsity & Runtime & Approach \\
      \midrule
      \cite{CohenKPPRSV17} & $n\eps^{-2}\log^C n$ & $m \log^C n$ & expanders \\
      \cite{ChuGPSSW18} & $n\eps^{-2}\log^4 n$ & $mn$ & short cycles \\
      \cite{ChuGPSSW18, LiuSY19, ParterY19} & $n\eps^{-2} \log^{k} n $ & $m + n^{1+O(\frac 1 k)}$ & short cycles \\
      \cite{ParterY19} & $n^{1+o(1)} +n\eps^{-2}\log^4 n$  & $m \log^C n$ & short cycles \\
      \cite{AhmadinejadPPSV23} & $n\eps^{-2}\log^{12} n$ & \textup{existential} & SV sparsification \\
      \cite{AhmadinejadPPSV23} & $n\eps^{-2}\log^{20} n$ & $m\log^7 n$ & SV sparsification \\
      \cite{ParterY19, SachdevaTZ23} & $n\eps^{-2}\log^3 n$& 
 $m^{1 + \delta}$ & short cycles \\
      \cite{SachdevaTZ23} & $n\eps^{-2} \log^2 n + n\eps^{-4/3}\log^{8/3}n$ & $n^C$ & short cycles \\
       \midrule
      Theorem~\ref{thm:fastsparsify} & $n\eps^{-2}\log^2 n$ & $m\log^3 n $& ER decomposition \\
      Theorem~\ref{thm:existential} & $n\eps^{-2} \log n + n\eps^{-4/3}\log^{5/3} n$ & $n^C$ & ER decomposition \\
      Theorem~\ref{thm:existential} & $n\eps^{-2} \log^{3/2} n$ & $n^C$ & ER decomposition \\
       \bottomrule
    \end{tabular}
    \caption{\textbf{Eulerian sparsification algorithms.} All results apply to Eulerian $\vG = (V, E, \ww)$ with $n \defeq |V|$ and $m \defeq |E|$. For simplicity, $\ww \in [1, \poly(n)]^E$ and all algorithms fail with probability $\poly(\frac 1 n)$. $C$ denotes an unspecified (large) constant, $\delta$ denotes an arbitrarily small constant, and we hide $\polyloglog(n)$ factors. The third row requires $k \ge 4$. The \cite{CohenKPPRSV17} sparsifiers were not reweighted subgraphs of the original graph, but all other sparsifiers in this table are.}
    \label{tab:sota}
\end{table}

For $\eps \le \log^{-1} n$, Theorem~\ref{thm:existential} establishes that $O(n\eps^{-2} \log n)$-edge Eulerian sparsifiers exist and are constructible in polynomial time. Moreover for any $\eps,$ the sparsity is at most $n\eps^{-2}\log^{\frac{3}{2}} n$. 
In Appendix~\ref{sec:conjectures}, we discuss potential directions towards showing the existence of even sparser Eulerian sparsifiers, e.g., with only $O(n\eps^{-2})$ nonzero edge weights (matching the optimal sparsity for undirected graph sparsifiers \cite{BatsonSS09, CarlsonKST19}).

We further demonstrate the power of our framework by giving an efficient construction of graphical spectral sketches~\cite{AndoniCKQWZ16, JambulapatiS18, ChuGPSSW18}, i.e., sparse graphs which satisfy \eqref{eq:undir_lap_approx} 
for any fixed vector $\xx \in \R^V$ w.h.p.\ (rather than for all $\xx \in \R^V$). The only previously known construction of graphical spectral sketches was based on short cycle decompositions~\cite{ChuGPSSW18, LiuSY19, ParterY19}. 
We provide an algorithm that efficiently computes sparse weighted subgraphs that are simultaneously graphical spectral sketches, spectral sparsifiers (for a larger value of $\epsilon$), and sketches of the pseudoinverse in a suitable sense.

\begin{theorem} \label{thm:fastsketch_undir}
    There is an algorithm that, given undirected graph $G = (V,E,\ww)$ with $|V| =
    n$, $|E| = m$, $\ww \in [1, \poly(n)]^E$ and $\eps \in (0, \frac{1}{100})$,  in $\bO\Par{m\log^{9}(n)}$ time
    returns an undirected graph $H$ such that $|E(H)| = 
    O\Par{n \epsilon^{-1} \log^9(n)
    	\log^2\log(n)}$ and the following properties hold.
    \begin{enumerate}
    \item $H$ is a $\sqrt{\eps}$-approximate spectral sparsifier of $G$ w.h.p.
    \item $H$ is a $\eps$-approximate graphical sketch of $G$, i.e., for an
        arbitrarily fixed vector $\xx \in \R^V$, w.h.p.\ over $H$, $\left|\xx^\top (\LL_H - \LL_G) \xx\right| \le \eps \cdot \xx^\top \LL_H \xx$.
    \item $H$ is a $\eps$-approximate inverse sketch of $G$, i.e., for an
        arbitrarily fixed vector $\xx \in \R^V$, w.h.p.\ over $H$,
    $\left|\xx^\top (\LL_H^\dagger - \LL_G^\dagger) \xx\right|
        \le \eps \cdot \xx^\top \LL_H^\dagger \xx$.
    \end{enumerate}
    \end{theorem}

    While this more general guarantee was also achieved by the short-cycle decomposition based constructions, the previous best construction of a graphical spectral sketch with $n\varepsilon^{-1}\cdot \polylog (n)$ edges required $m^{1+o(1)}$ time~\cite{ParterY19}. 
    Additionally, in \Cref{sec:sketch} we generalize this notion of graphical spectral sketches to Eulerian graphs (Definition~\ref{def:sketch_directed}) and provide analogous runtimes and sparsity bounds for such sketches (Theorem~\ref{thm:fastsketch}); these are the first such results to the best of our knowledge.
    
\subsection{Overview of approach}\label{ssec:summary}

In this paper, we provide a new, simpler framework for sparsifying Eulerian graphs. Despite its simplicity, our approach yields Eulerian sparsification algorithms which improve upon prior work in both runtime and sparsity. We briefly overview our framework and technical contributions here; see Section~\ref{sec:overview} for a more detailed technical overview.

Our framework is motivated by the following simple undirected graph sparsification algorithm.
\begin{itemize}
	\item For all edges $e \in E$ with an effective resistance (ER) smaller than $\rho$, toss an independent coin and either drop the edge or double its weight.
	\item Repeat until there are no edges left with a small ER.
\end{itemize}

It is straightforward to show that this algorithm produces a spectral sparsifier. In each iteration, the algorithm's relative change to the Laplacian (in a multiplicative sense) is  $\sum_{e \in E} \ss_e \AA_e,$ where $\ss_e$ is a random $\pm 1$ sign and $\AA_e = \ww_e \LL_G^{\dagger/2} \bb_e \bb_e^{\top} \LL_G^{\dagger/2}$ denotes the normalized contribution of the edge Laplacian. The key step of the analysis is bounding the total matrix variance $\sum_{e \in E} \AA_e \AA_e^{\top}$, across all iterations. When setting $\rho = c \frac n m$ where $m$ is the current number of edges and $c$ is a sufficiently large constant, the variance contribution for each edge forms an increasing geometric progression (as $m$ decreases geometrically) where the sum is bounded by the last term. Moreover, each edge Laplacian only contributes if its leverage score is at most $\rho$, so $\AA_e \AA_e^{\top} \preceq \rho \AA_e.$ Summing over all edges, the total matrix variance is $\preceq \rho \II$. 
Stopping when $\rho = O(\frac{\eps^2}{\log n})$ for an appropriate constant, standard matrix concentration bounds then show the total relative spectral error is $O(\sqrt{\rho \log n}) \cdot \II = \eps \II$.

Emulating such a strategy for Eulerian graphs faces an immediate obstacle: adding and dropping edges independently might result in a non-Eulerian graph, i.e., one that does not satisfy the  degree balance constraints of an Eulerian graph. In fact, there may be no setting of $\ss \in \{\pm 1\}^E$ for which the relative change in edge weights, $\ww \circ \ss$, satisfies the necessary degree balance.
As mentioned previously, one approach to Eulerian sparsification~\cite{CohenKPPRSV17} independently samples $\pm 1$ signs for edges inside an expander, fixes the resulting degree imbalance, and uses the expansion property to bound the resulting error. 
Another approach, based on short cycle decomposition~\cite{ChuGPSSW18}, toggles cycles, keeping either only the clockwise or counterclockwise edges, thus ensuring degrees are preserved. Additionally,  \cite{AhmadinejadPPSV23} samples $\pm 1$ signs for cycles (not necessarily short) inside an expander. Each of these results in large polylogarithmic factors or worse in their guarantees, due to limitations in algorithms for expander or short-cycle decomposition.

To obtain faster and simpler algorithms with improved sparsity guarantees, we take an alternative approach. As a starting point, consider sampling a random signing $\ss$ on edge Laplacians, and projecting $\ss$ down to the degree balance-preserving subspace. We make the simple, yet crucial, observation: this projection step does not increase the matrix variance (Lemma~\ref{lemma:varianceproj})! This fact, which lets us bound spectral error as we would if all edge signings were independent, has not been exploited previously for efficient degree balance-preserving sparsification to our knowledge.

Our second key contribution is recognizing that to bound the variance of an independent edge Laplacian signing in a subgraph, requiring the subgraph to be an expander is stronger than necessary. In Lemma~\ref{lemma:variance}, we show it suffices to work in subgraphs with bounded ER diameter (implied by expansion in high-degree unweighted graphs, cf.\ Lemma~\ref{lemma:ex_to_er}). 
Decomposing a graph into low ER diameter pieces can be achieved more simply, efficiently, and with better parameters (for our purposes) as compared to expander or short cycle decompositions (Proposition~\ref{prop:er_partition}).

To implement this approach to Eulerian sparsification efficiently, we overcome several additional technical hurdles. 
The first one is ensuring (in nearly-linear time) that the updated edge weight vector is nonnegative; negative weight edges could occur when projecting a large vector to the degree-preserving space. In previous discrepancy works, e.g., \cite{Rothvoss17}, this problem was alleviated by projecting the random vector to the intersection of the subspace with the $\pm 1$ hypercube. This projection is expensive; on graphs it could be implemented with oblivious routings, but unfortunately, the fastest routings of sufficient quality in the literature do not run in nearly-linear time. We show that by scaling down the step size by a polylogarithmic factor and appealing to sub-Gaussianity of random projection vectors, we can ensure the nonnegativity of weights.

Secondly, since the weight updates are small in magnitude, there is no immediate reduction in sparsity. 
Using a careful two-stage step size schedule (see discussion in Section~\ref{sec:algos}), we give a potential argument showing that after adding roughly $\log^2 (n)$ random signings, each projected by solving an undirected Laplacian system, suffices to make a constant fraction of the weights tiny. These tiny edge weights can then be rounded to zero, decreasing the sparsity by a constant factor. 
Combining our framework with state-of-the-art undirected Laplacian solvers gives our overall runtime of $\bO(m\log^3 (n))$ in Theorem~\ref{thm:fastsparsify}.
    
\subsection{Related work}\label{ssec:related}

\paragraph{Undirected sparsifiers and Laplacian solvers.}
 The first nearly-linear time algorithm for solving undirected Laplacian linear systems was obtained in groundbreaking work of Spielman and Teng~\cite{SpielmanT04}. Since then, there has been significant work on developing faster undirected Laplacian solvers~\cite{KoutisMP10, KoutisMP11, PengS14, CohenKMPPRS14, KyngLPSS16, KyngS16, JambulapatiS21, ForsterGLPSY22, SachdevaZ23}, culminating in an algorithm that runs in $\bO(m\log \frac 1 \eps)$ time for approximately solving undirected Laplacian linear systems up to expected relative error $\eps$ (see Proposition~\ref{prop:js21} for a formal statement).

The first spectral sparsifiers for undirected graphs were constructed by Spielman and Teng~\cite{SpielmanT04}, which incurred significant polylogarithmic overhead in their sparsity. Spielman and Srivastava~\cite{SpielmanS08} then gave a simple algorithm for constructing undirected spectral sparsifiers with $O(n\varepsilon^{-2} \log n)$ edges in nearly-linear time. Batson, Spielman, and Srivastava~\cite{BatsonSS09} gave a polynomial time algorithm for constructing undirected spectral sparsifiers with $O(n\varepsilon^{-2})$ edges, and established that this sparsity bound is optimal. Faster algorithms for $O(n\varepsilon^{-2})$-edge undirected sparsifiers were later given in~\cite{LeeS17, LeeS18, JambulapatiRT23}. We also mention an additional notion of sparsification in undirected graphs, \emph{degree-preserving sparsification}, which has been studied in the literature as an intermediary between undirected and Eulerian sparsification \cite{ChuGPSSW18, JambulapatiRT23}. Degree-preserving undirected sparsifiers of sparsity $O(n\eps^{-2})$ were recently shown to exist and be constructible in almost-linear time by \cite{JambulapatiRT23}, motivating our work in the related Eulerian sparsification setting.

\paragraph{Eulerian sparsifiers and directed Laplacian solvers.} The study of efficient directed Laplacian solvers was initiated by Cohen, Kelner, Peebles, Peng, Sidford, and Vladu \cite{CohenKPPSV16}, who established that several computational problems related to random walks on directed graphs can be efficiently reduced to solving linear systems in Eulerian Laplacians. This work also gave an algorithm for solving Eulerian Laplacian linear systems in $O((mn^{2/3} + m^{3/4}n) \cdot \polylog( n))$ time, the first such solver with a runtime faster than that known for linear system solving in general.
Subsequently, the aforementioned authors and Rao \cite{CohenKPPRSV17} introduced the notion of Eulerian sparsifiers and gave the first $O(m \cdot \polylog(n))$-time algorithm for constructing Eulerian sparsifiers with  $O(n\varepsilon^{-2}\cdot\polylog(n))$ edges, based on expander decompositions. They used their method to give the first $m^{1+o(1)}$ time algorithm for solving linear systems in directed Eulerian Laplacians. A follow-up work by the aforementioned authors and Kyng~\cite{CohenKKPPRSV18} later gave an improved $O(m \cdot \polylog(n))$-time solver for directed Laplacian linear systems.

As an alternative approach to Eulerian sparsification, Chu, Gao, Peng, Sachdeva, Sawlani, and Wang~\cite{ChuGPSSW18} introduced the short cycle decomposition, and used it to give an $O(mn)$ time algorithm for computing Eulerian sparsifiers with $O(n\varepsilon^{-2} \log^4 n)$ edges.
Improved short cycle decomposition constructions by Liu, Sachdeva, and Yu~\cite{LiuSY19}, as well as Parter and Yogev~\cite{ParterY19} resulted in an improved running time of $O(m^{1+\delta})$ for any constant $\delta > 0,$ for the same sparsity.

Very recently, Sachdeva, Thudi, and Zhao~\cite{SachdevaTZ23} gave an improved analysis of the short cycle decomposition-based construction of Eulerian sparsifiers from \cite{ChuGPSSW18}, improving the resulting sparsity to $O(n\varepsilon^{-2} \log^3 n)$ edges. They complemented their algorithmic construction with an existential result showing that Eulerian sparsifiers with $\bO(n\varepsilon^{-2} \log^2 n + n\varepsilon^{-4/3} \log^{8/3} n)$ edges exist, using recent progress on the matrix Spencer's conjecture~\cite{BansalJM23}. Our fast algorithm in Theorem~\ref{thm:fastsparsify} yields an improved sparsity compared to the strongest existential result in \cite{SachdevaTZ23} with a significantly improved runtime, and departs from the short cycle decomposition framework followed by that work. Moreover, our existential result in Theorem~\ref{thm:existential}, which also applies \cite{BansalJM23} (combined with our new framework), improves \cite{SachdevaTZ23}'s existential result by a logarithmic factor.

Finally, we note that our applications in Section~\ref{sec:apps} follow from known implications in the literature, e.g., \cite{CohenKPPSV16, AhmadinejadJSS19, PengS22}. In particular, our directed Laplacian linear system solver follows from reductions in \cite{CohenKPPSV16, PengS22}, who showed that an efficient Eulerian sparsification algorithm implies efficient solvers for all directed Laplacian linear systems. Building upon this result, our other applications follow \cite{CohenKPPSV16, AhmadinejadJSS19}, which show how various other primitives associated with Markov chains can be reduced to solving appropriate directed Laplacian systems.

\paragraph{Discrepancy-theoretic approaches to sparsification.} The use of discrepancy-theoretic techniques for spectral sparsification has been carried out in several prior works. First, \cite{ReisR20} showed how to use matrix variance bounds in undirected graphs with the partial coloring framework of \cite{Rothvoss17} to construct linear-sized sparsifiers. Subsequently, this partial coloring-based sparsification algorithm was sped up to run in nearly-linear time by \cite{JambulapatiRT23} and \cite{SachdevaTZ23} showed how to adapt these techniques to the Eulerian sparsification setting, by using an improved analysis of the matrix variance induced by algorithms using short cycle decompositions. 

Our strongest existential sparsification result (cf.\ Theorems~\ref{thm:existential},~\ref{thm:existentialdetailed}) follows the discrepancy-based partial coloring approach to sparsification pioneered in these works, combining it with our new matrix variance bounds via ER decomposition (Lemma~\ref{lemma:variance}) instead of short cycles, as was done in \cite{SachdevaTZ23}. Recently, concurrent and independent work of \cite{LauWZ24} gave a derandomized partial colouring framework for spectral sparsification using the ``deterministic discrepancy walk" approach from \cite{PesentiV23}, and applied it to obtain polynomial-time deterministic Eulerian sparsifiers satisfying a stronger notion of spectral approximation known as ``singular value (SV) approximation'' \cite{AhmadinejadPPSV23}. This result of  \cite{LauWZ24} complements, but is largely orthogonal to, our results: it yields directed sparsifiers with larger sparsities and runtimes than ours, but which satisfy stronger notions of sparsification (i.e., SV sparsification) and are obtained deterministically.

\subsection{Roadmap}\label{ssec:roadmap}

In Section~\ref{sec:prelims}, we introduce notation and useful technical tools used  throughout the paper. In  Section~\ref{sec:overview} we then provide a technical overview of the rest of the paper. Next, we give our effective resistance decomposition algorithm in Section~\ref{sec:partition}, a key building block in our sparsification methods. In Section~\ref{sec:variance}, we then show how to take advantage of this decomposition by proving a new matrix variance bound for directed edge Laplacians after an electric projection. Crucially, this bound is parameterized by the effective resistance diameter of decomposition pieces.

The remainder of the paper contains applications of our sparsification framework. In Section~\ref{sec:existence}, we prove Theorem~\ref{thm:existential}, our result with the tightest sparsity guarantees. In Section~\ref{sec:algos}, we prove Theorem~\ref{thm:fastsparsify}, which obtains a significantly improved runtime at the cost of slightly worse sparsity. In Section~\ref{sec:apps}, we combine our sparsification methods with existing reductions in the literature and overview additional applications of our algorithms for directed graph primitives. 
Finally, in Section~\ref{sec:sketch}, we show how to apply our sparsification subroutines to design state-of-the-art graphical spectral sketches, proving Theorem~\ref{thm:fastsketch_undir} and an extension to Eulerian graphs that we introduce.

\section{Preliminaries}\label{sec:prelims}

\paragraph{General notation.} All logarithms are base $e$ unless otherwise specified. When discussing a graph clear from context with $n$ vertices and edge weight ratio bounded by $U$, we use the $\bO$ notation to hide $\polyloglog(nU)$ factors for brevity (in runtimes only). We let $[n] \defeq \{i \in \N \mid 1 \le i \le n\}$. 

\paragraph{Vectors.}
Vectors are denoted in lower-case boldface.  $\vzero_d$ and $\vone_d$ are the all-zeroes and all-ones vector respectively of dimension $d$.  $\ee_i$ denote the $i^{\text{th}}$ basis vector.   $\uu \circ \vv$ denotes the entrywise product of $\uu, \vv$ of equal dimension. 

\paragraph{Matrices.}
Matrices are denoted in upper-case boldface. We refer to the $i^{\text{th}}$ row and $j^{\text{th}}$ column of matrix $\MM$ by $\MM_{i:}$ and $\MM_{:j}$ respectively. We use $[\vv]_i$ to index into the $i^{\text{th}}$ coordinate of vector $\vv$, and let $[\MM]_{i:} \defeq \MM_{i:}$, $[\MM]_{:j} \defeq \MM_{:j}$, and $[\MM]_{ij} \defeq \MM_{ij}$ in contexts where $\vv$, $\MM$ have subscripts.

$\id_d$ is the $d \times d$ identity matrix.  For $\vv \in \R^d$, $\diag{\vv}$ denotes the associated diagonal $d \times d$ matrix. For linear subspace $S$ of $\R^d$, $\dim(S)$ is its dimension and $\PP_S$ is the orthogonal projection matrix onto $S$. We let $\ker(\MM)$ and $\MM^{\dagger}$ denote the kernel and pseudoinverse of $\MM$.  We denote the operator norm (largest singular value) of matrix $\MM$ by $\normop{\MM}$, and the Frobenius norm (entrywise $\ell_2$ norm) of $\MM$ by $\normf{\MM}$. The number of nonzero entries of a matrix $\MM$ (resp.\ vector $\vv$) is denoted $\nnz(\MM)$ (resp.\ $\nnz(\vv)$), and the subset of indices with nonzero entries is $\supp(\MM)$ (resp.\ $\supp(\vv)$). 

We use $\preceq$ to denote the Loewner partial order on $\Sym^d$, the symmetric $d \times d$ matrices. We let $\Uni^d$ denote the set of $d \times d$ real unitary matrices. For $\MM \in \Sym^d$ and $i \in [d]$, we let $\lam_i(\MM)$ denote the $i^{\text{th}}$ smallest eigenvalue of $\MM$, so $\lam_1(\MM) \le \lam_2(\MM) \le \ldots \le \lam_d(\MM)$. For positive semidefinite $\AA \in \Sym^d$, we define the seminorm induced by $\AA$ by $\|\xx\|_{\AA}^2 \defeq \xx^\top \AA \xx$.

\paragraph{Distributions.} $\textup{Geom}(p)$ for $p \in (0, 1]$ denotes the geometric distribution on $\N$ with mean $\frac 1 p$. $\Nor(\mmu, \msig)$ denotes the multivariate normal distribution with mean $\mmu$ and covariance $\msig$. $\gamma_d$ denotes the Gaussian measure in dimension $d$, i.e., for $\set \subseteq \R^d$, $\gamma_d(\set) \defeq \Pr_{g \sim \Nor(\vzero_d, \id_d)}[g \in \set]$; when $S$ is a linear subspace of $\R^d$, we define $\gamma_S(\set) \defeq \Pr_{g \sim \Nor(\vzero_d, \PP_S)}[g \in \set]$. 

\paragraph{Graphs.} All graphs throughout this paper are assumed to be simple without loss of generality, as collapsing parallel multi-edges does not affect (undirected or directed) graph Laplacians. We denote undirected weighted graphs without an arrow and directed weighted graphs with an arrow, i.e., $G = (V, E, \ww)$ is an undirected graph with vertices $V$, edges $E$, and weights $\ww \in \R_{\ge 0}^E$, and $\vec{G}$ is a directed graph.
A directed Eulerian graph is a directed graph where weighted in-degree equals weighted out-degree for every vertex.
We refer to the vertex set and edge set of a graph $G$ (resp.\ $\vec{G}$) by $V(G)$ and $E(G)$ (resp.\ $V(\vec{G})$ and $E(\vec{G})$). We associate a directed edge $e$ from $u$ to $v$ with the tuple $(u, v)$, and an undirected edge with $(u, v)$ and $(v, u)$ interchangeably. We define $h(e) = u$ and $t(e) = v$ to be the head and tail of a directed edge $e = (u, v)$. 

Finally, when we are discussing Eulerian sparsification of a graph $\vG$ in the sense of Definition~\ref{def:eulerian_sparsifier}, we will always assume henceforth that $G = \und(\vG)$ is connected. This is without loss of generality: otherwise, we can define an instance of Definition~\ref{def:eulerian_sparsifier} on each connected component of $G$. The left and right kernels of $\vLL_{\vG}$ and $\LL_{G}$ are spanned by the all-ones vectors indicating each connected component of $G$. Moreover, each connected component in $G$ still corresponds to an Eulerian graph. Therefore, satisfying Definition~\ref{def:eulerian_sparsifier} for each component individually implies the same inequality holds for the entire graph, by adding all the component Laplacians. 

\paragraph{Subgraphs and graph operations.} We say $H$ is a subgraph of $G$ if the edges and vertices of $H$ are subsets of the edges and vertices of $G$ (with the same weights), denoting $H = G_{F}$ if $E(H) = F$, and defining the same notion for directed graphs. For $U \subseteq V$, we let $G[U]$ denote the induced subgraph of $G$ on $U$ (i.e., keeping all of the edges within $U$). We let $\rev(\vec{G})$ denote the directed graph with all edge orientations reversed from $\vec{G}$, and $\und(\vec{G})$ denote the undirected graph which removes orientations (both keeping the same weights). When $V$ is a set of vertices, we say $\{V_i\}_{i \in [I]}$ is a partition of $V$ if $\bigcup_{i \in [I]} V_i = V$, and all $V_i$ are disjoint. We say $\{G_j\}_{j \in [J]}$ are a family of edge-disjoint subgraphs of $G = (V, E, \ww)$ if all $E(G_j)$ are disjoint, and for all $j \in [J]$, $V(G_j) \subseteq V$, $E(G_j) \subseteq E$, and every edge weight in $G_j$ is the same as its weight in $G$.

\paragraph{Graph matrices.} For a graph with edges $E$ and vertices $V$, we let $\BB \in \{-1, 0, 1\}^{E \times V}$ be its edge-vertex incidence matrix, so that when $\vG$ is directed and $e = (u, v)$, $\BB_{e:}$ is $2$-sparse with $\BB_{eu} = 1$, $\BB_{ev} = -1$ (for undirected graphs, we fix an arbitrary consistent orientation). For $u, v \in V$, we define $\bb_{(u, v)} \defeq \ee_u - \ee_v$. When $\BB$ is the incidence matrix associated with graph $G = (V, E, \ww)$ (resp.\ $\vec{G}$), we say $\xx$ is a circulation in $G$ (resp.\ $\vec{G}$) if $\BB^\top \xx = \vzero_V$; when $G$ (resp.\ $\vec{G}$) is clear we simply say $\xx$ is a circulation. We let $\HH, \TT \in \{0, 1\}^{E \times V}$ indicate the heads and tails of each edge, i.e., have one nonzero entry per row indicating the relevant head or tail vertex for each edge, respectively, so that $\BB = \HH - \TT$. When clear from context that $\ww$ are edge weights, we let $\WW \defeq \diag{\ww}$.  For undirected $G = (V, E, \ww)$ with incidence matrix $\BB$, the Laplacian matrix of $G$ is $\LL \defeq \BB^\top \WW \BB$. For directed $\vec{G} = (V, E, \ww)$, the directed Laplacian matrix of $\vec{G}$ is $\vec{\LL} \defeq \BB^\top \WW \HH$. To disambiguate, we use $\LL_G$, $\HH_G$, $\TT_G$, $\BB_G$, etc.\ to denote matrices associated with a graph $G$ when convenient.

Note that $\vec{\LL}^\top \vone_V = \vzero_V$ for any directed Laplacian $\vec{\LL}$. If $\vG$ is Eulerian, then its directed Laplacian also satisfies $\vLL \vone_V = \vzero_V$ and $\ww$ is a circulation in $\vG$ (i.e., $\BB^\top \ww = \vzero_V$).
Note that for a directed graph $\vG = (V,E,\ww)$ and its corresponding undirected graph $G \defeq \und(\vG)$, the undirected Laplacian is $\LL_G = \BB^\top \WW \BB$, and the reversed directed Laplacian is $\vLL_{\rev(\vG)} = -\BB^\top \WW \TT$.

We let $\PPi_V$ denote the Laplacian of the unweighted complete graph on $V$, i.e., $\PPi_V \defeq \id_V - \frac 1 {|V|} \vone_V\vone_V^\top$. Note that $\PPi_V$ is the orthogonal projection on the the subspace spanned by the vector that is $1$ in the coordinates of $V$ and $0$ elsewhere.

\paragraph{Effective resistance.}
For undirected $G = (V, E, \ww)$, the effective resistance (ER) of $u, v \in V$ is
$\ER_G(u, v) \defeq \bb_{(u, v)}^\top \LL_G^\dagger \bb_{(u, v)}$. We also define $\ER_G(e)$ for $e = (u, v) \in E$ by $\ER_G(e) \defeq \ER_G(u, v)$. 

\paragraph{Graph linear algebra.} 

In Appendix~\ref{app:deferred} we prove the following facts about graph matrices.

\begin{restatable}{fact}{restatecirceq} \label{lemma:circeq}
Let $\BB = \HH - \TT$ be the edge-vertex incidence matrix of a graph, let $\xx$ be a circulation in the graph (i.e.\ $\BB^\top \xx = \vzero$), and let $\XX \defeq \diag{\xx}$. Then $\HH^\top \XX \HH = \TT^\top \XX\TT$ and $\BB^\top \XX \HH = -\TT^\top \XX \BB$.
\end{restatable}

\begin{restatable}{fact}{restatedirundir} \label{fact:dirclose_undirclose}
Suppose $\vG=(V,E,\ww_{\vG}), \vH=(V,F,\ww_{\vH})$ share the same vertex set and $G \defeq \und(\vG)$, $H \defeq \und(\vH)$.
If $\BB_{\vG}^\top \ww_{\vG} = \BB_{\vH}^\top \ww_{\vH}$, then
$\normsop{\LL_{G}^{\frac \dagger 2}(\LL_G - \LL_H)\LL_{G}^{\frac \dagger 2}} \le 2 \normsop{\LL_{G}^{\frac \dagger 2} (\vLL_{\vG} - \vLL_{\vH})\LL_{G}^{\frac \dagger 2}}$.
\end{restatable}

\begin{restatable}{fact}{restatelinalgthree}\label{fact:opnorm_precondition}
Suppose $G, H$ are connected graphs on the same vertex set $V$, and $\normsop{\LL_{G}^{\dagger/ 2}\Par{\LL_G - \LL_H}\LL_{G}^{\dagger/ 2}} \le \eps$. Then for any $\MM \in \R^{V \times V}$, we have $\normsop{\LL_{G}^{\dagger/ 2} \MM \LL_{G}^{\dagger/ 2}} \le (1 + \eps) \normsop{\LL_{H}^{\dagger/ 2} \MM \LL_{H}^{\dagger/ 2}}$.
\end{restatable}

\section{Technical overview}\label{sec:overview}

In this section, we overview our strategy for preserving degree balance in efficient directed sparsification primitives, in greater detail than in Section~\ref{ssec:summary}. We first review a motivating construction for undirected sparsifiers via randomly signed edge weight updates. Then we introduce our extension of this construction to the Eulerian setting, based on electric projections of edge Laplacians. 

To bound the spectral error incurred by random reweightings in the Eulerian setting, we then describe a new asymmetric matrix variance bound under certain bounds on the effective resistance diameter and weight ratio of the edges under consideration (Lemma~\ref{lemma:variance}). This Lemma~\ref{lemma:variance} is the key technical tool enabling our results, proven in Section~\ref{sec:variance}. 

We then describe an \emph{effective resistance decomposition} (Definition~\ref{def:er_partition}) subroutine we introduce in Section~\ref{sec:partition}, used to guarantee the aforementioned weight and effective resistance bounds hold in our sparsification procedures. Finally, we explain how each of our algorithms (in proving Theorems~\ref{thm:fastsparsify} and~\ref{thm:existential})
and their applications in Sections~\ref{sec:existence},~\ref{sec:algos},~\ref{sec:apps}, and~\ref{sec:sketch} 
build upon these common primitives.

\paragraph{Sparsification from random signings.} To motivate our approach, consider the following conceptual framework in the simpler setting of undirected sparsification. (Variants of this framework have appeared in the recent literature \cite{ChuGPSSW18, ReisR20, JambulapatiRT23}.) Starting from  undirected graph $G = (V, E, \ww)$ with $n$ vertices and $m$ edges, we initialize $\ww_0 \gets \ww$ and in each iteration $t$, let 
\begin{equation}\label{eq:undir_reweight_intro}\ww_{t + 1} \gets \ww_t \circ (\vone_E + \eta \ss_t),\end{equation}
where $\ss_t \in \{\pm 1\}^E$ has independent Rademacher entries and $\eta \in (0, 1]$. Intuitively, the update \eqref{eq:undir_reweight_intro} drives edge weights rapidly to zero, as it induces an exponential negative drift on each weight:
\begin{equation}\label{eq:negative_weight_drift}\E \log\Par{\frac{[\ww_{t + 1}]_e}{[\ww_t]_e}} = \E \log(1 + \eta[\ss_t]_e) \approx -\eta^2.\end{equation}
This phenomenon is most obvious when $\eta = 1$ (which suffices for undirected sparsification), as a constant fraction of edges are immediately zeroed out in each iteration, but \eqref{eq:negative_weight_drift} quantifies this for general $\eta$.
Next, consider the spectral approximation error induced by the first step ($t = 0$), where we denote $G_0 \defeq G = (V, E, \ww_0)$ and $G_1 \defeq (V, E, \ww_1)$, and let $\eta = 1$. By standard matrix concentration inequalities on Rademacher signings (see, e.g., Lemma~\ref{lem:matrix_azuma}), w.h.p.,
\begin{equation}\label{eq:undir_error_onestep}
	\begin{aligned}\normop{\LL_{G_0}^{\frac \dagger 2}\Par{\LL_{G_1} - \LL_{G_0}} \LL_{G_0}^{\frac \dagger 2}} &= \normop{\LL_{G_0}^{\frac \dagger 2}\BB_{G}^\top\Par{\WW_1 - \WW_0}\BB_G\LL_{G_0}^{\frac \dagger 2}} \\
		&= \normop{\sum_{e \in E} \ss_e \AA_e} \lesssim \sqrt{\normop{\sum_{e \in E} \AA_e^2}}, \text{ where } \AA_e \defeq \ww_e\LL_G^{\frac \dagger 2} \bb_e\bb_e^\top \LL_G^{\frac \dagger 2}.
	\end{aligned}
\end{equation}
This argument suggests that it is crucial to control the following matrix variance statistic, $\sigma^2 \defeq \normsop{\sum_{e \in E} \AA_e^2}$, 
as we incur spectral approximation error $\approx \sigma$. It is straightforward to see that, letting $\rho_{\max} \defeq \max_{e \in E} \ww_e \bb_e^\top \LL_G^{\dagger}\bb_e = \max_{e \in E} \ww_e \ER_G(e)$
be the maximum weighted effective resistance of any edge in $G$, we have
\begin{equation}\label{eq:er_bound_sigma}
	\sum_{e \in E} \AA_e^2 = \sum_{e \in E} \ww_e \LL_G^{\frac \dagger 2} \bb_e\Par{\ww_e \bb_e^\top \LL_G^{\dagger} \bb_e} \bb_e^\top \LL_G^{\frac \dagger 2} \preceq \rho_{\max} \LL_G^{\frac \dagger 2}\Par{\sum_{e \in E} \ww_e  \bb_e \bb_e^\top} \LL_G^{\frac \dagger 2} \preceq \rho_{\max} \II_V.
\end{equation}
By zeroing entries of $\ss$ corresponding to the largest half of $\ww_e \ER_G(e)$ values, we can ensure $\rho_{\max} = O(\frac n m)$, since $\sum_{e \in E} \ww_e \ER_G(e) \le n$. Hence, \eqref{eq:undir_error_onestep} shows the spectral approximation error is $\lesssim \sqrt{n/m}$. Since the edge sparsity $m$ decreases by a constant factor in each iteration $t$ when $\eta = 1$, this induces a geometric sequence in the spectral approximation quality terminating at $\approx \eps$ when $\nnz(\ww_t) \approx n \eps^{-2}$, as desired. We remark that Rademacher signings are not the only way to instantiate this scheme; indeed, \cite{ReisR20, JambulapatiRT23} show how to use discrepancy-theoretic tools to choose the update \eqref{eq:undir_reweight_intro} in a way which does not lose logarithmic factors in the spectral error bound.

\paragraph{Asymmetric variance statistics and ER decomposition.} The aforementioned framework for undirected sparsification runs into immediate difficulties in the context of Eulerian sparsification (Definition~\ref{def:eulerian_sparsifier}), as it does not preserve degree balances. Previous Eulerian sparsification methods sidestepped this obstacle by either fixing degrees after sampling and incurring errors (e.g., via expander decomposition) or coordinating the sampling in a degree-preserving way (e.g., via short cycle decomposition). We propose an arguably more direct approach to preserving degrees, departing from prior work. Consider Eulerian $\vG_0 \defeq \vG = (V, E, \ww_0 \defeq \ww)$. On iteration $t \ge 0$, let 
\[\PP_t \defeq \II_E - \WW_t \BB_{\vG} \LL_{G_t^2}^{\dagger} \BB_{\vG}^\top \WW_t,\]
where $\LL_{G_t^2}$ is the undirected Laplacian of $G_t^2 \defeq (V, E, \ww_t^2)$, $\ww_t^2$ is entrywise, and $\WW_t \defeq \diag{\ww_t}$. Observe that $\PP_t$ is the orthogonal projection matrix onto the space of degree-preserving reweightings on the graph $G_t$ with weights $\ww_t$, i.e., for all $\xx \in \textup{Im}(\PP_t)$, we have $\BB_{\vG_t}^\top (\ww_t \circ \xx) = \vzero_V$. Our starting point is thus a modification of the reweighting scheme \eqref{eq:undir_reweight_intro}, where the Rademacher vector $\ss_t$ is replaced by $\xx_t \defeq \PP_t \ss_t$, and we choose an appropriate step size $\eta \approx \log^{-1/2}(n)$ to ensure no edge weight falls below $0$. In other words, we simply let
\begin{equation}\label{eq:dir_reweight_intro}
	\ww_{t + 1} \gets \ww_t \circ (\vone_E + \eta \xx_t),\text{ where } \xx_t \gets \PP_t \ss_t.
\end{equation}
Because this reweighting scheme preserves degree imbalance by construction, it remains to analyze two properties of the reweighting. First, how much does the spectral approximation factor in \eqref{eq:eulerian_sparsifier_def} grow in each iteration? Second, does the reweighting significantly decrease the graph sparsity (ideally, after few iterations)? We postpone discussion of the second point until the end of this overview, when we discuss implementations of our framework. Our analysis of weight decay will ultimately carefully quantify the intuition in \eqref{eq:negative_weight_drift} with an appropriate step size schedule.

Regarding the first point, matrix Rademacher inequalities (extending \eqref{eq:undir_error_onestep} to the asymmetric setting) show that the spectral error in the first step $t = 0$ is controlled by 
\begin{equation}\label{eq:directed_sigma_intro}
	\begin{gathered}
		\sigma^2 \defeq \max\Par{\normop{\sum_{e \in E} \Atil_e \Atil_e^\top},\; \normop{\sum_{e \in E} \Atil_e \Atil_e^\top}}, \\
		\text{where } \Atil_e \defeq \sum_{f \in E} \PP_{fe} \AA_f \text{ and } \AA_e \defeq \ww_e \LL_{G}^{\frac \dagger 2} \bb_e \ee_{h(e)}^\top \LL_G^{\frac \dagger 2},
	\end{gathered}
\end{equation}
and we abbreviate $G = \und(\vG)$ and $\PP_0 = \PP$ for short. To briefly explain the formula \eqref{eq:directed_sigma_intro}, note that analogously to \eqref{eq:undir_error_onestep}, the matrix $\AA_e$ is defined so that the one-step spectral error when reweighting by Rademacher $\ss$ (in the sense of \eqref{eq:dir_reweight_intro}) is precisely $\normsop{\sum_{e \in E} \ss_e \AA_e}$. Correspondingly, the matrix $\Atil_e$ is defined to capture the correct error statistic after first applying $\PP$ to $\ss$. 

A primary technical contribution of our work is quantifying a sufficient condition under which the asymmetric variance statistic \eqref{eq:directed_sigma_intro} is bounded, stated formally as Lemma~\ref{lemma:variance}. Recall that in the undirected setting, \eqref{eq:er_bound_sigma} bounds $\sigma^2$ in terms of the maximum weighted ER of the edges we choose to reweight. Similar logic suggests that the Eulerian variance statistic in \eqref{eq:directed_sigma_intro} is small if $\ee_{v}^\top \LL_G^{\dagger} \ee_{v}$ is bounded for each vertex $v \in V$, i.e., the diagonal entries of $\LL_G^{\dagger}$ are small. In the undirected, unweighted case, $\ee_{v}^\top \LL_G^{\dagger} \ee_{v}$ is bounded for all $v \in V$ if $G$ has small \emph{effective resistance diameter}, i.e., $\ER_G(u, v) \defeq \bb_{(u, v)}^\top \LL_G^\dagger \bb_{(u, v)}$ is small for all $(u, v) \in V \times V$ (Lemma~\ref{lemma:reffentry}).

This intuition neglects at least three factors: it only captures the variance matrix $\sum_{e \in E} \AA_e \AA_e^\top$ (rather than $\sum_{e \in E} \AA_e^\top \AA_e$), it is based on the matrices $\AA_e$ (rather than $\Atil_e$), and it ignores the effects of weights.
Our bound in Lemma~\ref{lemma:variance} tackles all three of these factors by using graph-theoretic construction we introduce, called an ER decomposition (Definition~\ref{def:er_partition}). Again considering only the first step for simplicity, we prove that if $\vH = (V, F, \ww_F)$ is a subgraph of $\vG$ whose vertices all lie in $U$, the quantities $\sum_{f \in F} \Atil_f^\top \Atil_f$ and $\sum_{f \in F} \Atil_f \Atil_f^\top$ are both bounded (in the Loewner ordering) by
\[\rho_{\max}(F) \cdot \LL_G^{\frac \dagger 2} \LL_H \LL_G^{\frac \dagger 2},
\text{ where } \rho_{\max}(F) \defeq \Par{\max_{f \in F} \ww_f} \cdot \Par{\max_{u, v \in U} \ER_G(u, v)},\text{ and } H \defeq \und(\vH). \]
This suggests that if we can isolate a cluster of edges $F$ on a vertex set $U$, such that all edges in $F$ have roughly even edge weight, and such that $U$ has bounded effective resistance diameter through $G$ (inversely proportional to the weights in $F$), we can pay for the contribution of all $\Atil_f$ for $f \in F$ to the variance statistic in \eqref{eq:directed_sigma_intro}. We accordingly define ER decompositions to decompose $E$ into such clusters $\{F_k\}_{k \in [K]}$, each with bounded $\rho_{\max}(F_k) \approx \frac n m$.

\paragraph{Our ER decomposition scheme.} We take a brief digression to answer: how do we find such an edge-disjoint decomposition $\{F_k\}_{k \in [K]}$, each with bounded $\rho_{\max}(F_k)$? In fact, such a decomposition is immediately implied by the related ER decomposition of \cite{AlevALG18}, save two issues. The ER decomposition of \cite{AlevALG18} only guarantees that a constant fraction of edges \emph{by total weight} are cut, as opposed to by edge count (which our recursion can tolerate). The more pressing issue is that the \cite{AlevALG18} algorithm uses $\Omega(mn)$ time, necessitating design of a faster decomposition scheme.

In Section~\ref{sec:partition}, we provide a simple near-linear time decomposition scheme which makes use of the well-known fact that effective resistances in a graph form a metric. We first partition the undirected graph $G$ in question into subgraphs $\{G^j\}_{j_{\min} \le j \le j_{\max}}$ for appropriate $j_{\max} - j_{\min} + 1 = O(\log U)$, where $G^j$ consists of edges with weight between $2^j$ and $2^{j+1}$, and $U$ is the multiplicative range of edge weights. In each $G^j$, it suffices to partition the vertices to induce subgraphs $\{G^j_i\}_{i \in [K_j]}$, each with ER diameter $\approx \frac n m \cdot 2^{-j}$, and such that few edges are cut. We accomplish this by first providing constant-factor estimates to all edge effective resistances using standard sketching tools (Lemma~\ref{lem:approx_er}). 
Within each subgraph $G^j$, we induce a shortest path metric based on our ER overestimates, and then apply classic region-growing techniques~\cite{GargVY96} to partition the subgraphs into pieces of bounded shortest path diameter without cutting too many edges.

\paragraph{Implementations of our framework.} Finally, we briefly outline how Theorems~\ref{thm:fastsparsify},~\ref{thm:existential}, and~\ref{thm:fastsketch_undir}
 follow from the frameworks we outlined. Our Eulerian sparsification algorithms (for establishing Theorems~\ref{thm:fastsparsify} and~\ref{thm:existential}) simply interleave computation of an ER decomposition on the current graph, with a small number of reweightings roughly of the form \eqref{eq:dir_reweight_intro}. For our nearly-linear time algorithm in Theorem~\ref{thm:fastsparsify}, in each reweighting \eqref{eq:dir_reweight_intro}, we zero out the half of entries of $\ss_t$ which are cut by the ER decomposition, and additionally enforce a linear constraint that the total weight is preserved. We show that by making the intuition \eqref{eq:dir_reweight_intro} rigorous, after polylogarithmically many reweightings, a constant fraction of edge weights have decreased by a polynomial factor, which is enough to explicitly delete them from the graph and (after fixing degrees by routing through a spanning tree) incur small spectral error. This lets us recurse and obtain the same geometric sequence behavior on our accumulated spectral error bound as in the undirected setting.

Our proof of Theorem~\ref{thm:existential} applies carefully-coordinated reweighting vectors $\xx_t$ which yield smaller spectral error than na\"ive random signing. We choose these vectors $\xx_t$ based on recent progress towards the matrix Spencer conjecture (a well-known open problem in discrepancy theory) due to \cite{BansalJM23}. Specifically, \cite{BansalJM23} (along with earlier works, e.g., \cite{Rothvoss17, ReisR23}) provide tools which construct ``partial colorings'' $\xx_t$ such that $[\xx_t]_e = -1$ for a constant fraction of $e \in E$, and
\[\normop{\sum_{e \in E} [\xx_t]_e \AA_e}\]
is smaller than what matrix Rademacher inequalities would predict for random $\xx_t$ (based on the matrix variance statistic). Applying these higher-quality reweightings $\xx_t$ in each iteration through \eqref{eq:dir_reweight_intro} (with $\eta = 1$) then directly decreases the edge sparsity by a constant factor in each iteration, allowing for simple control of the spectral error in \eqref{eq:eulerian_sparsifier_def}. This strategy immediately yields Theorem~\ref{thm:existential} upon recursing. As mentioned previously, in Appendix~\ref{sec:conjectures}, we examine natural routes which could further improve upon the sparsity bounds of Theorem~\ref{thm:existential}.

Finally, we show that our subroutines designed for Eulerian sparsification compose well with a framework for obtaining graphical spectral sketches by \cite{ChuGPSSW18}, based on expander decomposition. Specifically, \cite{ChuGPSSW18} (based on similar ideas in \cite{AndoniCKQWZ16, JambulapatiS18}) showed that spectral bounds between degree matrices and Laplacians which hold in expander graphs yield improved per-vector quadratic form guarantees. We make the simple observation that expander subgraphs with large minimum degree also have bounded effective resistance diameter (Lemma~\ref{lemma:ex_to_er}). Hence, directly using our algorithms for sparsifying pieces of an effective resistance decomposition using electric projections in place of short cycle decompositions (as used in \cite{ChuGPSSW18}) improves state-of-the-art runtimes by $m^{o(1)}$ factors. Our spectral sketch algorithm is flexible enough to extend straightforwardly to the Eulerian setting (following Definition~\ref{def:sketch_directed}), as described in Theorem~\ref{thm:fastsketch}.

\section{Effective resistance decomposition}\label{sec:partition}

In this section, we show how to efficiently decompose a weighted, undirected graph into subgraphs with bounded weight ratio, small effective resistance diameter (relative to the edge weights it contains), a limited number of edges cut, and each vertex appearing in a limited number of subgraphs. This procedure will be a key subroutine in our sparsification algorithms, as captured by the variance bound in Lemma~\ref{lemma:variance}. Below in \Cref{def:er_partition} we formally define this type of decomposition guarantee and then in \Cref{prop:er_partition} we provide our main result on computing said decompositions.

\begin{definition}[ER decomposition]\label{def:er_partition}
We call $\{G_i\}_{i \in [I]}$ a \emph{$(\rho, r, J)$-effective resistance (ER) decomposition} if $\{G_i\}_{i \in [I]}$ are edge-disjoint subgraphs of $G = (V, E, \ww)$, and the following hold.
\begin{enumerate}
    \item \label{item:effres:partition:weight} \emph{Bounded weight ratio}: For all $i \in [I]$, $\frac{\max_{e \in E(G_i)}\ww_e}{\min_{e \in E(G_i)}\ww_e} \le r$.
    \item \label{item:effres:partition:diameter} \emph{Effective resistance diameter}: For all $i \in [I]$, $(\max_{e \in E(G_i)}\ww_e) \cdot (\max_{u, v \in V(G_i)} \ER_G(u, v)) \le \rho$.
    \item \label{item:effres:partition:cut} \emph{Edges cut}: $|E(G) \setminus ( \bigcup_{i \in [I]} E(G_i) )| \leq \frac m 2$.
    \item \label{item:effres:partition:vertex} \emph{Vertex coverage}: Every vertex $v \in V(G)$ appears in at most $J$ of the subgraphs.
\end{enumerate}
\end{definition}

\begin{proposition}\label{prop:er_partition}
There is an algorithm, $\ERPalgo(G, r, \delta)$, which given any $G = (V,E,\ww)$ with $n = |V|$, $m = |E|$, $\frac{\max_{e \in E} \ww_e}{\min_{e \in E} \ww_e} \le W$ and $r \geq 1$, $\delta \in (0, 1)$, computes a 
\[\Par{\frac{8rn\log(n + 1)}{m},\, r,\, \log_r(W) + 3}\text{-ER decomposition of } G,\]
with probability $\ge 1 - \delta$ in time\footnote{The $O(n \log n)$ term arises from the use of Fibonacci heaps to compute shortest paths in undirected graphs in \Cref{prop:region}. There are results that have since obtained faster algorithms for computing shortest paths in undirected graphs \cite{Thorup99, DuanMSY23}. Moreover, the shortest paths do not necessarily need to be computed exactly, so it is possible that this factor could be improved as it has been in other region growing settings \cite{MillerPX13, AbrahamN19}. However, since this is not a bottleneck in the runtimes of our main results, we make no attempt to improve it here.
}
\[\bO\Par{m \log\Par{\frac n \delta} + n \log (n) \log_r(W)}.\]
\end{proposition}

In the remainder of this section, we prove \Cref{prop:er_partition}. The algorithm consists of two components. First, we use standard randomized algorithms (\Cref{lem:effres_over}) to efficiently compute an ER overestimate for the graph edges (\Cref{def:effres_over}). Then, we apply a standard result on region growing (\Cref{prop:region}) from \cite{GargVY96} to efficiently partition the edges within one weight range (\Cref{lem:region_applied}). Applying this decomposition scheme at every weight scale to the graph with edge lengths given by the effective resistance overestimates then yields the result. 
Interestingly, the only use of randomization in this algorithm is in computing overestimates of effective resistances and if a sufficiently efficient deterministic subroutine for this was developed, substituting this subroutine into our algorithm would would obtain a deterministic counterpart of \Cref{prop:er_partition}.

\begin{definition}[Effective resistance overestimate]
\label{def:effres_over}
Given $G = (V, E, \ww)$ with $n = |V|$, we call $\trr \in \R^E$ an $\alpha$\emph{-approximate effective resistance (ER) overestimate} if 
\[\ww^\top \trr \leq \alpha n \text{ and } \tilde{\rr}_{e} \geq \ER_G(e) \text{ for all } e \in E.\]
\end{definition}

To efficiently compute ER overestimates for use in our decomposition algorithms, we rely on near-linear time undirected Laplacian linear system solvers. To begin, we first provide a statement of the current fastest Laplacian linear system solver in the literature.

\begin{proposition}[Theorem 1.6, \cite{JambulapatiS21}]\label{prop:js21}
	Let $\LL_G$ be the Laplacian of $G = (V, E, \ww)$. There is an algorithm which takes $\LL_G$, $\bb \in \R^V$, and $\delta, \xi \in (0, 1)$, and outputs $\xx$ such that with probability $\ge 1 - \delta$, $\xx$ is an $\xi$-approximate solution to $\LL_G \xx = \bb$, i.e.,
	\[\norm{\xx - \LL_G^\dagger \bb}_{\LL_G} \le \xi \norm{\LL_G^\dagger \bb}_{\LL_G},\]
	in time $\bO(|E| \cdot \log \frac{1}{\delta\xi})$. Moreover, the algorithm returns $\xx = \MM \bb$ where $\MM$ is a random linear operator constructed independently of $\bb$, such that the above guarantee holds with $1 - \delta$ for all $\bb$.
\end{proposition}
The runtime guarantee of the above proposition follows from Theorem~1.6 of \cite{JambulapatiS21}. We now briefly justify the second clause in Proposition~\ref{prop:js21}, i.e.\ that the Laplacian solver is a randomized linear function of $\bb$, as it is not explicitly stated in \cite{JambulapatiS21}. Theorem~1.6 follows by combining an algorithm which constructs low-stretch subgraphs with a recursive preconditioning framework (Algorithm 12). Algorithm 12 returns the result of an error-robust accelerated gradient descent procedure $\mathsf{PreconNoisyAGD}$, which only applies linear transformations and a procedure $\mathsf{RichardsonSolver}$, to $\bb$. In turn, $\mathsf{RichardsonSolver}$ performs only linear transformations and another procedure $\mathsf{PreconRichardson}$ to its input. Finally, $\mathsf{PreconRichardson}$ applies linear transformations and Algorithm 12 to its input: in addition, these calls to Algorithm 12 operate on strictly smaller problems. Thus, if we assume that these inner calls to Algorithm 12 perform a linear transformation of $\bb$, the outer call is also a linear transformation: the last claim in Proposition~\ref{prop:js21} follows. 

Proposition~\ref{prop:js21} combined with a Johnson-Lindenstrauss based sketching approach from \cite{SpielmanS08} shows we can efficiently approximate a set of effective resistances to constant multiplicative error, which we summarize in the following. We remark that the runtime in \cite{SpielmanS08} is larger than in Lemma~\ref{lem:approx_er}; our improvement stems from replacing the solver used there with Proposition~\ref{prop:js21}.

\begin{lemma}[Theorem 2, \cite{SpielmanS08}]\label{lem:approx_er}
	Let $\delta \in (0, 1)$, let $\LL_G$ be the Laplacian of $G = (V, E, \ww)$, and let $S \subseteq V \times V$. There is an algorithm, \AER{$G, S, \delta$}, which runs in time $\bO((|E| + |S|)\log(\frac{|S|}{\delta}) )$ and outputs $\rr = \{\rr_{(u, v)}\}_{(u, v) \in S}$ satisfying with probability $\ge 1 - \delta$,
	\[\frac 2 3 \ER_G(u, v) \le \rr_{(u, v)} \le \frac 4 3 \ER_G(u, v), \text{ for all } (u, v) \in S.\]
\end{lemma}
\begin{proof}
	Consider the following algorithm for approximating $\ER_G(u, v)$ for some $(u, v) \in S$. We output the median of $K = \Theta(\log\frac{|S|}{\delta})$ independent evaluations of
	\begin{equation}\label{eq:one_estimate}\norm{\QQ \WW_G^{\half} \BB_G \MM (\ee_v - \ee_u)}_2^2,\end{equation}
	for $\QQ \in \RR^{\Theta(1) \times |E|}$ filled with random scaled Gaussian entries, and where $\MM$ is the random linear operator given by the approximate solver in Proposition~\ref{prop:js21} with a sufficiently small constant $\xi$. We claim that \eqref{eq:one_estimate} lies in the range $[\frac 2 3 \ER_G(u, v), \frac 4 3 \ER_G(u, v)]$ with probability $\frac 2 3$. By standard Johnson-Lindenstrauss guarantees (see, e.g., the proof of Theorem 2 in \cite{SpielmanS08}), it suffices to prove that with probability $\frac 5 6$, letting $\MM$ be the resulting linear operator from Proposition~\ref{prop:js21},
	\[\Abs{\norm{\WW_G^{\half}\BB_G \MM \bb}_2^2 - \norm{\WW_G^{\half} \BB_G \LL_G^\dagger \bb}_2^2} \le \frac 1 4 \norm{\WW_G^{\half} \BB_G \LL_G^\dagger \bb}_2^2.\]
	To this end, using $0.9\norm{\uu}_2^2 - 11\norm{\vv}_2^2 \le \norm{\uu + \vv}_2^2 \le 1.1\norm{\uu}_2^2 + 11\norm{\vv}_2^2$, we have
	\begin{align*}
		\norm{\WW_G^{\half}\BB_G \MM\bb}_2^2 &\le 1.1\norm{\WW_G^{\half} \BB_G \LL_G^\dagger \bb}_2^2 + 11\norm{\WW_G^{\half}\BB_G \Par{\LL_G^\dagger \bb - \MM \bb}}_2^2, \\
		\norm{\WW_G^{\half}\BB_G \MM\bb}_2^2 &\ge 0.9\norm{\WW_G^{\half} \BB_G \LL_G^\dagger \bb}_2^2 - 11\norm{\WW_G^{\half}\BB_G \Par{\LL_G^\dagger \bb - \MM \bb}}_2^2,
	\end{align*}
	so choosing $\delta = \frac 1 6$ and $\xi = \frac{1}{100}$ in Proposition~\ref{prop:js21} yields the desired claim on each individual evaluation of \eqref{eq:one_estimate}. 
	Thus, by Chernoff bounds the median estimate will lie in the specified range with probability $\ge 1 - \frac{\delta}{|S|}$, yielding correctness after a union bound over all of $S$. 
	
	We now discuss how to implement the above algorithm within the stated runtime. For each independent run $k \in [K]$, we first precompute $\QQ \WW_G^{1/2} \BB_G$ in the given time, and apply $\MM$ from Proposition~\ref{prop:js21} to each of the $\Theta(1)$ rows of this matrix. Notably, we can reuse the same random seed in the solver of \cite{JambulapatiS21} so that the random linear operator $\MM$ provided by Proposition~\ref{prop:js21} is the same for all rows of  $\QQ \WW_G^{1/2}  \BB_G$. The random linear function $\MM$ is constructed obliviously to the choice of $\QQ$, so $\QQ$ is independent of these calls and Johnson-Lindenstrauss applies. 
	Each evaluation of \eqref{eq:one_estimate} takes constant time, which we need to repeat $|S|K$ times in total. 
\end{proof}

Our ER overestimate computations then follow from an immediate application of Lemma~\ref{lem:approx_er}.

\begin{lemma}
\label{lem:effres_over}
There is a randomized algorithm, that given any $G = (V, E, \ww)$ with $n = |V|$, $m = |E|$, computes a $2$-approximate ER overestimate with probability $\ge 1- \delta$ in $\bO(m \log \frac n \delta)$ time.
\end{lemma}
\begin{proof}
Consider applying \Cref{lem:approx_er} with $S = E$ and the specified $\delta$. In $\bO(m \log \frac n \delta)$ time this procedure computes $\rr \in \R^E$ such that with probability $\ge 1 - \delta$,
\[
\frac 2 3 \ER_G(u, v) \le \rr_{(u, v)} \le \frac 4 3 \ER_G(e), \text{ for all } e \in E.
\]
Our algorithm simply computes this $\rr$ and then outputs $\tilde{\rr} = \frac 3 2 \rr$. The output $\tilde{\rr}$
has the desired properties as $\tilde{\rr}_{e} \geq \ER_G(e)$ for all $e \in E$ and
\[
\sum_{e \in E} \ww_e \tilde{\rr}_{e} 
\leq \Par{\frac 4 3 \cdot \frac 3 2} \sum_{e \in E} \ww_e \cdot \ER_G(e)
\leq 2 n ,
\]
as $\sum_{e \in E} \ww_e \ER_G(e)$ is $n - c$ where $c$ is the number of connected components in $G$.
\end{proof}

Next, we provide a key subroutine from prior work used in our decomposition.

\begin{proposition}[Region growing, \cite{GargVY96}, Section 4]
\label{prop:region}
There is a deterministic algorithm that given $G = (V, E, \ww)$ with $n = |V|$, $m = |E|$, edge lengths $\ll \in \R^{E}_{> 0}$ and $d > 0$, in $O(m + n \log n)$-time outputs a partition $\{S_k\}_{k \in [K]}$ of $V$, each with diameter $\leq 2 d \log(n + 1)$ with respect to $\ll$, and with 
\[d \cdot \sum_{e \in \partial(\{S_k\}_{k \in [K]})} \ww_e  \leq 2\ww^\top \ll,\]
where $\partial(\{S_k\}_{k \in [K]})$ is the set of edges $(u, v) \in E$ with $u \in S_i$, $v \in S_j$ and $i \neq j$.
\end{proposition}

By applying Proposition~\ref{prop:region} instantiated with appropriate edge lengths, we have the following.

\begin{lemma}
\label{lem:region_applied}
There is a deterministic algorithm that given $G = (V, E, \ww)$ with $n = |V|$, $m = |E|$, edge lengths $\ll \in \R^{E}_{> 0}$, and parameters $v, \alpha > 0$ and $r > 1$, in $O(m + n \log n)$-time outputs
vertex-disjoint subgraphs $\{G_k\}_{k \in [K]}$ such that the following hold.
\begin{enumerate}
    \item \label{item:partition:edges} $\bigcup_{k \in [K]} E(G_k) \subseteq F$ for $F \defeq \{e \in E \mid \ww_e \in (\frac{v}{r},v] \}$.
    \item \label{item:partition:diameter} For all $k \in [K]$, the diameter of $G_k$ with respect to $\ll$ is at most $\frac \alpha {\max_{e \in E(G_k)}\ww_e}$.
    \item \label{item:partition:cut} $|F \setminus \{ \bigcup_{k \in [K]} E(G_k) \}| \leq \frac{4 r \ln(n + 1)}{\alpha} \cdot \sum_{e \in F} \ww_e \ll_e$.
\end{enumerate}
\end{lemma}

\begin{proof} 
Let $\bww_e = \ww_e$ for all $e \in F$ and $\bww_e = 0$ for all $e \in E \setminus F$. We apply \Cref{prop:region} to $G$ with $\ww \gets \bww$  and $d \gets \frac{\alpha}{2v\log(n + 1)}$ to obtain $\{S_k\}_{k \in [K]}$. Define $\{G_k\}_{k \in [K]}$ so that $V(G_k) = S_k$ and $E(G_k)$ are the edges of $F$ with both endpoints in $S_k$, with the same weight as in $G$. 

We prove that the $\{G_k\}_{k \in [K]}$ satisfy Items~\ref{item:partition:edges}, \ref{item:partition:diameter}, and \ref{item:partition:cut}. \Cref{item:partition:edges} follows directly by construction. Next, \Cref{prop:region} implies that the diameter of each $G_k$ with respect to $\ell$ is at most $\frac \alpha v$. \Cref{item:partition:diameter} then follows as $\max_{e \in E(G_k)}\ww_e \leq v$. For \Cref{item:partition:cut}, note that \Cref{prop:region} implies that 
\[
\left(\frac{\alpha}{2 v \ln(n + 1)}\right)\sum_{e \in E \setminus (\bigcup_{k \in [K]} E(G_k))} \bww_e \cdot  \leq 2 \bww^\top \ll.
\]
\Cref{item:partition:cut} then follows from combining the above, $ \bww^\top \ll = \sum_{e \in F} \ww_e \ll_e$, and
\[
\left| F \setminus \left\{ \bigcup_{k \in [K]} E(G_k) \right\} \right|
= \sum_{\substack{e \in E \setminus (\bigcup_{k \in [K]} E(G_i)) \\ \bww_e > 0}} \frac{v}{v}
< \sum_{e \in E \setminus (\bigcup_{k \in [K]} E(G_k))} \frac{r \cdot \bww_e}{v}.
\]
\end{proof}

\begin{proof}[Proof of \Cref{prop:er_partition}]
Consider the following algorithm. First, apply \Cref{lem:effres_over} to compute a 2-approximate effective resistance overestimate with probability $\ge 1 - \delta$, and save these as $\ll \in \R^E_{> 0}$. We then apply \Cref{lem:region_applied} for all integers $j \in [j_{\min}, j_{\max}]$ where 
$j_{\min} = \lfloor \log_r(\min_{e \in E}\ww_e) \rfloor$ and $j_{\max} = \lceil\log_r(\max_{e \in E}\ww_e)\rceil$ with 
\[v \gets v_j \defeq r^j,\; \alpha \gets \frac{16rn\log(n + 1)}{m},\text{ and } r \gets r.\]
For all $j \in [j_{\min}, j_{\max}]$ we let $\{G_i^j\}_{i \in [K_j]}$ be the vertex-disjoint subgraphs output by \Cref{lem:region_applied} and we let $F_j$ be the value of $F$ for this application of \Cref{lem:region_applied}. This algorithm has the desired runtime as applying \Cref{lem:effres_over} takes time $\bO(m \log \frac n \delta)$ and each application of \Cref{lem:region_applied} takes time $O(|E(G_k)| + n \log n)$. Note that the sum of all the $O(|E(G_k)|)$ terms only contributes a single $O(m)$ to the runtime. Additionally, the number of distinct $j$ is
\begin{equation}
\label{eq:jcount}
j_{\max} - j_{\min} + 1 \leq \log_r \left(\max_{e \in E(G_i)}\ww_e\right) + 1 - \left(\log_r\left(\min_{e \in E(G_i)}\ww_e\right) - 1\right) + 1
= \log_r(W) + 3\,.
\end{equation}
The runtime follows and it remains only to show that the output $\{G_i^j\}_{j_{\min} \le j \le j_{\max}, i \in [K_j]}$ have the desired properties provided that the $\ll$ were indeed a $2$-approximate ER overestimate.

\paragraph{Bounded weight ratio (\Cref{item:effres:partition:weight}).} This follows directly by construction from \Cref{lem:region_applied}.

\paragraph{Effective resistance diameter (\Cref{item:effres:partition:diameter}).} By \Cref{lem:region_applied}, \Cref{item:partition:diameter} we know that for any $G_i^j$ it is the case that the diameter of $G_i^j$ with respect to $\ll$ is at most $\alpha (\max_{e \in E(G_i^j)}\ww_e)^{-1}$. Consequently, for each $u, v \in V(G_i^j)$ it is the case that there is a path of edges whose sum of lengths is at most $\alpha (\max_{e \in E(G_i^j)}\ww_e)^{-1}$. Each of these lengths is at least the effective resistance of the associated edge. Since effective resistances form a metric, by triangle inequality this means 
\[\max_{u, v \in V(G_i^j)} \ER_G(u,v) \leq \frac{\alpha}{\max_{e \in E(G_i^j)}\ww_e}\]
and \Cref{item:effres:partition:diameter} follows by the setting of $\alpha$.

\paragraph{Edges cut (\Cref{item:effres:partition:cut}).} Note that by our choice of $v_j$ and \Cref{lem:region_applied}, the $\{F_j\}_{}$ partition $E$. Since $E(G_i^j) \subseteq F_j$ for all $i \in [K_j]$ and $j \in [j_{\min},j_{\max}]$ we have that
\begin{align*}
\left|E(G) \setminus \Brace{ \bigcup_{j_{\min} \le j \le j_{\max}, i\in [K_j]} E(G_i^j) }\right|
&=
\sum_{j_{\min} \le j \le j_{\max}}
\left|E(G) \setminus \Brace{ \bigcup_{i\in [K_j]} E(G_i^j) }
\right|\\
&\leq \sum_{j_{\min} \le j \le j_{\max}}\Par{ 
\frac{4r\ln(n+1)}{\alpha} \sum_{e \in F_j} \ww_e \ll_e}
= \frac{m}{4n} \sum_{e \in E} \ww_e \ll_e
\end{align*}
where we applied \Cref{lem:region_applied}, \Cref{item:partition:cut} in the inequality. Since $\sum_{e \in E} \ww_e \ll_e \leq 2n$ by the definition of a $2$-approximate effective resistance overestimate, the result follows.

\paragraph{Vertex coverage (\Cref{item:effres:partition:vertex}).} Each collection of $\{G_i^j\}_{i \in [K_j]}$ for fixed $j \in [j_{\min}, j_{\max}]$ is vertex-disjoint by \Cref{lem:region_applied}. Consequently, each vertex $v \in V(G)$ is in at most $j_{\max} - j_{\min} + 1$ subgraphs and the result follows by our earlier bound \eqref{eq:jcount}.
\end{proof}

\section{Variance bounds from effective resistance diameter}\label{sec:variance}

In this section, we provide an operator norm bound on a matrix variance quantity,
used to bound the Gaussian measure of convex bodies induced by operator norm bounds encountered in our sparsification procedures. This variance bound (Lemma~\ref{lemma:variance}) is a key new structural insight which enables our applications in the remainder of the paper. In particular, it shows bounded ER diameter of decomposition pieces can be used to control the spectral error incurred by our reweightings.

We first provide a helpful result which upper bounds matrix variances after a projection operation, by the corresponding variance before the projection.

\begin{lemma} \label{lemma:varianceproj}
    Let $\{\AA_i\}_{i \in [m]} \in \R^{n \times n}$ and let $\PP, \QQ \in \R^{m \times m}$ be orthogonal projection matrices such that $\ker(\QQ) \subseteq \ker(\PP)$. 
    For each $i \in [m]$, let $\Atil_i \defeq \sum_{j \in [m]} \PP_{ji} \AA_j$ and $\Ahat_i \defeq \sum_{j \in [m]} \QQ_{ji} \AA_j$. Then, 
    \[
        \sum_{i \in [m]} \Atil_i \Atil_i^\top \preceq \sum_{i \in [m]} \Ahat_i \Ahat_i^\top.
    \] 
\end{lemma}
\begin{proof}
Throughout this proof, we denote the Kronecker product of matrices $\AA$ and $\BB$ by $\AA \otimes \BB$.  By $\ker(\QQ) \subseteq \ker(\PP)$, we have $\PP \preceq \QQ$. 
Define the $n \times mn$ block-partitioned matrices
\[
\Acal \defeq \begin{pmatrix}
\AA_1 & \AA_2 & \cdots & \AA_m
\end{pmatrix},\; 
\widetilde{\Acal} \defeq \begin{pmatrix}
\Atil_1 & \Atil_2 & \cdots & \Atil_m
\end{pmatrix},\;
\widehat{\Acal} \defeq \begin{pmatrix}
\Ahat_1 & \Ahat_2 & \cdots & \Ahat_m
\end{pmatrix}.
\]
Since $\widetilde{\Acal} = \Acal(\PP \otimes \id_m)$ and $\widehat{\Acal} = \Acal(\QQ \otimes \id_m)$ it now suffices to prove $\widetilde{\Acal}\widetilde{\Acal}^\top \preceq \widehat{\Acal}\widehat{\Acal}^\top$. Note that
\[
(\PP\otimes \id_m)^2 =(\PP^2) \otimes (\id_m)^2 = \PP \otimes \id_m
\preceq
\QQ \otimes \id_m =(\QQ)^2 \otimes (\id_m)^2 = (\QQ \otimes \id_m)^2,
\]
where the equality utilizes $\PP, \QQ$ are orthogonal projection matrices and the inequality holds since since $\PP \preceq \QQ.$
Now utilizing the fact that if $\AA \preceq \BB$ and $\CC$ is any matrix of compatible dimension, then
$ \CC \AA \CC^\top \preceq \CC \BB \CC^\top$ and we get the desired bound that 
\[
\widetilde{\Acal}\widetilde{\Acal}^\top = 
\Acal(\PP \otimes \id_m)^2\Acal^\top 
\preceq
\Acal( \QQ \otimes \id_m)^2\Acal^\top 
= \widehat{\Acal}\widehat{\Acal}^\top\,.
\]
\end{proof}

We also show that effective resistance decomposition pieces have bounded diagonal entries in an appropriate subgraph inverse Laplacian.

\begin{lemma} \label{lemma:reffentry}
For any $G = (V, E, \ww)$, $U \subseteq V$, and $u \in U$,	$\ee_u^\top \PPi_U \LL_G^\dagger \PPi_U \ee_u \leq \max_{a,b \in U} \ER_G(a, b)$.
\end{lemma}
\begin{proof}
First, observe that $\PPi_U \ee_u = \ee_u - \frac 1 {|U|} \vone_U = \frac 1 {|U|}\sum_{v \in U, v \neq u} \bb_{(u,v)}$. The conclusion follows from
    \begin{align*}
    \ee_u^\top \PPi_U \LL_G^\dagger \PPi_U \ee_u 
    &= \frac 1 {|U|^2} \Par{\sum_{v \in U, v \neq u} \bb_{(u, v)}}^\top \LL_G^\dagger \Par{\sum_{v \in U, v \neq u} \bb_{(u, v)}} \\
    &\leq \frac{|U| - 1}{|U|^2} \sum_{v \in U, v \neq u} \bb_{(u, v)}^\top \LL_G^\dagger \bb_{(u, v)} \leq \frac{(|U|-1)^2}{|U|^2} \max_{a,b \in U} \ER_G(a,b),
    \end{align*}
    where the first inequality was the Cauchy-Schwarz inequality.
\end{proof}

We now combine \Cref{lemma:circeq}, \Cref{lemma:varianceproj}, and \Cref{lemma:reffentry} to obtain the main result of this section.

\begin{lemma}\label{lemma:variance}
    Let $\vec{G} = (V,E,\ww)$ and let $\vec{H}$ be a subgraph on vertex set $U
    \subseteq V$. Suppose that for $\rho> 0$, $(\max_{e \in E(\vH)}\ww_e) \cdot (\max_{u, v \in U} \ER_G(u, v)) \leq \rho$. 
    Define 
    \begin{equation}\label{eq:edge_circ_def}
    \begin{aligned} 
    \LL_{H^2} &\defeq \BB_{\vec{H}}^\top \WW_{E(\vec{H})}^2 \BB_{\vec{H}},\\
    \PP_{\vH} &\defeq \id_{E(\vec{H})} - \WW_{E(\vec{H})}\BB_{\vec{H}}\LL_{H^2}^\dagger\BB_{\vec{H}}^\top\WW_{E(\vec{H})},\\
    \Atil_e &\defeq \LL_G^{\frac \dagger 2}\Par{\sum_{f \in E(\vec{H})}[\PP_{\vH}]_{fe} \ww_f \bb_f \ee_{h(f)}^\top}\LL_G^{\frac \dagger 2} \text{ for all } e \in E(\vec{H}),
    \end{aligned}
    \end{equation}
    where $\WW_{E(\vH)}$, $\id_{E(\vH)}$ zero out entries of $\WW_{\vG}$, $\id_{E(\vG)}$ not corresponding to edges in $E(\vH)$. Then, 
    \begin{gather*}
    \sum_{e \in E(\vec{H})} \Atil_e \Atil_e^\top \pleq \rho \cdot \LL_G^{\frac \dagger 2}\LL_{H}\LL_G^{\frac \dagger 2}, \;
    \sum_{e \in E(\vec{H})} \Atil_e^\top  \Atil_e \pleq \rho \cdot \LL_G^{\frac \dagger 2}\LL_{H}\LL_G^{\frac \dagger 2},\\ \text{ where } G \defeq \und(\vG),\; H \defeq \und(\vH).
    \end{gather*}
\end{lemma}
\begin{proof}
    For simplicity, we write $\WW_{\vH} = \WW_{E(\vec{H})}$ and $\BB_{\vec{H}} = \BB_{\vG} \II_{E(\vec{H})}$.
    We first note that $\PP_{\vH}$ is a orthogonal projection matrix, since $\WW_{\vH} \BB_{\vH} \LL_{H^2}^\dagger \BB_{\vH}^\top \WW_{\vH}$ is an orthogonal projection on the restriction to $\vH$. This justifies our notation: the $\Atil_e$ are as in \Cref{lemma:varianceproj}, where $\AA_e \defeq \LL_G^{\dagger/ 2}\ww_e \bb_e \ee_{h(e)}^\top\LL_G^{\dagger/ 2}$. Next, let $\xx_e \defeq [\PP_{\vH}]_{e:}$ and $\XX_e \defeq \diag{\xx_e}$, so $\Atil_e = \LL_G^{\dagger/ 2}\BB_{\vH}^\top \WW_{\vH} \XX_e \HH_{\vH}\LL_G^{\dagger/ 2}$. Since $\PP_{\vH}$ is an orthogonal projection matrix,
    \[\PP_{\vH} \xx_e = \xx_e \implies \WW_{\vH} \BB_{\vH} \LL^\dagger_{H^2} \BB_{\vH}^\top \WW_{\vH} \xx_e = \vzero_{V} \implies \BB_{\vH}^\top \WW_{\vH} \xx_e = \vzero_{V}.\] 
    To see the last implication, note that $\BB_{\vH}^\top \WW_{\vH} \xx_e$ is always orthogonal to the kernel of $\LL_{H^2}=\BB_{\vH}^\top \WW_{\vH} \BB_{\vH}$.
    The last equality then follows by noticing that $\ker(\LL_{H^2}) = \ker(\LL_{H^2}^\dagger)$.
    In other words, $\WW_{\vH} \xx_e$ is a circulation on $\vH$.
    Since $\PPi_U$ is the projection onto the coordinates of $U$ orthogonal to $\vone_U$, 
    by $\ker(\HH_{\vH}) \supseteq \spn(\vone_U) \cup \R^{V \setminus U}$, we further have
    \[\BB_{\vH}^\top \WW_{\vH} \XX_e \HH_{\vH} = \BB_{\vH}^\top \WW_{\vH} \XX_e \HH_{\vH} \PPi_U \implies \Atil_e = \LL_G^{\frac \dagger 2}\BB_{\vH}^\top \WW_{\vH} \XX_e \HH_{\vH} \PPi_U\LL_G^{\frac \dagger 2}.\]
    Applying \Cref{lemma:varianceproj} to $\{\Atil_e \}_{e \in E(\vH)}$ using the characterization in the above display then gives
    \begin{align*}
    \sum_{e \in E(\vec{H})} \Atil_e \Atil_e^\top 
        &\pleq \LL_G^{\frac \dagger 2}\Par{
        \sum_{e \in E(\vH)} w_e^2 \cdot \bb_e \ee_{h(e)}^\top \PPi_U \LL_G^\dagger \PPi_U \ee_{h(e)} \bb_e^\top}\LL_G^{\frac \dagger 2}
        \\
        &\pleq
        \LL_G^{\frac \dagger 2}\Par{\sum_{e \in E(\vH)} \rho w_e \cdot \bb_e \bb_e^\top}\LL_G^{\frac \dagger 2} =
        \rho \cdot \LL_G^{\frac \dagger 2}\LL_{H}\LL_G^{\frac \dagger 2}.
    \end{align*}
    The second inequality follows from \Cref{lemma:reffentry}  and the $w_e \leq
    \max_{e \in E(\vH)} w_e$. This yields the first claim. To see the second, since $\WW_{\vH} \xx_e$ is a circulation, by \Cref{lemma:circeq}, $\Atil_e = -\LL_G^{\dagger/ 2}\TT_{\vH}^\top \WW_{\vH} \XX_e \BB_{\vH}\LL_G^{\dagger/ 2}$. 
    By instead applying \Cref{lemma:varianceproj} to the matrices $\{-\Atil_e^\top\}_{e \in E(\vH)}$ (as $\XX_e \TT_{\vH} \vone_U = \XX_e \HH_{\vH} \vone_U = \xx_e$) and following an analogous derivation, we obtain the desired bound.
\end{proof}

\section{Sparser Eulerian sparsifiers}\label{sec:existence}

In this section, we give the first application of our framework by proving our Eulerian sparsification result obtaining the best-known sparsity bound in \Cref{thm:existential}. This application serves as a warmup for our nearly-linear time sparsification result in Section~\ref{sec:algos}. 

Our approach is to recursively apply \Cref{lemma:variance} on each subgraph component in a ER decomposition (\Cref{prop:er_partition}), 
with known results from the literature on discrepancy theory, to sparsify an Eulerian graph. Specifically, our main tools are a powerful matrix discrepancy Gaussian measure lower bound recently developed by \cite{BansalJM23} (motivated by the matrix Spencer conjecture), and a corresponding partial coloring framework from \cite{Rothvoss17, ReisR23}.

\begin{proposition}[Proof of Lemma 3.1, \cite{BansalJM23}]\label{prop:bjm_measure}
For every constant $c \in (0, \half)$, there is a constant $\ccolor$ such that for any $\{\Atil_i\}_{i \in [m]} \subset \Sym^n$ with $m > 2n$ that  satisfy $\normsop{\sum_{i \in [m]} \Atil_i^2} \leq \sigma^2$, $\sum_{i\in [m]} \normsf{\Atil_i}^2 \leq mf^2$, and letting
    \[
        \tgK \defeq \Brace{ \xx \in \R^m \Bigg| \normop{ \sum_{i\in [m]} \xx_i \Atil_i } \leq \ccolor\min\Brace{\sigma + \sqrt{\sigma f}\log^{\frac 3 4}n,\;  \sigma\log^{\frac 1 4} n + \sqrt{\sigma f}\log^{\frac 1 2}n}},
    \]
    there is a subspace $T \subseteq \R^m$ with $\dim(T) \geq (1-c)m$, $\gamma_T(\gK) \geq \exp(-c m)$.
\end{proposition}

We note that the proof of Lemma 3.1 in \cite{BansalJM23} only showed how to
obtain the first of the two operator norm upper bounds 
 within the $\min$ expression in Proposition~\ref{prop:bjm_measure}, but the second follows straightforwardly by substituting an alternative matrix concentration inequality from \cite{Tropp18} into the same proof of \cite{BansalJM23}. We formally show how to obtain the second bound in Appendix~\ref{app:concentration}.

\begin{proposition}[Theorem 6, \cite{ReisR23}]\label{prop:rr_partialcolor}
Let $\ctight \in (0, 1)$ be a constant, let $S \subseteq \R^m$ be a subspace with $\dim(S) \ge 2\ctight m$, and let $\gK \subseteq \R^m$ be symmetric and convex. Suppose $\gamma_m(\gK) \ge \exp(-Cm)$ for a constant $C$. There is $\cset > 0$ depending only on $\ctight, C$ such that if
$\gg \sim \Nor(\vzero_m, \II_m)$, and 
\[\xx \defeq \arg\min_{\xx \in \cset \gK \cap [-1,1]^m \cap S} \|\xx - \gg\|_2,\]
then $|\{i \in [m] \mid |\xx_i| = 1\}| \ge \ctight m$
with probability $1-\exp(-\Omega(m))$.
\end{proposition}

Roughly speaking, \Cref{prop:bjm_measure} shows that a convex body over $\xx \in \R^m$, corresponding to a sublevel set  of $\normsop{\sum_{i \in [m]} \xx_i \Atil_i}$, has large Gaussian measure restricted to a subspace. \Cref{prop:rr_partialcolor} then produces a ``partially colored'' point $[-1, 1]^m$ with many tight constraints, i.e., coordinates $i \in [m]$ with $|\xx_i| = 1$, which also lies in the convex body from \Cref{prop:bjm_measure}. We summarize a useful consequence of \Cref{prop:rr_partialcolor} that is more compatible with \Cref{prop:bjm_measure}. The difference is that the variant in \Cref{cor:partialcolor_variance} only requires a Gaussian measure lower bound on the convex body restricted to a subspace, the type of guarantee that \Cref{prop:bjm_measure} gives.

\begin{corollary}\label{cor:partialcolor_variance}
In the setting of \Cref{prop:rr_partialcolor}, assume that $\gamma_{S}(\gK) \ge
\exp(-Cm)$ for a constant $C$, instead of $\gamma_m(\gK) \ge \exp(-Cm)$. There
is $\cset > 0$ depending only on $\ctight, C$ such that if 
$\gg \sim \Nor(\vzero_m, \II_m)$, and 
\[\xx \defeq \arg\min_{\xx \in \cset \gK \cap [-1,1]^m \cap S} \|\xx - \gg\|_2,\]
then $|\{i \in [m] \mid |\xx_i| = 1\}| \ge \ctight m$
with probability $1 - \exp(-\Omega(m))$.
\end{corollary}
\begin{proof}
Define $\gK' \subseteq \R^m$ to be $\gK \cap S$ expanded by a hypercube (centered at the origin and with side length $2$) in the subspace orthogonal to $S$, denoted $S_\perp$; concretely, let $\gK' \defeq (\gK \cap S) \oplus (\PP_{S_\perp}) [-1,1]^{\dim(S_\perp)}$, where $\oplus$ denotes the direct sum of two sets. Note that $\gK'$ is symmetric and convex, and $\gamma_m(\gK') \ge \exp(-C'm)$ for a constant $C'$ depending only on $C$ and the universal constant $\gamma_1([-1, 1])$, since the probability $\gg \sim \Nor(\vzero_m, \II_m)$ falls in $\gK'$ is the product of the probabilities of the independent events $\gg \in \gK' \cap S$ and $\gg \in \gK' \cap S_\perp$. Therefore, applying \Cref{prop:rr_partialcolor} to the subspace $S$ and the set $\gK'$ yields the conclusion, as $\cset \gK' \cap S = \cset \gK \cap S$.
\end{proof}

Finally, we give an equivalence we will later use.

\begin{lemma}\label{lemma:atilequala}
For $\{\AA_i\}_{i \in [m]} \subset \Sym^n$, a subspace $S \subseteq \R^m$ and a parameter $R \ge 0$, define 
\[\Atil_i \defeq \sum_{j \in [m]} \Brack{\PP_S}_{ji} \AA_j \text{ for all } i \in [m],\]
and their induced operator norm bodies
\[
\gK \defeq \Brace{ \xx \in \R^m \Bigg| \normop{ \sum_{i\in [m]} \xx_i \AA_i } \leq R},\; 
\tgK \defeq \Brace{ \xx \in \R^m \Bigg| \normop{ \sum_{i\in [m]} \xx_i \Atil_i } \leq R}.
\]
Then $\gK \cap T = \tgK \cap T$ for any subspace $T \subseteq S$.
\end{lemma}
\begin{proof}
It suffices to note that for $\xx \in T$, $\PP_S \xx = \xx$ and therefore 
\[
    \sum_{i \in [m]} \xx_i \Atil_i = \sum_{i \in [m]} \sum_{j \in [m]} [\PP_S]_{ji}\xx_i \AA_j = \sum_{j \in [m]} [\PP_S \xx]_j \AA_j = \sum_{j \in [m]} \xx_j \AA_j.
\]
This shows that for $\xx \in T$, $\normsop{\sum_{i \in [m]} \xx_i \Atil_i} \le R \iff \normsop{\sum_{i \in [m]} \xx_i \AA_i} \le R$, so $\gK \cap T = \tgK \cap T$.
\end{proof}

Next, we state a guarantee on a degree-rounding algorithm, $\ROalgo$. This algorithm is used in all of our sparsification subroutines, to deal with small degree imbalances induced by approximation errors in projection operations. The algorithm (Algorithm~\ref{alg:round}) follows a standard approach of rerouting the vertex imbalances $\BB_{\vG}^\top \zz$ through a spanning tree.
We bound the incurred discrepancy in the directed Laplacian by the size of $\zz$.
This procedure is related to, and inspired by, other tree-based rounding schemes in the literature, see e.g., \cite{KelnerOSZ13}.

\begin{algorithm2e}[ht!]\label{alg:round}
\caption{$\ROalgo(\vG, \zz, T)$}
\DontPrintSemicolon
\codeInput $\vG = (V, E, \ww)$, $\zz \in \R^E$, $T$ a tree subgraph of $G \defeq \und(\vG)$ \;
$\dd \gets \BB_{\vG}^\top \zz$\;
$\yy \gets $ unique vector in $\R^E$ with $\supp(\yy) \subseteq E(T)$ and $\BB_{\vG}^\top \yy = \dd$\;
\Return{$\yy$}
\end{algorithm2e}

\begin{restatable}{lemma}{restaterounding}\label{lemma:rounding}
Given $\vG = (V,E,\ww)$, a tree subgraph $T$ of $G \defeq \und(\vG)$ with $\min_{e \in E(T)} \ww_e \ge 1,$ \RO (Algorithm~\ref{alg:round}) returns in $O(n)$ time $\yy \in \R^{E}$ with $\supp(\yy) \subseteq T$ satisfying:
\begin{enumerate}
    \item $\BB_{\vG}^\top \yy = \dd.$
    \item $\norm{\yy}_{\infty} \le \frac{1}{2}\norm{\dd}_1.$
    \item For any $\zz \in \R^{E}$ satisfying $\BB_{\vG}^\top \zz = \dd,$ we have  $\normsop{\LL^{\dagger/ 2}_G \BB^{\top}_{\vG}(\YY - \ZZ)\HH_{\vG} \LL^{\dagger/ 2}_G} \leq n \norm{\zz}_1.$
    \item $\normsop{\LL^{\dagger/ 2}_G \BB^{\top}_{\vG}\YY\HH_{\vG} \LL_G^{\dagger/ 2}} \le n\norm{\yy}_1$.
\end{enumerate}
\end{restatable}

A proof of \Cref{lemma:rounding} is deferred to Appendix~\ref{app:rounding}. 

We next show how to combine \Cref{cor:partialcolor_variance} with our variance
bound in \Cref{lemma:variance} to slightly sparsify an Eulerian graph, while incurring small operator norm discrepancy.

\begin{algorithm2e}[ht!]\label{alg:eps}
\caption{$\ESOalgo(\{\vG_i\}_{i \in [I]}, \vG, T, \eps, W)$}
\DontPrintSemicolon
\codeInput $\{\vG^{(i)}\}_{i \in [I]}$, subgraphs of simple $\vG = (V, E, \ww)$
with $\max_{e \in \supp(\ww)} \ww_e \le W$, and such that $\{G^{(i)} \defeq
\und(\vG^{(i)})\}_{i \in [I]}$ are a $(\rho, 2, J)$-ER decomposition of $G \defeq
\und(\vG)$, $T$ a tree subgraph of $G$ with $\min_{e \in E(T)} \ww_e \ge 1$ and
$E(T) \cap \bigcup_{i \in I} E(\vG^{(i)}) = \emptyset$, 
$\delta, \eps \in (0, \frac 1 {100})$ \;\label{line:input_eps}
$\hm \gets \nnz(\ww)$, $\hE \gets \supp(\ww)$, $G \defeq \und(\vG)$, $n \gets |V|$ \;
\If{$\hm \ge 8nJ$}{\label{line:ifdense}
$S_i \gets \{\xx \in \R^{\hE} \mid \supp(\xx) \subseteq E(\vG^{(i)}), \BB_{\vG^{(i)}}^\top \WW_{\vG_i} \xx = \vzero_V\}$ for all $i \in [I]$\;
$S \gets \bigcup_{i \in [I]} S_i$\;\label{line:sdefine}
$\xx \gets $ point in $[-1, 1]^{\hE} \cap S$ such that for universal constants $\ceso$, $\ctight$,
\begin{equation}
\label{eq:eso_bounds}
\begin{gathered}\normop{\LL_{G}^{\frac \dagger 2}\Par{\sum_{i \in I}\sum_{e \in E(\vG_i)} \xx_e \ww_e \bb_e \ee_{h(e)}^\top } \LL_{G}^{\frac \dagger 2}} \\
    \le \ceso\min\Brace{\rho^{\frac 1 2} \log^{\frac 1 2}+ \rho^{\frac 3
    4}\log^{\frac 5 4}(n),\; \rho^{\frac 1 2}\log^{\frac 3 4}(n) + \rho^{\frac 3 4}\log(n)},\\
\Abs{\Brace{e \in \hE \mid \xx_e = -1}} \ge \ctight \hm
\end{gathered}
\end{equation}\;
\label{line:xexist}
\Comment*{Existence of $\xx$, $\ceso$, $\ctight$ follow from
\Cref{lemma:variance}, \Cref{prop:bjm_measure}, and \Cref{cor:partialcolor_variance}, see \Cref{lemma:eps}.}
$\xx' \gets $ extension of $\xx$ to $\R^E$ with $\xx'_e = \xx_e$ if $e \in \bigcup_{i \in I} E(G_i)$ and $\xx'_e = 0$ otherwise \;\label{line:xextend}
$\ww \gets \ww \circ (\vone_E + \xx')$\;
}
$D \defeq \{e \in \hE \mid \ww_e \le \ell\}$ \label{line:eso_d}\;
\Return{$\vG' \gets (V,E,[\ww]_{\hE \setminus D} + \ROalgo(\vG,[\ww]_D,T))$}\;
\end{algorithm2e}

\begin{lemma} \label{lemma:eps}
Suppose that 
$\ESO$ (Algorithm~\ref{alg:eps}) is run on inputs as specified in Line~\ref{line:input_eps}. Then, it returns $\vG' = (V,E,\ww')$ satisfying the following properties, with probability $\ge 1 - \delta$.
\begin{enumerate}
\item $\max_{e \in \hE} \ww'_e\le 2W$, $\min_{e \in \hE} \ww_e > \ell$ and 
    $\min_{e \in E(T)} \ww_e' \ge \min_{e \in E(T)} \ww_e - n\hm\ell$.
    \label{item:eso_1}
\item $\BB_{\vG}^\top \ww' = \BB_{\vG}^\top \ww$.
    \label{item:eso_2}
\item $\nnz([\ww']_{\hE}) \le (1-\ctight)\hm + \ceso \cdot nJ$.
    \label{item:eso_3}
\item 
    \[ \normsop{\LL_G^{\dagger/ 2} \BB_{\vG}^\top (\WW' - \WW)\HH_{\vG} \LL_G^{\dagger/ 2}} \le 
        \begin{gathered}
    \ceso\min\{\rho^{\frac 1 2} \log^{\frac 1 2}n + \rho^{\frac 3 4}\log^{\frac 5 4}n, \\ 
    \rho^{\half}\log^{\frac 3 4}n + \rho^{\frac 3 4}\log n\} + n\hm \ell. 
    \end{gathered} \]
    \label{item:eso_4}
\end{enumerate}
Moreover, $\ESOalgo$ is implementable in $\poly(n, \log U, \log \frac 1 \delta)$ time.
\end{lemma}
\begin{proof}
If \Cref{line:ifdense} does not pass, then \Cref{item:eso_1,item:eso_2,item:eso_3}
trivially hold and it only incurs the second term in the spectral error
(\Cref{item:eso_4}) due to \Cref{lemma:rounding}.
We then assume it does pass for the remainder of the proof.
We defer proving existence of $\xx, \ceso, \ctight$ in \eqref{eq:eso_bounds} until the end.
Since $\xx \in [-1, 1]^E$ and $\supp(\xx) \subseteq \hE$, no edge weight in
$\hE$ more than doubles, giving the first claim of \Cref{item:eso_1}.
Our definition of $D$ on \Cref{line:eso_d} and \RO ensures the second claim of
\Cref{item:eso_1}.
Next, since $\ww \circ \xx$ is only supported on $E' \defeq \bigcup_{i \in [I]}
E(G^{(i)})$ and $[\ww \circ \xx]_{E'}$ is the sum of disjoint circulations on
each $G_i$ by the definition of each $S_i$, $\ww \circ \xx$ is itself a
circulation on $G$.
Combining with the first guarantee of \Cref{lemma:rounding}, this implies
\Cref{item:eso_2}.
Since any $e \in \hE$ where $\xx_e = -1$ necessary has $\ww_e (1+\xx_e) = 0$ and
that \RO only introduces new non-zero entries on $E(T)$, \Cref{item:eso_3}
holds.
\Cref{item:eso_4} is follows from the definitions of $\ww'$ and
$D$, \eqref{eq:eso_bounds} and the third guarantee of \Cref{lemma:rounding}.

It remains to prove $\xx, \ceso, \ctight$ exist when \Cref{line:ifdense} passes. 
For each $e \in E'$, define $\AA_e$ and $\Atil_e$ as in the proof of
\Cref{lemma:variance} where $\vH$ is set to the partition piece $\vG^{(i)}$ with
$E(\vG^{(i)}) \ni e$. 
Summing the bound in \Cref{lemma:variance} over all pieces gives $\sigma^2 =
\rho$ in \Cref{prop:bjm_measure}, where we overload 
\[\Atil_e \gets \begin{pmatrix} & \Atil_e \\ \Atil_e^\top & \end{pmatrix} \]
in its use (padding with zeroes as necessary).
Correctness follows from the observations
\[\begin{pmatrix} & \AA \\ \AA^\top & \end{pmatrix}^2 = \begin{pmatrix} \AA \AA^\top &  \\  & \AA^\top \AA \end{pmatrix},\; \normop{\begin{pmatrix} & \AA \\ \AA^\top & \end{pmatrix}} = \normop{\AA}. \]
Further, we always have $f^2 \le \frac{n\sigma^2}{m}$ by linearity of trace and $\normsf{\Atil}^2 = \Tr(\Atil^2)$ for $\Atil \in \Sym^n$.
This gives a Gaussian measure lower bound on $\tgK$ restricted to a subspace
$S'$ of $S$.
By the characterization in \Cref{lemma:atilequala}, this also implies a Gaussian
measure lower bound on $\gK$ restricted to $S'$.
We next observe that $S$ is a subspace of $\R^{E'}$ where each $S_i$ enforces
$|V(G^{i})| - 1$ linear constraints (corresponding to weighted degrees in the subgraph).
By \Cref{def:er_partition}, the total number of such linear constraints is $\le
nJ$ and $|E'| \ge \frac{1}{2} \hm$.
The condition on \Cref{line:ifdense} then guarantees our final subspace has
sufficiently large dimension to apply \Cref{cor:partialcolor_variance}.
Finally, \Cref{cor:partialcolor_variance} guarantees existence of $\xx, \ctight, \ceso$ satisfying the guarantees in \eqref{eq:eso_bounds} (we may negate $\xx$ if it has more $1$s than $-1$s, and halve $\ctight$).

Lastly, we observe that Algorithm~\ref{alg:eps} is implementable in polynomial time. This is clear for
\RO and \Cref{line:xextend} to~\ref{line:eso_d}. The most computationally intensive
step is \Cref{line:xexist}, which consists of finding a subspace of large
Gaussian measure and solving a convex program. 
The latter is polynomial time \cite{GrotschelLS88}; the former is due to intersecting the explicit subspace from \Cref{line:sdefine} and the subspace from \Cref{prop:bjm_measure}.
The subspace from \Cref{prop:bjm_measure} is explicitly described in the proof of Lemma 3.1 of \cite{BansalJM23}; it is an eigenspace of a flattened second moment matrix.

All steps are deterministic except for the use of Corollary~\ref{cor:partialcolor_variance} in Line~\ref{line:xexist} (note that we can bypass Lemma~\ref{lem:approx_er} via exact linear algebra computations). This line succeeds with probability $\ge \half$ for a random draw. Finally, we can boost this line to have failure probability $\delta$ by running $\log(\frac 1 \delta)$ independent trials, as we can verify whether a run succeeds in $\poly(n, \log U)$ time.
\end{proof}

\begin{algorithm2e}[ht!]
\caption{$\ESalgo(\vG, \epsilon, \delta)$}
\label{alg:existsparse_new}
\DontPrintSemicolon
\codeInput Eulerian $\vG = (V, E, \ww)$ with $\ww_e \in [1, U]$ for all $e \in E$, $\eps \in (0, 1)$\;
$n \gets |V|$, $m \gets |E|$\;%
$T \gets $ arbitrary spanning tree of $G \defeq \und(\vG)$, $\hE \gets E \setminus E(T)$\;
$R \gets \lfloor \log_{1 - \frac{1}{2} \ctight} \frac 1 n \rfloor + 1$, $C_1 \gets
(\frac{256\ceso}{\ctight})^2, C_2 \gets (\frac{256\ceso}{\ctight})^{4/3}, C_3
\gets (\frac{256\ceso}{\ctight})^{3/2}$ for $\ctight, \ceso$ in \eqref{eq:eso_bounds}\;
$U_{\max} \gets U \cdot 2^R$, $J_{\max} \gets \log_2\Par{\frac{64mnRU_{\max}}{\eps}}$\;
$t \gets 0$, $\vG_0 \gets \vG$\;
\While{$\nnz([\ww_t]_{\hE}) > \max\{2 \ceso \ctight \cdot nJ_{\max}, \min\{C_1 \cdot \frac{n\log
n}{\eps^2} + C_2 \cdot \frac{n\log^{5/3}n}{\eps^{4/3}}, 2C_3 \cdot
\frac{n\log^{3/2} n}{\eps^{2}}\}\}$}{\label{line:while_es_start}
    $G_t \gets \und(\vG_t)$\;
    $S \gets \ERPalgo([G_t]_{\hE}, 2, \frac{\delta}{R})$ \Comment*{See \Cref{prop:er_partition}.}
    $\vG_{t+1} \defeq (V,E,\ww_{t+1}) \gets \ESOalgo(S,G_t,T,\frac{\eps}{8mnR},U_{\max})$\;
    $t \gets t+1$\;\label{line:while_es_end}
}
\Return{$\vH \gets (V,\supp(\ww_t),\ww_t)$}\;
\end{algorithm2e}

We are now ready to state and analyze our overall sparsification algorithm, $\ESalgo$ (Algorithm~\ref{alg:existsparse_new}). The following is a refined version of \Cref{thm:existential}.

\begin{restatable}{theorem}{restateexistentialdetailed}\label{thm:existentialdetailed}
    Given Eulerian $\vec{G} = (V,E,\ww)$ with $|V| = n$, $|E| = m$, $\ww \in [1,U]^E$ and $\epsilon \in (0,1)$, \ES (Algorithm~\ref{alg:existsparse_new}) returns Eulerian $\vec{H}$ such that $\vH$ is an $\eps$-approximate Eulerian sparsifier of $\vG$, and
    \begin{gather*}
    |\vH| = O\left(n\log U + \frac{n\log n}{\epsilon^2} \min\Brace{1 + (\epsilon\log n)^{\frac{2}{3}},\; {\log^{\frac 1 2}n}}\right),\\
    \log\Par{\frac{\max_{e \in \supp(\ww')}\ww'_e}{\min_{e \in \supp(\ww')}\ww'_e}} = O\Par{\log\Par{nU}}.
    \end{gather*}
    \ES succeeds with probability $\ge 1 - \delta$ and runs in time $\poly(n, \log U, \log\frac 1 \delta)$.
\end{restatable}
\begin{proof}
Recall from Section~\ref{sec:prelims} that we assume without loss of generality that $G$ is connected. Throughout, condition on the event that all of the at most $R$ calls to
\ERP succeed, which happens with probability $\ge 1 - \delta$.
Because \ESO guarantees that no weight grows by more than a $2$ factor in each call, $U_{\max}$ is a valid upper bound for the maximum weight of any edge throughout the algorithm's execution.
Moreover, since no weight falls below $\frac{\eps}{8mnR}$ throughout by
\ESO, $J_{\max} \defeq \log_2(\frac{64mnRU_{\max}}{\eps})$ is an upper bound
on the number of decomposition pieces ever returned by $\ERPalgo$, by
\Cref{prop:er_partition}.

Next, note that under the given lower bound on $\nnz([\ww_t]_{\hE})$ in a given
iteration (which is larger than $2\ceso \ctight \cdot nJ_{\max}$), the sparsity progress
guarantee in \Cref{item:eso_3} of \Cref{lemma:eps} shows that the number of
edges in each iteration is decreasing by at least a $(1-\ctight) +
\frac{1}{2}\ctight = (1-\frac{1}{2}\ctight)$ factor until termination.
Since $m \le n^2$ and the algorithm terminates before reaching $n$ edges, $R$ is
a valid upper bound on the number of iterations before the second condition in
\Cref{line:while_es_start} fails to hold, which gives the sparsity claim.

Let $\hm_i \defeq \nnz([\ww]_{\hE})$.
To prove the spectral error bound, we show by induction that until the algorithm
terminates, the following conditions hold, where we use $t$ to denote the number
of times the while loop runs in total:
\begin{enumerate}
    \item $\BB^\top \ww_i = \BB^\top \ww_0$ \label{item:es_2}
    \item $\hm_i \leq (1-\frac{1}{2}\ctight)^i\hm_0$. \label{item:es_3}
    \item 
    $ 
    \normsop{\LL^{\dagger/ 2} \BB^\top(\WW_i - \WW_0) \HH \LL^{\dagger/ 2}} \leq 
    2\ceso\sum_{j=0}^{i-1} 
    \min\{(\frac{n\log n}{\hm_j})^{\frac 1 2} + (\frac n {\hm_j})^{\frac 3 4}\log^{\frac 5 4}n, (\frac n {\hm_j})^{\half}\log^{\frac 3 4}n + (\frac n {\hm_j})^{\frac 3 4}\log n\}
    + \frac{i}{4R}.
    $\label{item:es_4}
\end{enumerate}
Note that \Cref{item:es_2,item:es_3,item:es_4} all hold trivially for $i=0$. Suppose inductively all conditions above hold for all iterations $k \leq i < t$.
By our stopping condition, $n \le \hm_i \le (1-\frac{1}{2}\ctight)^{i-1}\hm_0$
and hence $i \leq \frac{\log n}{-\log(1-\frac{1}{2}\ctight)} < R$.
\Cref{item:eso_2,item:eso_3} of \Cref{lemma:eps} then implies \Cref{item:es_2,item:es_3} are satisfied for iteration $i+1$.
We also have by \Cref{item:eso_4} of \Cref{lemma:eps} that
\begin{gather*}
    \normop{\LL_{G_i}^{\frac \dagger 2} \BB^\top (\WW_{i+1} - \WW_i) \HH \LL_{G_i}^{\frac \dagger 2}} \\
    \leq
    \ceso
    \min\Brace{
    \sqrt{\frac{n\log n}{\hm_i}} + \left(\frac{n\log^{\frac 5 3}n}{\hm_i}\right)^{\frac 3 4},\; \sqrt{\frac n {\hm_i}} \log^{\frac 3 4}(n) + \Par{\frac{n}{\hm_i}}^{\frac 3 4} \log(n)
} + \frac{\eps}{8R},
\end{gather*}
where we define $\vG_i \defeq (V,E,\ww_i)$ and $G_i \defeq \und(\vG_i)$ for any $0 \le i \le t$.
Note that $\vG_0 = (V,E,\ww_0) = \vG$, the original input Eulerian graph.
Moreover, $\LL_{\vG_i} - \LL_{\vG_0} = \BB^\top (\WW_i - \WW_0) \HH$.
By our choice of $C_1,C_2$, the stopping condition, \Cref{item:es_3}, and
\Cref{lemma:eps},
\begin{gather*}
    2\ceso \sum_{j=0}^{i-1} 
\min\Brace{
    \sqrt{\frac{n\log n}{\hm_j}} + \left(\frac{n\log^{\frac 5 3}n}{\hm_j}\right)^{\frac 3 4},\; \sqrt{\frac n {\hm_j}} \log^{\frac 3 4}(n) + \Par{\frac{n}{\hm_j}}^{\frac 3 4} \log(n)} \\
    \le 2\ceso \sum_{j = 0}^{i - 1} \Par{1 - \frac{\ctight}{4}}^{i - 1 - j} \\
    \cdot  \min\Brace{
        \sqrt{\frac{n\log n}{\hm_{i - 1}}} + \left(\frac{n\log^{\frac 5 3}n}{\hm_{i - 1}}\right)^{\frac 3 4},\; \sqrt{\frac n {\hm_{i - 1}}} \log^{\frac 3 4}(n) + \Par{\frac{n}{\hm_{i - 1}}}^{\frac 3 4} \log(n)} \\
    \le \frac{8\ceso}{\ctight}  \min\Brace{
        \sqrt{\frac{n\log n}{\hm_{i - 1}}} + \left(\frac{n\log^{\frac 5
    3}n}{\hm_{i - 1}}\right)^{\frac 3 4},\; \sqrt{\frac n {\hm_{i - 1}}}
\log^{\frac 3 4}(n) + \Par{\frac{n}{\hm_{i - 1}}}^{\frac 3 4} \log(n)} \le \frac
\eps 8 \le \frac{1}{8}.
\end{gather*}
As we also have $\frac{i \eps}{4R} \le \frac{\eps}{4} \le \frac{1}{4}$,
\Cref{fact:dirclose_undirclose} then gives $\half \LL \preceq \LL_{G_i} \preceq \frac 3 2 \LL$. Consequently, $G_i$ has the same connected components as the original graph $G$, i.e., since we assumed $G$ is connected, so is $G_i$. 
Hence, \Cref{fact:opnorm_precondition} implies that
\begin{gather*}
    \normop{\LL^{\frac \dagger 2} \BB^\top (\WW_{i+1} - \WW_i) \HH \LL^{\frac \dagger 2}}
    \leq
    2 \cdot \normop{\LL_{G_i}^{\frac \dagger 2} \BB^\top (\WW_{i+1} - \WW_i) \HH \LL_{G_i}^{\frac \dagger 2}}
    \\
    \leq
    2\ceso\min\Brace{
        \sqrt{\frac{n\log n}{\hm_i}} + \left(\frac{n\log^{\frac 5 3}n}{\hm_i}\right)^{\frac 3 4},\; \sqrt{\frac n {\hm_i}} \log^{\frac 3 4}(n) + \Par{\frac{n}{\hm_i}}^{\frac 3 4} \log(n)
    } + \frac{\eps}{4R}.
\end{gather*}
This proves \Cref{item:es_4} in the inductive hypothesis, as desired, and also implies that after the $t^{\text{th}}$ loop,
\begin{equation}\label{eq:before_round_opnorm}
\normop{\LL^{\frac \dagger 2}\BB^\top\Par{\WW_t - \WW_0}\HH\LL^{\frac \dagger 2}} \le \eps.
\end{equation}
The sparsity bound follows by explicitly removing any $e \in E$ where $[\ww_t]_e = 0$ from $\vH$. 
In light of Lemma~\ref{lemma:eps}, we note that each of the $\poly(n)$ calls to $\ESalgo$ can be implemented in $\poly(n, \log U, \log \frac 1 \delta)$ time, and all steps of Algorithm~\ref{alg:existsparse_new} other than $\ESOalgo$ run in linear time. We adjust the failure probability by a $\poly(n)$ factor to account for the multiple uses of Corollary~\ref{cor:partialcolor_variance} via a union bound, giving the claim.
\end{proof}

Theorem~\ref{thm:existential} is one logarithmic factor in $nU$ away from being optimal, up to low-order terms in $\eps$. The extra logarithmic factor is due to the parameters of our ER decomposition in Proposition~\ref{prop:er_partition}, and the low-order terms come from the additive terms with polylogarithmic overhead in Proposition~\ref{prop:bjm_measure}. In Appendix~\ref{sec:conjectures}, we discuss routes towards removing this overhead, and relate them to known results and open problems in the literature on graph decomposition (i.e., the \cite{AlevALG18} decomposition scheme) and matrix discrepancy (i.e., the matrix Spencer conjecture).

\section{Eulerian sparsification in nearly-linear time}\label{sec:algos}

In this section, building upon our approach from Section~\ref{sec:existence}, we provide a nearly-linear time algorithm for sparsifying Eulerian directed graphs. We develop our algorithm via several reductions.
\begin{itemize}
    \item In Section~\ref{ssec:fastonce}, we develop $\BFSalgo$, a basic subroutine which takes as input an initial subgraph with bounded ER diameter (in the sense of Definition~\ref{def:er_partition}), and edge weights within a constant multiplicative range. It then returns a reweighting of the initial subgraph which decreases weights by a constant factor on average.
    \item In Section~\ref{ssec:twophase}, we give a two-phase algorithm which builds upon $\BFSalgo$. In the first phase, the algorithm calls $\BFSalgo$ $\approx \log\log n$ times, and we demonstrate that these applications decrease a constant fraction of the edge weights from the original subgraph by a $\polylog(n)$ factor. We separate out this small cluster of edges and pass it to the second phase, which applies $\BFSalgo$ $\approx \log n$ times to decrease a constant fraction of edge weights by a polynomial factor. We then apply $\ROalgo$ to fully sparsify these edge weights, incurring small spectral error. Our sparsity-spectral error tradeoff in the second phase loses a polylogarithmic factor over our final desired tradeoff; this is canceled out by the mild edge weight decrease from the first phase, and does not dominate.
    \item In Section~\ref{ssec:outer_algo}, we recursively call our ER decomposition algorithm from Section~\ref{sec:partition}, and the two-phase procedure described above. Each round of calls makes constant factor progress on the overall sparsity of our final graph, and hence terminates quickly.
\end{itemize}

As a preliminary, we provide tools in Section~\ref{ssec:approximation} to streamline handling of approximation error incurred by state-of-the-art undirected Laplacian solvers, when projecting into circulation space.

\subsection{Approximating modified circulations}\label{ssec:approximation}

In this section, we give a self-contained solution to the key computational bottleneck in Section~\ref{ssec:fastonce} when using approximate Laplacian system solvers. We begin by introducing some notation to simplify our presentation. Let $\vH$ be a subgraph of $\vG = (V, E, \ww)$ with edge set $F$.
We define $H \defeq \und(\vH)$ and $H^2 \defeq (V(H), F, \ww^2_F)$, where $\ww^2$ is $\ww$ with its entries squared. We further define
\begin{equation}\label{eq:pch_def}
\PP_{\vH} \defeq \II_{F} - \CC_{\vH}, \text{ where }
\CC_{\vH} \defeq \WW_{F} \BB_{\vH} \LL_{H^2}^\dagger \BB_{\vH}^\top \WW_{F},
\end{equation}
and where $\II_{F}, \WW_{F} \in \R^{E \times E}_{\ge 0}$ zero out entries of $\II_{E}, \WW$ which do not correspond to edges in $F$. In Section~\ref{ssec:fastonce}, we apply reweightings which are circulations on $\vH$, but which also are orthogonal to a specified vector $\vv$. We will eventually set $\vv$ to be a current weight vector, to enforce that the total weight of the edges remains unchanged. We hence define the modified projection matrix
\begin{equation}
\label{eq:ppw_def}
\PP_{\vH, \vv} \defeq \PP_{\vH} - \uu_{\vH, \vv} \uu_{\vH, \vv}^\top, \text{ where }
\uu_{\vH, \vv} \defeq \frac{1}{\sqrt{\vv^\top \PP_{\vH} \vv} }\PP_{\vH} \vv.
\end{equation}
We prove a basic fact about $\PP_{\vH, \vv}$, motivated by the Sherman-Morrison formula.
\begin{lemma}\label{lem:rankone_fix}
For any $\uu \in \R^E$, $\PP_{\vH, \vv}$ defined in \eqref{eq:ppw_def} satisfies 
\[\PP_{\vH, \vv} \vv = \vzero_E,\; \PP_{\vH, \vv}^2 = \PP_{\vH, \vv}, \text{ and } \BB_{\vH}^\top \WW_{E(\vH)} \PP_{\vH, \vv} \uu = \vzero_E.\]
\end{lemma}
\begin{proof}
The first claim follows from directly computing $\uu_{\vH, \vv} \uu_{\vH, \vv}^\top \vv = \PP_{\vH} \vv$. The second follows similarly: since $\PP_{\vH}$ is an orthogonal projection matrix, $\uu_{\vH, \vv}$ is a unit vector, and we observe
\[\PP_{\vH} \uu_{\vH, \vv}\uu_{\vH, \vv}^\top = \uu_{\vH, \vv}\uu_{\vH, \vv}^\top \PP_{\vH} = \uu_{\vH, \vv}\uu_{\vH, \vv}^\top.\]
Finally, the last follows from the fact that $\BB_{\vH}^\top \WW_{E(\vH)} \PP_{\vH}$ is the zero operator on $\R^{E \times E}$.
\end{proof}
Thus, $\PP_{\vH, \vv}$ is the projection matrix into the subspace of $\PP_{\vH}$'s span that is orthogonal to $\vv$.

Algorithm~\ref{alg:approxflow} solves the following problem: on input $\xi > 0$, $\zz \in \R^E$ with $\supp(\zz) \subseteq F$, $\norm{\zz}_\infty \le 1$, return $\xx \in \R^E$ with
\begin{equation}\label{eq:approx_problem}
\supp(\xx) \subseteq F,\; \norm{\xx - \PP_{\vH, \vv} \zz}_\infty \le \xi,\; \norm{\BB_{\vG}^\top \WW \xx}_\infty \le \xi,\; |\inprod{\xx}{\vv}| \le \xi \norm{\vv}_2. 
\end{equation}
In other words, for an error parameter $\xi$, we wish to enforce that $\ww \circ \xx$ is an approximate circulation, and that $\xx$ is approximately orthogonal to $\vv$ and approximates the true $\PP_{\vH, \vv} \zz$ we wish to compute. We remark that $\xx = \PP_{\vH, \vv} \zz$ satisfies \eqref{eq:approx_problem} with $\xi = 0$. We will ultimately call Algorithm~\ref{alg:approxflow} with inverse-polynomially small $\xi$, and apply $\RO$ to incur small error when rounding the residual.

\begin{algorithm2e}[ht!]
\caption{$\PMROalgo(\vH, \vv, \zz, \delta, \xi)$}
\label{alg:approxflow}
\DontPrintSemicolon
\codeInput $\vH$, a subgraph of $\vG = (V, E, \ww)$ with $\|\ww\|_\infty \le u$ and $F \defeq E(\vH)$, $\vv,\zz \in \R^{E}$ with $\supp(\vv), \supp(\zz) \subseteq F$ and $\|\zz\|_\infty \leq 1$, $\delta, \xi \in (0, 1)$  \;
$n \gets |V|$\;
$\xi' \gets \frac{\xi}{9nu\sqrt{m}}$\;
$\aa \gets \xi'$-approximate solution to $\LL_{H^2} \aa = \BB^\top_{\vH} \WW_F \vv$, with probability $\ge 1 - \frac \delta 2$\; \label{line:aa}
$\bb \gets \xi'$-approximate solution to $\LL_{H^2} \bb = \BB^\top_{\vH} \WW_F \zz$, with probability $\ge 1 - \frac \delta 2$\;\label{line:bb}
$\uu \gets \frac{\WW_F\BB_{\vH}^\top \aa}{\|\WW_F \BB_{\vH}^\top \aa\|_2}$,
$\yy \gets \WW_F \BB_{\vH}^\top \bb$\;
\Return{$\xx \gets \zz - \yy - \inprods{\yy}{\uu}\uu$}\;
\end{algorithm2e}

Before giving our analysis in Lemma~\ref{lem:pmro}, we require one elementary helper calculation.

\begin{lemma}\label{lem:helper_unit_vector}
	Let $\aa, \aas \in \R^d$ satisfy $\norm{\aa - \aas}_2 \le \alpha \norm{\aa}_2$ for $\alpha \in (0, \frac 1 5)$. Then, for $\uu \defeq \frac{\aa}{\norm{\aa}_2}$ and $\uus \defeq \frac{\aas}{\norm{\aas}_2}$, we have $\norm{\uu - \uus}_2 \le 2\alpha$.
\end{lemma}
\begin{proof}
	The problem statement is invariant under scaling $\aa$, so without loss of generality assume $\uu = \aa$, which implies $\norm{\aas} \in [1 - \alpha, 1 + \alpha]$. The conclusion follows by triangle inequality:
	\begin{align*}
		\norm{\uu - \uus}_2 \le \norm{\uu - \aas}_2 + \norm{\aas - \uus}_2 \le \alpha + \Abs{\norm{\aas} - 1} \le 2\alpha.
	\end{align*}
\end{proof}

\begin{lemma}\label{lem:pmro}
Under the stated input assumptions, \PMRO (Algorithm~\ref{alg:approxflow}) using Proposition~\ref{prop:js21} in Lines~\ref{line:aa}-\ref{line:bb} returns $\xx$ satisfying \eqref{eq:approx_problem} in time $\bO(|F| \log\frac{nu}{\xi\delta})$ with probability $\ge 1 - \delta$.
\end{lemma}
\begin{proof}
The problem definition and error guarantee \eqref{eq:approx_problem} are invariant under scaling $\vv$, so we assume $\norm{\vv}_2 = 1$ without loss of generality. Further, the problem is identical if we eliminate all coordinates on $E \setminus F$ (as the input and output are supported in $F$), so we only handle the case $E = F$. Finally, for simplicity in this proof, we let $\LL \defeq \LL_{H^2}$, $\BB \defeq \BB_{\vH}$, $\WW \defeq \WW_F$, $\II \defeq \II_F$, and $n \defeq |V|, m \defeq |F|$, and define the ideal vectors (which would be computed in the algorithm if $\xi = 0$):
\begin{gather*}
\aas \defeq \LL^\dagger \BB^\top\WW \vv,\; \bbs \defeq \LL^\dagger \BB^\top\WW \zz,\\
\uus \defeq\frac{\WW\BB\aas}{\|\WW\BB\aas\|_2} = \uu_{\vH,\vv},\; \yys \defeq \WW\BB\bbs,\; \xxs \defeq \zz - \yys - \inprods{\yys}{\uus}\uus = \PP_{\vH,\vv}\zz.
\end{gather*}
First, by the definition of approximate solutions (see \Cref{prop:js21}), we have
\[
    \norm{\WW\BB(\aa - \aas)}_2 = \norm{\aa - \aas}_{\LL} \leq \xi' \norm{\aas}_{\LL} = \xi' \norm{\WW\BB\aas}_2.
\]
Hence, by applying Lemma~\ref{lem:helper_unit_vector}, we have $\norm{\uu - \uus}_2 \le 2\xi'$.
Similarly, 
\[
    \norm{\yy - \yys}_2 =\norm{\WW\BB(\bb - \bbs)}_2 = \norm{\bb-\bbs}_{\LL} \leq \xi' \norm{\bbs}_{\LL} = \xi' \norm{\WW\BB\bbs}_2 = \xi' \norm{\yys}_2
    \leq \xi' \norm{\zz}_2,
\]
where the last equality follows by $\yys = \CC_{\vH} \zz$ and the fact that $\CC_{\vH}$ is a orthogonal projection.
Now,
\begin{align*}
    \xx - \xxs 
    &=
    (\yys - \yy) + (\inprod{\yys}{\uus}\uus - \inprod{\yy}{\uu}\uu)
    \\
    &=
    (\yys - \yy) + \inprod{\yys}{\uus-\uu}\uus + \inprod{\yys-\yy}{\uu}\uus + \inprod{\yy}{\uu}(\uus-\uu),
\end{align*}
so that by the triangle and Cauchy-Schwarz inequalities, the first conclusion in \eqref{eq:approx_problem} holds:
\begin{align*}
    \norm{\xx-\xxs}_\infty 
    &\leq \norm{\xx-\xxs}_2
    \\
    &\leq 
    \norm{\yys - \yy}_2 + \norm{\uus-\uu}_2 \norm{\yys}_2 \norm{\uus}_2 +
    \norm{\yys-\yy}_2 \norm{\uu}_2 \norm{\uus}_2 + 
    \norm{\uus-\uu}_2 \norm{\yy}_2 \norm{\uu}_2
    \\
    &\leq
    \norm{\yys - \yy}_2 + 2\xi' \norm{\yys}_2 + \norm{\yys-\yy}_2 + 2\xi' (\norm{\yys}_2 + \norm{\yys-\yy}_2)
    \\
    &\leq
    9\xi' \norm{\zz}_2 < 9\xi'\sqrt{m} \norm{\zz}_\infty \le 9\xi'\sqrt{m} < \xi,
\end{align*}
given that $\xi' < 1$.
Moreover, letting $\norm{\AA}_{\infty \to \infty} \defeq \sup_{\norms{\xx}_\infty = 1} \norm{\AA \xx}_\infty$ be the largest $\ell_1$ norm of a row of $\AA$, and noting that $\norm{\BB}_{\infty \to \infty} \le n$ and $\norm{\WW}_{\infty \to \infty} \le u$, we have
\begin{align*}
\norm{\BB^\top \WW\xx}_\infty &\stackrel{(a)}{=} \norm{\BB^\top \WW(\xx-\xxs)}_\infty \\
&\le \norm{\BB^\top}_{\infty \to \infty} \norm{\WW}_{\infty \to \infty} \norm{\xx - \xxs}_\infty \le nu\norm{\xx - \xxs}_\infty \le 9nu\xi'\sqrt{m},
\\
|\inprod{\xx}{\vv}| &\stackrel{(b)}{=} |\inprod{\xx-\xxs}{\vv}| \leq \norm{\xx-\xxs}_2 \norm{\vv}_2 \leq 9\xi'\sqrt{m} \norm{\vv}_2.
\end{align*}
Here, both (a) and (b) followed from \Cref{lem:rankone_fix}.
By our choice of $\xi' = \frac{\xi}{9nu\sqrt{m}} < 1$, we can guarantee all the desired bounds in \eqref{eq:approx_problem}.
Finally, the runtime bound follows directly from \Cref{prop:js21}.
\end{proof}

\subsection{Basic partial sparsification}\label{ssec:fastonce}

In this section, we give the basic subroutine of our fast sparsification algorithms, which modifies the edge weights on a well-controlled subgraph (formally, see Definition~\ref{def:cluster}). We first require stating several standard helper matrix concentration results from the literature.
\begin{lemma}[Theorem 7.1, \cite{Tropp11}]\label{lem:matrix_azuma}
Let $\delta \in (0, 1)$ and let $\Brace{\MM_k}_{k \in [K]} \in \R^{d \times d}$ be a sequence of matrices, and let $\ss \in \Brace{\pm 1}^K$ be a martingale sequence of Rademachers, i.e., $\ss_k$ is a Rademacher random variable conditioned on $\{\ss_j\}_{j \in [k - 1]}$ for all $k \in [K]$. Further, suppose for $\sigma \ge 0$,
\begin{equation}\label{eq:matrix_variance_azuma}\sum_{k \in [K]} \MM_k \MM_k^\top \preceq \sigma^2 \II_d,\; \sum_{k \in [K]} \MM_k^\top \MM_k \preceq \sigma^2 \II_d. \end{equation}
Then with probability $\ge 1 - \delta$,
\[\normop{\sum_{k \in [K]} \ss_k \MM_k} \le \sigma \sqrt{8\log\Par{\frac{2d}{\delta}}}.\]
\end{lemma}

\begin{lemma}\label{lem:proj_sign_small}
Let $\delta \in (0, 1)$, let $\PP \in \R^{d \times d}$ be an orthogonal projection matrix, and let $\ss \in \{\pm 1\}^d$ have independent Rademacher entries. There is a universal constant $\csign$ such that 
\[\norm{\PP \ss}_\infty \le \csign\sqrt{\log \frac d \delta} \text{ with probability} \ge 1 - \delta.\]
\end{lemma}
\begin{proof}
For any fixed $j \in [d]$, the random variable $X \defeq \ee_j^\top \PP \ss$ is sub-Gaussian with parameter $\sigma \defeq \norms{\PP_{j:}}_2 \le 1$. Standard sub-Gaussian concentration bounds (e.g., \cite{Vershynin18}, Proposition 2.5.2) now imply that with probability $\ge 1 - \frac \delta {d}$, we have for a universal constant $\csign$,
$X \le \csign \sqrt{\log \frac {d} \delta}$.
Applying a union bound for all $j \in [d]$ concludes the proof.
\end{proof}

We also use the following helper scalar concentration inequality.

\begin{lemma}\label{lem:moment_error}
Let $X$ be a $1$-sub-Gaussian random variable with $\E X = 0$, and let $\event$ be an event on the outcome of $X$ with $\Pr[\event] \ge 1 - \delta$ where $\delta \le \frac 1 {10}$. Then, $ \Abs{\E\Brack{X^2 - \E\Brack{X^2} \mid \event}} \le 300\sqrt{\delta}$.
\end{lemma}
\begin{proof}
Let $\ind_{\event}$ and $\ind_{\event^c}$ denote the $0$-$1$ indicator variables for $\event$ and its complement $\event^c$. Further, we will assume $\Pr[\event^c] = \delta$ as the stated bound is monotone in $\delta$. The random variable $Z \defeq X^2 - \E[X^2]$ is $16$-sub-exponential (Lemma 1.12, \cite{RigolletH17}), so applying the Cauchy-Schwarz inequality and standard sub-exponential moment bounds (Lemma 1.10, \cite{RigolletH17}) yields
\begin{align*}
\Abs{\E\Brack{Z \mid \event}} = \frac 1 {\Pr[\event]}\Abs{\E[Z \cdot \ind_{\event}]} = \frac 1 {1 - \delta} \Abs{\E[Z \cdot \ind_{\event^c}]}\le \frac 1 {1 - \delta} \E[Z^2]^{\half} \E\Brack{\ind_{\event^c}}^{\half} \le 300\sqrt{\delta}.
\end{align*}
\end{proof}

Finally, to simplify the statement of the input to our algorithm, we give a useful definition.

\begin{definition}[Cluster]\label{def:cluster}
We say $\vH$ is a $(\bw, \rho)$-cluster in $\vG = (V, E, \ww)$ if $\vH$ is a subgraph of $\vG$, $\ww_e \in [\bw, 2\bw]$ for all $e \in E(\vH)$, and letting $G \defeq \und(\vG)$,
\[\Par{\max_{e \in E(\vH)} \ww_e} \cdot \Par{\max_{u, v \in V(\vH)} \ER_G(u, v)} \le \rho.\]
\end{definition}
By definition, any piece in a $(\rho, 2, J)$-ER decomposition of $G = \und(\vG)$ (Definition~\ref{def:er_partition}) is a $(\bw, \rho)$-cluster in $\vG$, for some $\bw$.
We now state our main algorithm in this section, $\BFSalgo$. 

\begin{algorithm2e}[ht!]\label{alg:basic_fast_sparsify}
\caption{$\BFSalgo(\vH, \vG, \wws, \ell, \delta, \eps, F, T)$}
\DontPrintSemicolon
\codeInput $\ell \in (0, 1)$, $\vH$ a subgraph of $\vG = (V, E, \ww)$ with $|E(\vH)| \ge 40 |V(\vH)|$, $\wws \in \R^E$ with
\begin{equation}\label{eq:wws_reqs} \frac{\norms{\ww}_1}{\norm{\wws}_1} \in [0.99, 1.01],\; [\wws]_{E \setminus E(\vH)} = \ww_{E \setminus E(\vH)}, \text{ and } \frac \ell 2 [\wws]_e \le \ww_e \le 60[\wws]_e \text{ for all } e \in E(\vH),\end{equation}
and $\vHs \defeq (V(\vH), E(\vH), [\wws]_{E(\vH)})$ is a $(\bw, \rho)$-cluster in $\vGs \defeq (V, E, \wws)$ and $0.9\LL_G \preceq \LL_{\Gs} \preceq 1.1 \LL_G$ for $\Gs \defeq \und(\vGs)$, $\delta, \eps \in (0, \frac 1 {100})$, $F \subseteq E(\vH)$ with $|F| \ge \frac{|E(\vH)|}{4}$, $T$ a tree subgraph of $G \defeq \und(\vG)$ with $\min_{e \in E(T)} \ww_e \ge 1$ \;
$m \gets |E(\vG)|$, $n \gets |V(\vG)|$\;
$\xi \gets \min(\frac \ell {10}, \frac{1}{1000\csign\log(\frac{60m\tau}{\delta})}, \frac{\eps}{200mn^2\tau})$, for $\csign$ from Lemma~\ref{lem:proj_sign_small} \;
$\eta \gets \frac 1 {20\csign \sqrt{\log \frac {60m\tau}{\delta}}}$, $\tau \gets \lceil \frac {720} {\eta^2}\rceil$\;
$t \gets 0$, $L_t \gets \{e \in F \mid [\ww_{t}]_e \ge 50\min([\wws]_e, \frac{\|\ww_F\|_1}{|F|})\}$, $S_t \gets \{e \in F \mid [\ww_{t}]_e \le \ell[\wws]_e \}$\;
\While{$|S_t| < \frac 1 4 |F|$ $\mathbf{and}$ $\sum_{e \in E(\vH)} \log([\ww_t]_e) - \log([\ww_0]_e) > -|E(\vH)|$}{\label{line:while_bfs_start}
$\ww_0 \gets \ww$\;
\For{$0 \le t \le \tau$}{\label{line:while_inner_start}
	 $L_t \gets \{e \in F \mid [\ww_{t}]_e \ge 50\min([\wws]_e, \frac{\|\ww_F\|_1}{|F|})\}$, $S_t \gets \{e \in F \mid [\ww_{t}]_e \le \ell[\wws]_e \}$\;
\If{$|S_t| < \frac 1 4 |F|$ $\mathbf{and}$ $\sum_{e \in E(\vH)} \log([\ww_t]_e) - \log([\ww_0]_e) > -|E(\vH)|$}{\label{line:if_bfs_start}
$\vH_t \gets (V(\vH_t), F \setminus (S_t \cup L_t), [\ww_t]_{F \setminus (S_t \cup L_t)})$\;
$\ss \gets $ random vector in $\{-1, 0, 1\}^E$, where $\ss_e$ is an independent $\pm 1$ random variable for all $e \in E(\vH_t)$, and $\ss_e = 0$ for all $e \in E \setminus E(\vH_t)$ \;
$\xx_t \gets \PMROalgo(\vH_t, \ww_t,  \eta \ss, \frac{\delta}{4\tau\log_2(\frac 4 \delta)}, \xi)$\;\Comment*{That is, $\xx_t \approx [\xxs]_t \defeq \eta \PP_{\vH_t, \ww_t} \ss$.}\label{line:pmro}
$\ww_{t + 1} \gets \ww_t \circ (\vone_E+ \xx_t)$\;
}\label{line:if_bfs_end}
\Else{\label{line:else_bfs_start}
$\ww_{t + 1} \gets \ww_t$\;
}\label{line:else_bfs_end}
}\label{line:while_inner_end}
}\label{line:while_bfs_end}
$\dd \gets \BB_{\vG}^\top (\ww - \ww_{t})$\;\label{line:dd_define}
$\yy \gets $ unique vector in $\RR^E$ with $\supp(\yy) \subseteq E(T)$ and $\BB_{\vG}^\top \yy = \dd$\;
$\ww_{t} \gets \ww_{t} + \yy$\;\label{line:degree_fixing}
\Return{$\ww' \gets \ww_t$}
\end{algorithm2e}

Intuitively, $\BFSalgo$ randomly reweights a current subset of edges in each of $\tau$ iterations, after removing any edge whose weight has significantly changed with respect to a reference vector $\wws$. In each loop of Lines~\ref{line:while_bfs_start} to~\ref{line:while_bfs_end}, the algorithm terminates if either a constant fraction of edge weights in $E(\vH)$ have decreased by an $\ell$ factor compared to $\wws$, or a certain potential function bounding the change in weights has decreased significantly. Moreover, each reweighting adds a circulation (and hence preserves degrees), while maintaining that $\|\ww_t\|_1$ is unchanged, up to an inverse-polynomial approximation error due to our subroutine $\PMROalgo$. The algorithm simply iterates this loop until termination. We now analyze Algorithm~\ref{alg:basic_fast_sparsify}, by bounding the spectral error and showing that each loop of Lines~\ref{line:while_bfs_start} to~\ref{line:while_bfs_end} is likely to terminate.

\begin{lemma}\label{lem:bfs_guarantee}
There is a universal constant $\cbfs$ such that if $\cbfs \cdot \alpha\rho\log(\frac m \delta) \le 1$, where
\[\alpha \defeq \frac{\norm{\ww_F}_1}{|F|\bw},\]
\BFS (Algorithm~\ref{alg:basic_fast_sparsify}) returns $\ww'$ satisfying, with probability $\ge 1 - \delta$:
\begin{enumerate}
\item $\BB_{\vG}^\top \ww' = \BB_{\vG}^\top \ww$ and $\frac{\norm{\ww'}_1}{\norm{\ww}_1} \in [1 - \eps, 1 + \eps]$.\label{item:bfs_1}
\item $\ww'_e \in [\frac \ell 2 [\wws]_e, 60 [\wws]_e]$ for all $e \in E(\vH)$. \label{item:bfs_2}
\item Either $|\{e \in E(\vH) \mid \ww'_e \le \ell[\wws]_e\}| \ge \frac 1 4 |F|$, or 
$\sum_{e \in E(\vH)} \log\Par{\frac{\ww'_e}{\ww_e}} \le -|E(\vH)|$.\label{item:bfs_3}
\item $\normsop{\LL_G^{\dagger/ 2} \BB_{\vG}^\top (\WW' - \WW) \HH_{\vG} \LL_G^{\dagger/ 2}} \le \cbfs \cdot \sqrt{\alpha\rho\log(\frac m \delta)} + \eps$ , where $G \defeq \und(\vG)$.\label{item:bfs_4}
\end{enumerate}
The runtime of \BFS is, for $Z \sim \textup{Geom}(p)$ where $p \in [\half, 1]$,\footnote{The $\polyloglog$ factors hidden by the $\bO$ notation will be $\polyloglog(nU)$ factors where $U$ is the edge weight ratio of the original graph we sparsify in Section~\ref{ssec:outer_algo}, as discussed in that section.}
\[\bO\Par{|E(\vH)| \log\Par{\frac{n}{\delta\eps\ell}}\log\Par{\frac n \delta} \cdot Z + |V|}.\]
\end{lemma}
\begin{proof}
Let $\hm \defeq |E(\vH)|$. Because the algorithm continues looping Lines~\ref{line:while_bfs_start} to~\ref{line:while_bfs_end} until the condition in Item~\ref{item:bfs_3} is met, the conclusion that Item~\ref{item:bfs_3} holds is immediate. 
The remainder of the proof proceeds as follows. We first prove the runtime claim by giving a constant lower bound on the probability a single run of Lines~\ref{line:while_bfs_start} to~\ref{line:while_bfs_end} ever fails to enter the else branch on Line~\ref{line:else_bfs_start}, assuming for simplicity that all calls to $\PMROalgo$ are exact, i.e., that every time Line~\ref{line:pmro} is run,
\begin{equation}\label{eq:pmro_exact}\xx_t = [\xxs]_t = \eta\PP_{\vH_t, \vv_t} \ss.\end{equation}
We next prove that Items~\ref{item:bfs_1},~\ref{item:bfs_2}, and~\ref{item:bfs_4} hold with the requisite failure probability. Finally, we modify the argument to handle approximation error due to inexactness in Line~\ref{line:pmro}.

\paragraph{Runtime bound.} 
Our goal in this part of the proof is to establish that each run of Lines~\ref{line:while_bfs_start} to~\ref{line:while_bfs_end} results in the else branch on Line~\ref{line:else_bfs_start} being entered with probability $\ge \half$. We use this claim to obtain our runtime bound. In the following discussion, fix a single run of Lines~\ref{line:while_bfs_start} to~\ref{line:while_bfs_end}. We let $\calE_t$ denote the event that $\norm{\xx_t}_\infty \le \frac 1 {10}$ conditioned on the randomness of all iterations $0 \le s < t$.
We also let $\calF_t$ denote the event that the algorithm enters the if branch on Line~\ref{line:if_bfs_start} on iteration $t$, and
\begin{equation}\label{eq:potential_def}
p_t \defeq \Pr\Brack{\bigcup_{0 \le s < t} \calF_s \mid \bigcup_{0 \le s < t} \event_s},\; \Phi_t \defeq \E\Brack{\sum_{e \in E(\vH)}\log\Par{\frac{[\ww_{t + 1}]_e}{[\ww_0]_e}} \mid \bigcup_{0 \le s \le t} \event_s},
\end{equation}
where both definitions in \eqref{eq:potential_def} are taken with respect to all randomness used in the current run of Lines~\ref{line:while_bfs_start} to~\ref{line:while_bfs_end}. In other words, $p_t$ is the probability the algorithm has not entered the else branch on Line~\ref{line:else_bfs_start} in any iteration $0 \le s < t$, and $\Phi_t$ is an expected potential function tracking edge weights over iterations $0 \le s \le t$, both conditioned on $\bigcup_{0 \le s \le t} \event_s$ occurring. Also, note that by Lemma~\ref{lem:proj_sign_small}, $\Pr[\event_t] \ge 1- \frac \delta {4\tau}$, so $\Pr[\bigcup_{0 \le t \le \tau} \event_t] \ge 1 - \frac{\delta\tau}{4\tau} \ge \frac 3 4$. Thus, if we can show $p_{\tau} \ge \frac 2 3$, we have our goal:
\begin{equation}\label{eq:terminate_half}\Pr\Brack{\bigcup_{0 \le s \le \tau} \mathcal{F}_t} = \Pr\Brack{\bigcup_{0 \le s \le \tau} \calF_t \mid \bigcup_{0 \le s \le \tau} \event_t}\Pr\Brack{\bigcup_{0 \le s \le \tau}\calE_t} \ge \frac 2 3\cdot \frac 3 4 = \half.\end{equation}

Suppose for contradiction that $p_{\tau} \le \frac 2 3$, so that $p_t \le \frac 2 3$ for all $0 \le t \le \tau$. First, we compute, following the convention that $[\xx_t]_e = 0$ if $e \not\in E(\vH_t)$ or we run the else branch in iteration $t$,
\begin{equation}\label{eq:pot_drop}
	\begin{aligned}
\Phi_t - \Phi_{t - 1} &= \E\Brack{\sum_{e \in E(\vH_t)} \log\Par{1 + [\xx_t]_e} \mid \bigcup_{0 \le s \le t} \event_s} \\
&\le \E\Brack{\sum_{e \in E(\vH_t)} [\xx_t]_e - \frac 1 3[\xx_t]_e^2 \mid \bigcup_{0 \le s \le t} \event_s} \\
&= \Par{1 - p_t}  \E\Brack{\sum_{e \in E(\vH_t)} [\xx_t]_e - \frac 1 3 [\xx_t]_e^2 \mid \event_t \cup \bigcup_{0 \le s < t} (\event_s \cup \calF_s)} \\
&\le \frac 1 3\max\Par{\E\Brack{\sum_{e \in E(\vH_t)} [\xx_t]_e - \frac 1 3 [\xx_t]_e^2 \mid \event_t \cup \bigcup_{0 \le s < t} (\event_s \cup \calF_s)}, 0}.
\end{aligned}
\end{equation}
The second line used the approximation $\log(1 + x) \le x - \frac 1 3 x^2$ for $|x| \le \frac 1 {10}$, the third line used that no weight changes if we enter the else branch, and the last line used our assumption $p_t \le \frac 2 3$. 

We next upper bound the right-hand side of \eqref{eq:pot_drop}. Observe that the definition of $\xx_t$ (assuming \eqref{eq:pmro_exact}) ensures $\sum_{e \in E} [\xx_t \circ \ww_t]_e = 0$ using Lemma~\ref{lem:rankone_fix}, so $\norm{\ww_t}_1 = \norm{\ww_0}_1$ in every iteration. 
Since any $e \in L_t$ due to $[\ww_t]_e \ge 50[\wws]_e$ must have $[\ww_t]_e \ge 50\bw$, and $\norm{\vv_t}_1 \le 1.01\norm{\wws}_1 \le 2.02\hm \bw$, there can be at most $\frac{\hm} {24}$ such edges. Similarly, at most $\frac{\hm} {50}$ edges $e \in F$ can have $[\ww_t]_e \ge \frac{50\norm{\ww}_1}{|F|}$, so $|L_t| \le \frac 1 4 |F|$ throughout the algorithm. Hence under $\bigcup_{0 \le s < t} (\event_s \cup \calF_s)$, which also implies $|S_t| \le \frac 1 4 |F|$, we always have $|E(\vH_t)| \ge \half |F|$. Moreover, note that since $\xx_t = \eta\PP_{\vH_t, \vv_t} \ss$ for Rademacher $\ss$,
\begin{equation}\label{eq:unconditional_logbound}\E\Brack{\sum_{e \in E(\vH_t)} [\xx_t]_e} = 0,\; \E\Brack{\sum_{e \in E(\vH_t)}[\xx_t]_e^2 } = \eta^2 \E\norm{\PP_{\vH_t,\vv_t}}_2^2 = \eta^2\Tr\Par{\PP_{\vH_t,\vv_t}}.\end{equation}
However, note that the dimension of the subspace spanned by $\PP_{\vH_t,\vv_t}$ is at least
 \[|E(\vH_t)| - (|V(\vH)| - 1) - 1 \ge \frac{\hm}{8} - \frac{\hm}{40} = \frac{\hm}{10},\]
 under the assumption $|E(\vH_t)| \ge \half |F| \ge \frac{\hm}{8}$,
since it has $|V(\vH)| - 1$ degree constraints and one orthogonality constraint to $\ww_t$. 
We now handle conditioning on the event $\event_t$, which satisfies $1 - \Pr[\event_t] \le \frac 1 {6000^2}$. Combining \eqref{eq:unconditional_logbound} with the above, and using that each $[\xx_t]_e$ is $1$-sub-Gaussian (Lemma~\ref{lem:proj_sign_small}) and the set of $\ss$ satisfying $\event_t$ is closed under negation, applying Lemma~\ref{lem:moment_error} shows
\begin{equation}\label{eq:potdrop_exact}\E\Brack{\sum_{e \in E(\vH_t)} [\xx_t]_e \mid \event_t} = 0,\; \E\Brack{\sum_{e \in E(\vH_t)}[\xx_t]_e^2 \mid \event_t } \ge \eta^2\Par{\frac{\hm}{10} - \hm \cdot \Par{300 \cdot \frac 1 {6000}}} = \frac{\eta^2 \hm}{20}.\end{equation}
Therefore, combining with \eqref{eq:pot_drop} shows that $\Phi_t$ decreases by at least $\frac {\eta^2 \hm} {180}$ for each of the first $\tau$ iterations. However, we also have that with probability $1$,
\[\sum_{e \in E(\vH)}\log\Par{\frac{[\ww_{\tau}]_e}{[\ww_0]_e}} \mid \bigcup_{0 \le s \le \tau} \event_s \ge -2\hm. \]
This is because the algorithm freezes the weights $\ww_{t}$ as soon as $\sum_{e \in E(\vH)}\log\Par{\frac{[\ww_{t}]_e}{[\ww_0]_e}} \le -\hm$, and the potential can only change by $-\hm$ in an iteration $t$ assuming $\event_t$, since then $\log(1 + [\xx_t]_e) \ge -1$ entrywise for $e \in F$. This is a contradiction since $\tau \ge \frac{360}{\eta^2}$ (indeed, we choose $\tau$ larger by a constant factor to account for inexactness in $\PMROalgo$ later), so $p_\tau \ge \frac 2 3$ as claimed.
The runtime follows from Lemma~\ref{lem:pmro}, as the number of runs of Lines~\ref{line:while_bfs_start} to~\ref{line:while_bfs_end} is $Z \sim \textup{Geom}(p)$ for $p \ge \half$.

\paragraph{Items~\ref{item:bfs_1},~\ref{item:bfs_2}, and~\ref{item:bfs_4}.} We have shown that with probability $\ge 1 - \frac \delta 4$, Lines~\ref{line:while_bfs_start} to~\ref{line:while_bfs_end} terminate after
\[k \defeq \log_2\Par{\frac 4 \delta}\]
loops. Conditional on this event and following our earlier notation, the probability of $\bigcup_{0 \le t \le \tau} \event_t$ all occurring in each of the at most $k$ loops is at least $1 - \frac \delta 4$ by our choice of $\eta$ and Lemma~\ref{lem:proj_sign_small}. Under these events (i.e.\ that there are at most $k$ loops and all $\norm{\xx_t}_\infty$ are small), Item~\ref{item:bfs_2} is immediate, since edges $e$ with $[\ww_t]_e \not\in [\ell[\wws]_e, 50[\wws]_e]$ are removed from consideration in a current iteration $t$, and no edge weight changes by more than a $1.1$ factor multiplicatively. Also, assuming \eqref{eq:pmro_exact}, Item~\ref{item:bfs_1} is also immediate (we will analyze the inexactness tolerance later).

We now prove Item~\ref{item:bfs_4}. For all $0 \le t \le \tau$, let $\vG_t \defeq (V, E, \ww_t)$ and let $G_t \defeq \und(\vG_t)$. We assumed that $\vHs$ was a $(\bw, \rho)$-cluster in $\vGs$, and no entry of $\ww_t$ restricted to $E(\vH_t) = F \setminus (S_t \cup L_t)$ is larger than $50\alpha\bw$ by definition of $L_t$, so
\[\Par{\max_{e \in E(\vH_t)} [\ww_t]_e} \cdot \Par{\max_{u, v \in V(\vH)} \ER_G(u, v)}  \le 75\alpha\rho \text{ for all } 0 \le t \le \tau.\]
Here we used that $\ER_G(u, v) \le 1.5 \ER_{\Gs}(u, v)$ for all $u, v$ by assumption. By applying Lemma~\ref{lemma:variance} for all iterations $0 \le t' \le t$ to the sequence of matrices $\Atil_e$ in \eqref{eq:edge_circ_def} for $e \in E(\vH_t)$, we inductively apply Lemma~\ref{lem:matrix_azuma} to show that with probability $1 - \frac{\delta t}{4\tau k}$, on any of the $k$ runs of Lines~\ref{line:while_bfs_start} to~\ref{line:while_bfs_end},
\[\normop{\LL_{G}^{\frac \dagger 2} \BB_{\vG}^\top \Par{\WW_t - \WW} \HH_{\vG} \LL_{G}^{\frac \dagger 2}} 
\le \frac{1}{20\csign\sqrt{\log \frac{60m\tau}{\delta}}}\cdot 4\sqrt{75\alpha\rho t\log\Par{\frac{8m\tau k}{\delta t}}} \le \frac{4}{\csign} \cdot \sqrt{\alpha\rho t}. \]
There are a few subtleties in the above calculation. First, observe that Lemma~\ref{lemma:variance} implies that if the $\Atil_e$ are defined with respect to $\PP_{\vH_t, \ww_t}$ rather than $\PP_{\vH_t}$ (as in Algorithm~\ref{alg:basic_fast_sparsify}), the variance bound still holds, because Lemma~\ref{lemma:varianceproj} applies to $\PP_{\vH_t, \ww_t}$ as well. Second, inductively using the guarantee above with Fact~\ref{fact:dirclose_undirclose} shows that $0.9 \LL_{G}\preceq \LL_{G_t} \preceq 1.1 \LL_G$ for all iterations $t$, where we used the assumption on $\alpha\rho$ for a large enough choice of $\cbfs$, so we adjusted the right-hand side by a constant factor. Third, note that the above argument holds with probability $\ge 1 - \frac \delta {4k}$ for each of the $\le k$ runs of Lines~\ref{line:while_bfs_start} to~\ref{line:while_bfs_end}, so it holds with probability $\ge 1 - \frac \delta 4$ for all of them by a union bound.

Finally, we need to condition on all $\event_t$ holding in all loops. We give a simple argument which removes this conditioning. If any $\event_t$ fails, we set all future weight updates to zero. Therefore, regardless of whether the $\event_t$ occur, the matrix variance \eqref{eq:matrix_variance_azuma} in our application of Lemma~\ref{lem:matrix_azuma} is bounded as we claimed. In particular, in an iteration $t$, as long as no $\event_s$ has occured for $0 \le s < t$, Lemma~\ref{lemma:variance} holds, and if any have occured, the variance is trivially bounded by $0$. 

The overall failure probability of $\le \delta$ comes from union bounding on the three events we have conditioned on so far (finishing in $k$ loops, all $\event_t$ holding in all loops, Item~\ref{item:bfs_4} holding), and the event that all of the $\le k\tau$ executions of Line~\ref{line:pmro} succeeed, which occurs with probability $\ge 1 - \frac \delta 4$.

\paragraph{Inexactness of projection.} It remains to discuss the effect of replacing our exact projections with our approximation through $\PMROalgo$. Because we ensured $\xi \le \frac \ell {10}$, the first bound in \eqref{eq:approx_problem} shows that entrywise $\xx_t$ is not affected by more than $\frac \ell {10}$ by approximation, so accounting for slack in our earlier argument Item~\ref{item:bfs_2} remains true. Next, using 
\[-\frac 1 3[\xx_t]_{e}^2 \le -\frac {1} {3.3} [ [\xxs]_t]_e^2 + 4[\xx_t - [\xxs]_t]_e^2 \le -\frac {1} {3.3} [\xxs]_t]_e^2 + 4\xi^2,  \]
we have by $\xi \le \frac{1}{1000\csign\log(\frac{60m\tau}{\delta})}$ that the approximation negligibly affects the argument in \eqref{eq:potdrop_exact}, which we accommodated in the constant factors in $\tau$, so it is still the case that Lines~\ref{line:while_bfs_start} to~\ref{line:while_bfs_end} terminate with probability $\ge \half$ in each loop. Regarding Item~\ref{item:bfs_1}, note that
\[\BB_{\vG}^\top \ww_t + \BB_{\vG}^\top \yy = \BB_{\vG}^\top \ww\]
in each iteration after applying the degree fixing in Line~\ref{line:degree_fixing}, so the invariant on degrees holds as claimed. The bound $\norm{\ww_t}_2 \le \sqrt{m}\norm{\ww_t}_\infty \le 120\sqrt{m} \bw$, combined with the last claim in \eqref{eq:approx_problem} and $\xi \le \frac{\eps}{200\sqrt{m}\tau}$, shows the $\ell_1$ norm of the weights cannot grow by more than $\eps \bw$ throughout. Moreover, the assumption $\xi \le \frac{\eps}{mn^3\tau}$ with the second guarantee in \eqref{eq:approx_problem} shows that in each iteration, the total degree imbalance $\norm{\dd}_1 \le \frac \eps {3mn^2\tau}$, and the error vector $\zz$ (in the context of Lemma~\ref{lemma:rounding}) satisfies $\norm{\zz}_1 \le m\xi \le \frac \eps {3n^2\tau}$. Lemma~\ref{lemma:rounding} then shows that $\norm{\yy}_1 \le m\norm{\yy}_\infty \le m\norm{\dd}_1 \le \frac \eps {3n\tau}$. The last two guarantees in Lemma~\ref{lemma:rounding} combined with the triangle inequality show that in each iteration, the additional spectral error due to approximate solves is $\frac{2\eps}{3\tau}$, and the additional error due to rounding is $\frac{\eps}{3\tau}$
giving the additional spectral error term in Item~\ref{item:bfs_4} after accumulating over all iterations. Finally, the runtime follows directly from Lemma~\ref{lemma:rounding} (for computing $\yy$), and Lemma~\ref{lem:pmro}.
\end{proof}

We provide one additional result which helps in disjoint applications of $\BFSalgo$.

\begin{corollary}\label{cor:multiple_cluster}
Consider calling \BFS $I$ times, with shared parameters $\vG, \wws, \ell, \delta, \eps$, but on edge-disjoint subgraphs $\{\vH_i\}_{i \in [I]}$ through $\vG$, so that the corresponding $[\vHs]_i$ are all $(\bw_i, \rho)$-clusters in $\vGs$ for some value of $\bw_i$. Then with probability $\ge 1 - \delta I$, the total operator norm error (i.e., Item~\ref{item:bfs_4}) incurred by all calls is bounded by \[\cbfs \cdot \sqrt{\rho \log\Par{\frac m \delta}} + \eps I.\]
\end{corollary}
\begin{proof}
The claim is that we do not incur an $I$ factor overhead in the operator norm error on the first term in the spectral error, and also do not incur an $I$ factor overhead on the $|V|$ term in the runtime. Note that the bound came from combining the variance bound in Lemma~\ref{lemma:variance} with the high-probability guarantee in Lemma~\ref{lem:matrix_azuma}. By treating each of the at most $\tau$ reweightings applied by Algorithm~\ref{alg:basic_fast_sparsify} in parallel across the edge-disjoint clusters, the combined variance in the sense of Lemma~\ref{lemma:variance}, where $\vH$ is set to the union of all clusters, is still bounded. The failure probability is by a union bound over $I$ calls. For the runtime, note that we can compute the degree imbalances in Line~\ref{line:dd_define} for all clusters simultaneously, and route them on $T$ in time $O(|V|)$ per iteration.
\end{proof}

\subsection{Sparsifying an ER decomposition}\label{ssec:twophase}

In this section, we state and analyze $\PSalgo$, which is a two-phase application (with different parameters) of $\BFSalgo$ to components of an ER decomposition.

\begin{algorithm2e}[ht!]\label{alg:partition_sparsify}
\caption{$\PSalgo(\{\vG_i\}_{i \in [I]}, \vG, T, \delta, \eps, W)$}
\DontPrintSemicolon
\codeInput $\{\vG^{(i)}\}_{i \in [I]}$, subgraphs of simple $\vG = (V, E, \ww)$ with $\max_{e \in \supp(\ww)} \ww_e \le W$, and such that $\{G^{(i)} \defeq \und(\vG^{(i)})\}_{i \in [I]}$ are a $(\rho, 2, J)$-ER decomposition of $G \defeq \und(\vG)$, $T$ a tree subgraph of $G$ with $\min_{e \in E(T)} \ww_e \ge 1$, $\delta, \eps \in (0, \frac 1 {100})$ \;
$m \gets E(\vG)$, $n \gets V(\vG)$, $R \gets \emptyset$\;
\For{$i \in [I]$}{
$\vH \gets \vG^{(i)}$, $\hm \gets |E(\vH)|$, $\hn \gets |V(\vH)|$, $\wws \gets \ww$\;
\If{$\hm \ge 40\hn$}{
$\ww_0 \gets \ww$, $\vG_0 \gets \vG$, $\vH_0 \gets \vH$, $\ell_1 \gets \frac{1}{2\log^{2}(\frac{nW}{\eps})}$, $\tau_1 \gets \log(\frac 2 {\ell_1})$\;
\For{$0 \le t < \tau_1$}{\label{line:while_ps_1_start}
$\ww_{t + 1} \gets \BFSalgo(\vH_t, \vG_t, \wws, \ell_1, \frac{\delta}{4I\tau_1}, \frac{\eps}{4I\tau_1}, E(\vH), T)$\;
$\vG_{t + 1} \gets (V, E, \ww_{t + 1})$, $\vH_{t + 1} \gets (V(\vH), E(\vH), [\ww_{t + 1}]_{E(\vH)})$
\label{line:while_ps_1_end}\;
}
$F \gets \{e \in E(\vH) \mid [\ww_t]_e \le \ell_1 [\wws]_e\}$\;\label{line:F_def}
$\ww_0 \gets \ww_t$, $\vG_0 \gets \vG_t$, $\vH_0 \gets \vH$, $\ell_2 \gets \frac{\eps}{4nmW}$, $\tau_2 \gets \log(\frac 2 {\ell_2})$\;
\For{$0 \le t < \tau_2$}{\label{line:while_ps_2_start}
$\ww_{t + 1} \gets \BFS(\vH_t, \vG_t, \wws, \ell_2, \frac{\delta}{4I\tau_2}, \frac{\eps}{4I\tau_2}, F, T)$\;
$\vG_{t + 1} \gets (V, E, \ww_{t + 1})$, $\vH_{t + 1} \gets (V(\vH), E(\vH), [\ww_{t + 1}]_{E(\vH)})$\;
\label{line:while_ps_2_end}
}
$R \gets R \cup \{e \in E(\vH) \mid [\ww_t]_e \le \frac \eps {4nm}\}$, $\ww \gets \ww_t$\;
}
}
\Return{$\vG' \gets (V, E, \ww_{E \setminus R} + $\RO{$\vG, \ww_{R}, T)$}}
\;
\end{algorithm2e}

We use the following scalar concentration inequality to bound the runtime with high probability.

 \begin{lemma}\label{lem:sum_of_geom}
	Let $\delta \in (0, 1)$, and let $\{Z_i\}_{i \in [I]} \subset \N$ be distributed as $Z_i \mid \{Z_j\}_{j < i} \sim \textup{Geom}(p_i)$ where $p_i \in [\half, 1]$ for all $i \in [I]$. Then for $S \defeq \sum_{i \in [I]} Z_i$,
	\[\Pr\Brack{S > 5\Par{I + \log\Par{\frac 1 \delta}}} \le \delta.\]
\end{lemma}
\begin{proof}
	It suffices to handle the case where $p_i = \half$ for all $i \in [I]$, since otherwise we can couple $Z_i$ to an instance of $\textup{Geom}(\half)$ which never exceeds $Z_i$. 
	Then we compute the moment generating function of $S$: for $\lam < \log(2)$, $\E \exp(\lambda S) = (\frac{\exp(\lambda )}{2 - \exp(\lambda)})^I$, so by Markov's inequality, for $t \defeq 5(I + \log_2(\frac 1 \delta))$,
	\[\Pr\Brack{S > t} < \exp\Par{-\lambda t}\Par{\frac{\exp(\lambda )}{2 - \exp(\lambda)}}^I = \Par{\frac 2 3}^t 3^I < \delta,\]
	where we use the choice $\lambda = \log(\frac 3 2)$ and substituted our choice of $t $.
\end{proof}

We now state our guarantee on Algorithm~\ref{alg:partition_sparsify} and provide its analysis.

\begin{lemma}\label{lem:ps_guarantee}
There is a universal constant $\cps$ such that if $\cps \cdot \rho\log(\frac {nW} {\delta\eps})\log^2\log(\frac{nW}{\eps}) \le 1$, \PS (Algorithm~\ref{alg:partition_sparsify}) returns $\vG' = (V, E, \ww')$ satisfying, with probability $\ge 1 - \delta$,
\begin{equation}\label{eq:ps_guarantee_new}
\begin{gathered}\BB_{\vG}^\top \ww' = \BB_{\vG}^\top \ww,\quad \nnz(\ww') \le \frac{31}{32} \nnz(\ww)+ \cps \cdot nJ,\\
\text{and }\normop{\LL_{G}^{\frac \dagger 2}\BB_{\vG}^\top\Par{\WW' - \WW}\HH_{\vG}\LL_G^{\frac \dagger 2}} \le \cps \sqrt{\rho\log\Par{\frac {nW} {\delta\eps}}}\log\log\Par{\frac{nW}{\eps}} + \eps.
\end{gathered}
\end{equation}
Moreover, $\max_{e \in E} \frac{\ww'_e}{\ww_e} \le \cps$. The runtime of $\PSalgo$ is
\[\bO\Par{|E|\log^2\Par{\frac{nW}{\delta\eps}}\log\Par{\frac{nW}{\eps}}}.\]
\end{lemma}
\begin{proof}
Throughout the proof, condition on all calls to $\BFSalgo$ succeeding assuming their input conditions are met (i.e., the guarantees in Lemma~\ref{lem:bfs_guarantee} hold, with total spectral error controlled by Corollary~\ref{cor:multiple_cluster}), which gives a failure probability of $\frac \delta 2$. We claim that every $\vG_t$ used in calls to $\BFSalgo$ satisfies $0.9\LL_{G_t} \preceq \LL_{G} \preceq 1.1\LL_{G_t}$, where $G \defeq \und(\vG)$ for $\vG$ the original input to the algorithm, and $G_t \defeq \und(\vG_t)$. We defer the proof of this claim to the end.

Next, fix $i \in [I]$ and consider the $\tau_1$ loops of Lines~\ref{line:while_ps_1_start} to~\ref{line:while_ps_1_end}. In all calls to $\BFSalgo$, the conditions on $\wws$ are met by assumption (i.e., each $\vG^{(i)}$ is an ER decomposition piece with parameters $(1.2\rho, 2)$ in $\vG_t$, since we claimed $0.9\LL_{G_t} \preceq \LL_{G} \preceq 1.1\LL_{G_t}$). Moreover, $\BFSalgo$ is only called if $\hm \ge 40\hn$, and the conditions in \eqref{eq:wws_reqs} are preserved inductively by Lemma~\ref{lem:bfs_guarantee}, since the $\ell_1$ norm of the weights does not change by more than a $\frac{\eps}{4\tau_1}$ factor in each iteration. This shows that the $\tau_1$ loops of Lines~\ref{line:while_ps_1_start} to~\ref{line:while_ps_1_end} all have their input conditions met, so we may assume they succeed. We claim that in this case, $F$ on Line~\ref{line:F_def} must have $|F| \ge \frac{\hm} 4$. To see this, suppose $|F| < \frac{\hm} 4$, which means the second part of Item~\ref{item:bfs_3} in Lemma~\ref{lem:bfs_guarantee} holds for all iterations $0 \le t < \tau_1$. However, since Lemma~\ref{lem:bfs_guarantee} also guarantees
\[\sum_{e \in E(\vH)} \log\Par{\frac{[\ww_\tau]_e}{\ww_e}} > -\hm\log\Par{\frac 2 \ell} = -\hm \tau_1,\]
we arrive at a contradiction after $\tau_1$ iterations, so the first part of Item~\ref{item:bfs_3} must have held at some point. With this size bound (showing $F$ is a valid input), an analogous argument shows that after the $\tau_2$ loops in Lines~\ref{line:while_ps_2_start} to~\ref{line:while_ps_2_end} have finished, at least $\frac{\hm}{16}$ edges are added to $R$. Observe that each component $\vG^{(i)}$ with $\hm_i$ edges and $\hn_i$ vertices either has $\frac 1 {16}$ of its edges added to $R$ or $\hm_i \le 40\hn_i$, and further $\sum_{i \in [I]} \hn_i \le nJ$. Since all edges from $R$ are zeroed out in the final weighting $\ww'$, and at most half the edges do not belong to any $\vG^{(i)}$, this gives the bound on $\nnz(\ww')$. Similarly, if all calls to $\BFSalgo$ succeed, since applying $\RO$ at the end of the algorithm preserves degrees, recursively applying Item~\ref{item:bfs_1} in Lemma~\ref{lem:bfs_guarantee} shows that $\BB_{\vG}^\top \ww' = \BB_{\vG}^\top \ww$.

It remains to show the spectral error bound. Observe that we have $\alpha = 2$ in the first $\tau_1$ calls to $\BFSalgo$ for each cluster (in Lines~\ref{line:while_ps_1_end} to~\ref{line:while_ps_1_end}), and $\alpha = \frac{1}{\log^2(\frac{nW}{\eps})}$ in the last $\tau_2$ calls (in Lines~\ref{line:while_ps_2_start} to~\ref{line:while_ps_2_end}). Therefore, taking note of Corollary~\ref{cor:multiple_cluster} and since $I \le m$, the spectral error in all intermediate iterations across all decomposition pieces is bounded by
\[O\Par{\sqrt{\rho \log\Par{\frac {m\tau_1} \delta}} \cdot \tau_1 + \sqrt{\frac{\rho}{\log^2\Par{\frac{nW}{\eps}}} \log\Par{\frac{m\tau_2}{\delta}}}\cdot \tau_2} = O\Par{\sqrt{\rho\log\Par{\frac {mW}{\delta\eps}}}\log\log\Par{\frac{nW}{\eps}} }.\]
Additionally, there is an $\frac{\eps}{4\tau_1 I} \cdot \tau_1 I + \frac{\eps}{4\tau_2 I} \cdot \tau_2 I$ additive error term which comes from Corollary~\ref{cor:multiple_cluster}, which is bounded by $\frac{2\eps}{3}$ after accounting for the change in the graph Laplacian (i.e., by Fact~\ref{fact:opnorm_precondition}).
For appropriate $\cps$, this both proves the desired spectral error bound by the triangle inequality, as well as the claimed $0.9\LL_{G_t} \preceq \LL_{G} \preceq 1.1\LL_{G_t}$ throughout the algorithm by Fact~\ref{fact:dirclose_undirclose}, which again implies that $G_t$ is connected under our assumption that $G$ is connected (see discussion in Section~\ref{sec:prelims}). Finally, applying $\ROalgo$ incurs at most $\frac \eps 3$ spectral error through the final graph by Lemma~\ref{lemma:rounding}, which is at most $\eps$ spectral error through the original graph by Fact~\ref{fact:opnorm_precondition}. The guarantee on the weight increase is clear as we only modify weights within clusters, and Item~\ref{item:bfs_2} of Lemma~\ref{lem:bfs_guarantee} shows no edge weight grows by more than a factor of $60$. This concludes the correctness proof.

For the runtime, the total number of times we call $\BFSalgo$ on each piece of the ER decomposition is $\tau_1 + \tau_2 = O(\log\frac{nW}{\eps})$. Thus, Lemma~\ref{lem:sum_of_geom} shows that with probability $\le \frac \delta {2}$, the number of times Lines~\ref{line:while_bfs_start} to~\ref{line:while_bfs_end} runs is $O(\log\frac{nW}{\delta\eps})$, for all decomposition pieces simultaneously. This gives the first term in the runtime via Lemma~\ref{lem:bfs_guarantee}, as all decomposition pieces have disjoint edges. For the second term in the runtime, it suffices to note that Lines~\ref{line:dd_define} to~\ref{line:degree_fixing} can be applied in parallel (after summing the degree imbalances $\dd$ in Line~\ref{line:dd_define}) for all decomposition pieces which terminate in a given run of Lines~\ref{line:while_bfs_start} to~\ref{line:while_bfs_end}, so we do not pay a multiplicative overhead of $|I|$ on the runtime of Lemma~\ref{lemma:rounding}. The total failure probability is via a union bound over Lemmas~\ref{lem:bfs_guarantee} and~\ref{lem:sum_of_geom}.
\end{proof}

\subsection{Complete sparsification algorithm}\label{ssec:outer_algo}

We now provide our complete near-linear time Eulerian sparsification algorithm. Our algorithm iteratively applies the ER decomposition from Proposition~\ref{prop:er_partition}, sparsifies the decomposition using Algorithm~\ref{alg:partition_sparsify}, and calls Algorithm~\ref{alg:round} on small-weight edges to maintain a bounded weight ratio. The following theorem gives a refined version of \Cref{thm:fastsparsify}.

\begin{algorithm2e}[ht!]
\caption{$\FSalgo(\vG, \epsilon, \delta)$}
\label{alg:fastsparse}
\DontPrintSemicolon
\codeInput Eulerian $\vG = (V, E, \ww)$ with $\ww_e \in [1, U]$ for all $e \in E$, $\eps, \delta \in (0, 1)$  \;
$n \gets |V|$, $m \gets |E|$\;
$T \gets $ arbitrary spanning tree of $G \defeq \und(\vG)$, $\hE \gets E \setminus E(T)$\;
$R \gets 6\log n$, $U_{\max} \gets U \cdot \cps^R$ for $\cps$ in Lemma~\ref{lem:ps_guarantee} \;
$t \gets 0$, $\ww_0 \gets \ww$\;
\While{$t < R$ $\mathbf{and}$ $\nnz([\ww_t]_{\hE}) > n\log(n)\log(\frac{32R^2mnU_{\max}}{\delta\eps})\log^2\log(\frac{32RmnU_{\max}}{\eps})\cdot \frac{2^{22}\cps^2}{\eps^2}$} {\label{line:while_start_fs}
$\vG_t \gets (V, E, \ww_t)$, $G_t \gets \und(\vG_t)$\;
$S \gets \ERPalgo([G_t]_{\hE}, 2, \frac \delta {2R})$  \Comment*{See \Cref{prop:er_partition}.}
$\vG'_{t} \defeq (V, E, \ww'_t) \gets \PSalgo(S, \vG_t, T, \frac{\delta}{2R}, \frac{\eps}{4R}, U_{\max})$ \;
$D \gets \{e \in \hE \mid [\ww'_t]_e \le \frac{\eps}{4mn}\}$\;\label{line:d_define}
$\ww_{t + 1} \gets [\ww'_{t}]_{E \setminus D} + \ROalgo(\vG'_{t}, [\ww'_{t}]_D, T)$\;\label{line:round_fs}
$t \gets t + 1$\;\label{line:while_end_fs}
}
\Return{$\vH \gets (V, E, \ww_t)$}\;
\end{algorithm2e}

\begin{restatable}{theorem}{restatefastsparsifydetailed}
    \label{thm:fastsparsifydetailed}
    Given Eulerian $\vG = (V, E, \ww)$ with $|V| = n$, $|E| = m$, $\ww \in [1, U]^E$ and $\eps, \delta \in (0, 1)$, $\FSalgo$ (Algorithm~\ref{alg:fastsparse}) returns Eulerian $\vH$ such that with probability $\ge 1 - \delta$, $\vH$ is an $\eps$-approximate Eulerian sparsifier of $\vG$, and
    \[
       |E(\vH)| = 
        O\Par{\frac{n}{\eps^2}\log(n)\log\Par{\frac{nU}{\delta}}\log^2\log\Par{nU}},\;
        \log\Par{\frac{\max_{e \in \supp(\ww')}\ww'_e}{\min_{e \in \supp(\ww')}\ww'_e}} = O\Par{\log\Par{nU}}.
        \]
    The runtime of $\FSalgo$ is $\bO\Par{m\log^2\Par{\frac{nU}{\delta}}\log\Par{nU}}$.
\end{restatable}
\begin{proof}
Throughout, condition on the event that all of the at most $R$ calls to $\ERPalgo$ and $\PSalgo$ succeed, which happens with probability $\ge 1 - \delta$. Because $\PSalgo$ guarantees that no weight grows by more than a $\cps$ factor in each call, $U_{\max}$ is a valid upper bound for the maximum weight of any edge throughout the algorithm's execution. Moreover, we explicitly delete any edge whose weight falls below $\frac{\eps}{4mn}$ throughout the algorithm in Line~\ref{line:d_define}, and these edges never appear in a call to $\ERPalgo$ again. Hence, $J_{\max} \defeq \log_2(\frac{32mnU_{\max}}{\eps})$ is a valid upper bound on the number of decomposition pieces ever returned by $\ERPalgo$, by Proposition~\ref{prop:er_partition}.

Next, note that under the given lower bound on $[\ww_t]_{\hE}$ in a given iteration (which is larger than $2\cps \cdot nJ_{\max}$), the sparsity progress guarantee in \eqref{eq:ps_guarantee_new} shows that the number of edges in each iteration is decreasing by at least a $\frac 1 {64}$ factor until termination. Since $m \le n^2$ and the algorithm terminates before reaching $n$ edges, $R$ is a valid upper bound on the number of iterations before the second condition in Line~\ref{line:while_start_fs} fails to hold, which gives the sparsity claim. Moreover, because the first term in the spectral error bound in \eqref{eq:ps_guarantee_new} decreases by a geometric factor of $1 - \frac 1 {256}$ in each round (as $\rho$ scales inversely in the current support size of $\ww_t$), the sum of all such terms contributes at most $256$ times the final contribution before termination. By plugging in the bound $\rho \le \frac{33n\log(n)}{m}$ from Proposition~\ref{prop:er_partition} with the lower bound on $m$ throughout the algorithm, the total contribution of these terms is at most $\frac \eps 4$. Similarly, the second additive term in \eqref{eq:ps_guarantee_new} contributes at most $\frac \eps 4$ throughout the $R$ rounds, and the rounding on Line~\ref{line:round_fs} also contributes at most $\frac \eps 4$ by Lemma~\ref{lemma:rounding}. Here we remark that once an edge is rounded on Line~\ref{line:round_fs}, it is removed from the support of $\ww_t$ for the rest of the algorithm. Adjusting these error terms by a $\frac 4 3$ factor (i.e., because of Fact~\ref{fact:dirclose_undirclose} which shows $\LL_{G_t}$ for $G_t \defeq \und(\vG_t)$ is stable throughout the algorithm, and Fact~\ref{fact:opnorm_precondition} which shows how this affects the error terms), we have the claimed spectral error guarantee. The sparsity bound follows again by explicitly removing any $e \in E$ where $[\ww_t]_e = 0$ from $\vH$. 

Finally, the runtime follows from combining Proposition~\ref{prop:er_partition} (which does not dominate), and Lemma~\ref{lem:ps_guarantee}. Here we note that we do not incur an extra logarithmic factor over Lemma~\ref{lem:ps_guarantee} because the edge count is a geometrically decreasing sequence (with constant ratio).
\end{proof}

\section{Applications}\label{sec:apps}
A direct consequence of our improved nearly-linear time Eulerian sparsifier in \Cref{thm:fastsparsifydetailed} is a significant improvement in the runtime of solving Eulerian Laplacian linear systems due to Peng and Song~\cite{PengS22}. In turn, combined with reductions in \cite{CohenKPPSV16}, our improved Eulerian system solver implies faster algorithms for a host of problems in directed graphs. We summarize these applications in this section. As a starting point, we state the reduction of~\cite{PengS22} from solving Eulerian Laplacian linear systems to sparsifying Eulerian graphs.
\newcommand{\cT}{\mathcal{T}}
\newcommand{\cS}{\mathcal{S}}
\begin{proposition}[Theorem~1.1, \cite{PengS22}]\label{prop:ps22}
Suppose there is an algorithm which takes in Eulerian $\vG = (V, E, \ww)$ with $n = |V|$, $m = |E|$, $\ww \in [1, U]^E$, and returns an $\eps'$-approximate Eulerian sparsifier with $\cS(n,U,\eps')$ edges with probability $\ge 1 - \delta$, in time $\cT(m,n,U,\eps', \delta)$. 
Then given Eulerian $\vG = (V, E, \ww)$ with $n = |V|$, $m = |E|$, $\ww \in [1, U]^E$, $\bb \in \R^V$, and error parameter $\eps \in (0,1)$, there is an algorithm running in time 
\begin{gather*}
O\Par{m\log\Par{\frac{nU}{\eps}} + \cT\Par{m, n, U, 1, \frac{\delta}{\log nU}}} \\
+ \bO\Par{\cT\Par{\cS\Par{n, U, 1}, n, U, 1, \frac{\delta}{\log nU}}\log(nU) + \cS(n, U, 1)\log(nU)\log\Par{\frac{nU} \eps}}
\end{gather*}
which returns $\xx \in \R^V$ satisfying, with probability $\ge 1 - \delta$,
\begin{equation}\label{eq:eulerian_accuracy}
\norm{\xx - \vLL_{\vG}^\dagger \bb}_{\LL_G} \leq \eps \norm{\vLL_{\vG}^\dagger \bb}_{\LL_G},\text{ where } G \defeq \und(\vG).
\end{equation}
\end{proposition}
Plugging \Cref{thm:fastsparsifydetailed} into Proposition~\ref{prop:ps22}, we obtain our faster solver for Eulerian Laplacians. The following corollary is a refined version of~\Cref{cor:eulsolver}.

\begin{restatable}[Eulerian Laplacian solver]{corollary}{restateeulsolverdetailed}
    \label{cor:eulsolverdetailed}
    Given Eulerian $\vG = (V,E,\ww)$ with $|V|=n, |E|=m, \ww \in [1,U]^E$, 
    $\bb \in \R^V$, and error parameter $\eps \in (0,1)$, 
    there is an algorithm running in time
    \[
    \bO\Par{m\log^2\Par{\frac{nU}{\delta}}\log\Par{\frac{nU}{\eps}} + n\log^2\Par{nU}\log^3\Par{\frac{nU}{\delta}}\log\Par{\frac{nU}{\eps}}}
    \]
    which returns $\xx \in \R^V$ satisfying, with probability $\ge 1 - \delta$,
    \[
    \norm{\xx - \vLL_{\vG}^\dagger \bb}_{\LL_G} \leq \eps \norm{\vLL_{\vG}^\dagger \bb}_{\LL_G},\text{ where } G \defeq \und(\vG).
    \]
\end{restatable}    
We remark that there is a more precise runtime improving upon Corollary~\ref{cor:eulsolverdetailed} in the logarithmic terms when $\delta, \eps$ are sufficiently small or $U$ is sufficiently large,
but we state the simpler variant for the following applications and for readability purposes.
Plugging our primitive in Corollary~\ref{cor:eulsolverdetailed} into black-box reductions from \cite{CohenKPPSV16} then gives algorithms to solve linear systems in row-or-column diagonally dominant matrices, which we now define.

\begin{definition}
We say $\MM \in \R^{n \times n}$ is row-column diagonally dominant (RCDD) if $\MM_{ii} \geq \sum_{j \neq i} |\MM_{ij}|$ and $\MM_{ii} \geq \sum_{j \neq i} |\MM_{ji}|$ for all $i \in [n]$. We say $\MM \in \R^{n \times n}$ is row-or-column diagonally dominant (ROCDD) if either $\MM_{ii} \geq \sum_{j \neq i} |\MM_{ij}|$ for all $i \in [n]$, or $\MM_{ii} \geq \sum_{j \neq i} |\MM_{ji}|$ for all $i \in [n]$.
\end{definition}

Most notably, Eulerian Laplacians are RCDD, and all directed Laplacians are ROCDD. In \cite{CohenKPPSV16} (see also \cite{AhmadinejadJSS19} for an alternative exposition), the following reduction was provided.

\begin{proposition}[Theorem 42, \cite{CohenKPPSV16}]\label{prop:ckppsv16}
Let $\MM \in \R^{n \times n}$ be ROCDD, and suppose both $\MM$ and its diagonal have multiplicative range at most $\kappa$ on their nonzero singular values. There is an algorithm which, given $\MM$, $\bb \in \textup{Im}(\MM)$, and error parameter $\eps \in (0, 1)$, solves $\log^2(\frac{n\kappa}{\eps})$ Eulerian linear systems to relative accuracy $\poly(\frac \eps {n\kappa})$ (in the sense of \eqref{eq:eulerian_accuracy}) and returns $\xx \in \R^n$ satisfying \begin{equation}\label{eq:rocdd_approx}\norm{\MM \xx - \bb}_2 \le \eps \norm{\bb}_2.\end{equation} 
Moreover, if $\MM$ is RCDD, a single such Eulerian linear system solve suffices.
\end{proposition}

Combining Corollary~\ref{cor:eulsolverdetailed}, Proposition~\ref{prop:ckppsv16}, and a union bound then yields the following.

\begin{corollary}[Directed Laplacian solver]\label{cor:dirlapsolver}
Given $\vG = (V,E,\ww)$ with $|V|=n, |E|=m, \ww \in [1,U]^E$, 
$\bb \in \R^V$, and error parameter $\eps \in (0,1)$, 
there is an algorithm running in time
\[
\bO\Par{m\log^2\Par{\frac{nU}{\delta\eps}}\log^3\Par{\frac{nU}{\eps}} + n\log^2\Par{nU}\log^3\Par{\frac{nU}{\delta\eps}}\log^3\Par{\frac{nU}{\eps}}}
\]
which returns $\xx \in \R^V$ satisfying, with probability $\ge 1 - \delta$,
\[
\norm{\xx - \vLL_{\vG}^\dagger \bb}_{\LL_G} \leq \eps \norm{\vLL_{\vG}^\dagger \bb}_{\LL_G},\text{ where } G \defeq \und(\vG).
\]
\end{corollary}

Finally, we mention a number of results from \cite{CohenKPPSV16,CohenKPPRSV17,AhmadinejadJSS19} which leverage RCDD solvers as a black box. Plugging Corollary~\ref{cor:eulsolverdetailed}, Proposition~\ref{prop:ckppsv16}, and Corollary~\ref{cor:dirlapsolver} into these results, we obtain the following runtimes. For simplicity, we only consider problems with $\poly(n)$-bounded conditioning and $\poly(\frac 1 n)$-bounded failure probability, and let $\tsolve(m, n, \eps) \defeq \bO(m\log^2(n)\log(\frac n \eps) + n\log^5(n)\log(\frac n \eps))$ be the runtime of our Eulerian Laplacian solver.

\begin{itemize}
\item \textbf{Stationary distributions.} We can compute a vector within $\ell_2$ distance $\eps$ of the stationary distribution of a random walk on a directed graph in time $\tsolve(m,n, 1) \cdot O(\log^2(\frac n \eps))$.
\item \textbf{Random walks.} We can compute the escape probability, hitting times and commute times for a random walk on a directed graph to $\eps$ additive error in time $\tsolve(m,n, 1) \cdot O(\log^2(\frac n \eps))$.
\item \textbf{Mixing time.} We can compute an $\eps$-multiplicative approximation of the mixing time of a random walk on a directed graph in time $\tsolve(m,n, 1)\cdot O(\log^2(\frac n \eps))$.
\item \textbf{PageRank.} We can compute a vector within $\ell_2$ distance $\eps$ of the Personalized PageRank vector with restart probability $\beta$ on a directed graph in time $\tsolve(m,n, 1)\cdot O(\log^2(\frac n {\beta})+\log(\frac 1 \eps))$. 
\item \textbf{M-matrix linear systems.} We can compute a vector achieving relative accuracy $\eps$ (in the sense of \eqref{eq:rocdd_approx}) to a linear system in an M-matrix $\MM$ in time  
\[\tsolve(m,n, \eps)\cdot O\Par{\log^2(n)\log\Par{\frac {\norms{\MM^{-1}}_{1 \to 1} + \norms{\MM^{-1}}_{\infty \to \infty}} \eps}}.\]
\item \textbf{Perron-Frobenius theory.} Given a nonnegative matrix $\AA \in \R^{n \times n}$ with $m$ nonzero entries, we can find $s \in \R$ and $\vv_l, \vv_r \in \R^n$ such that $\frac{s}{\rho(\AA)} \in [1, 1+\eps]$,\footnote{$\rho(\AA)$ is the spectral radius of $\AA$: $\rho(\AA) \defeq \lim_{k \rightarrow \infty} \normsop{\AA^k}^{1/k}$.} $\norm{\AA \vv_r - s \vv_r}_\infty \leq \eps \norm{\vv_r}_\infty$, and $\norms{\AA^\top \vv_l - s \vv_l}_\infty \leq \eps \norm{\vv_l}_\infty$, in time 
\[\tsolve(m,n, \eps) \cdot O\Par{\log^3\Par{\frac{\norm{\AA}_{1 \to 1} + \norm{\AA}_{\infty \to \infty}}{\eps \rho(\AA)}}}.\]
\end{itemize}

\section{Graphical spectral sketches}\label{sec:sketch}

In this section, we give an additional application of the techniques we developed for efficiently constructing Eulerian sparsifiers in Sections~\ref{sec:existence} and~\ref{sec:algos}. Specifically, we show that they yield improved constructions of the following graph-theoretic object, originally introduced in \cite{AndoniCKQWZ16, JambulapatiS18, ChuGPSSW18} in the undirected graph setting.

\begin{definition}[Graphical spectral sketch]\label{def:sketch_undir}
Given a undirected graph $G = (V,E,\ww)$, a distribution $\gH$ over random
undirected graphs $H = (V,E',\ww')$ with $E' \subseteq E$ is said to be a $(\eps, \delta)$-\emph{graphical spectral
sketch for $G$} if for any fixed vector $\xx \in \R^V$, with probability $\ge 1 - \delta$
over the sample $H \sim \gH$, we have 
\[
   \Abs{ \xx^\top (\LL_H - \LL_G) \xx} \le \eps \cdot \xx^\top \LL_G \xx.
\]
\end{definition}

We generalize Definition~\ref{def:sketch_undir} to the Eulerian graph setting (which to our knowledge has not been studied before), and show that our primitives extend to capture this generalization.

\begin{definition}[Eulerian graphical spectral sketch]\label{def:sketch_directed}
Given an Eulerian graph $\vG = (V,E,\ww)$, a distribution $\gH$ over random
Eulerian graphs $\vH = (V,E',\ww')$ with $E' \subseteq E$ is said to be a \emph{$(\eps, \delta)$-Eulerian graphical spectral
sketch for $\vG$} if for any fixed vectors $\aa,\zz \in \R^V$, with
probability $\ge 1 - \delta$ over the sample $\vH \sim \gH$, we have for $G
\defeq \und(\vG)$,
\begin{equation}\label{eq:dgss}
    \Abs{\aa^\top (\vLL_{\vH} - \vLL_{\vG}) \zz} 
    \le \eps \cdot \norm{\aa}_{\LL_G} \norm{\zz}_{\LL_G}.
\end{equation}
\end{definition}

Our algorithm closely follows the framework of \cite{JambulapatiS18,ChuGPSSW18}.
We aim to recursively reduce a constant fraction of the edges while keeping a small additive error
for a bilinear form applied to fixed vectors $\aa, \zz$, as in \eqref{eq:dgss}.
Similar to our spectral sparsification algorithm in Section~\ref{sec:algos}, we repeat this process 
for $O(\log n)$ phases. Within each phase, we accomplish our goal by first using an expander decomposition from prior work \cite{AgassyDK23}, and then within each piece, we restrict to a subgraph on vertices with sufficiently large combinatorial (unweighted) degrees.
At this point, Cheeger's inequality (\Cref{lemma:cheeger}) gives us an effective resistance diameter bound on the decomposition piece as well, so we can use most of the guarantees from Section~\ref{sec:algos} directly. We are able to obtain the tighter per-vector pair parameter tradeoff required by spectral sketches by exploiting a tighter connection between the Laplacian and degree matrices within expanders, as in \cite{JambulapatiS18} which used this for undirected spectral sketches. This is used alongside a key degree-based spectral inequality from \cite{ChuGPSSW18} (see Lemma~\ref{lemma:ex_deg}).

\subsection{Degree-preserving primitives}\label{ssec:bilift}

In this section, we give several basic helper results which we use to ensure degree-preserving properties of our algorithms by working with bipartite lifts. Given a directed graph $\vG = (V,E,\ww)$, we let the directed graph
$\vG^\uparrow \defeq \blift(\vG)$ be its bipartite lift, which is defined so that $V_{\vG^\uparrow}
= V \cup \bar{V}$ where $\bar{V}$ is a copy of $V$, and $E_{\vG^\uparrow} = \{f =
(u,\bar{v}) | (u,v) \in E\}$ with $\ww_{u,\bar{v}} = \ww_{(u,v)}$.
Notice that our definition gives a canonical bijection between $E_{\vG^\uparrow}$ and $E_{\vG}$.

\begin{lemma} \label{lemma:blift_spectral}
Let $\vG = (V,E,\ww)$ be a directed graph and let its bipartite lift be
$\vG^\uparrow \defeq (V\cup V',E^\uparrow,\ww) = \blift(\vG)$, with 
$G \defeq \und(\vG)$ and $G^\uparrow \defeq \und(\vG^\uparrow)$. 
Suppose that for some $\eps > 0$, $\ww' \in \R^E_{>0}$ satisfies
\[
    \BB_{\vG^\uparrow}^\top \ww' = \BB_{\vG^\uparrow}^\top \ww, \quad
    \normop{\LL_{G^\uparrow}^{\frac \dagger 2} \BB_{\vG^\uparrow}^\top (\WW' - \WW) 
    \HH_{\vG^\uparrow} \LL_{G^\uparrow}^{\frac \dagger 2}} 
    \le \eps.
\]
Then, letting $|\BB_{\vG}|$ apply the absolute value entrywise,
\[
    \BB_{\vG}^\top \ww' = \BB_{\vG}^\top \ww,\quad |\BB_{\vG}|^\top \ww' = |\BB_{\vG}|^\top \ww, \quad
    \normop{\LL_G^{\frac \dagger 2} \BB_{\vG}^\top (\WW' - \WW) \HH_{\vG} 
    \LL_G^{\frac \dagger 2}} 
    \le \eps.
\]
\end{lemma}
\begin{proof}
Consider the edge-vertex incidence matrix of $\vG$, $\BB_{\vG} = \HH_{\vG} - \TT_{\vG}$.
The edge-vertex incidence matrix of $\vG^\uparrow$ is then
$\BB_{\vG^\uparrow} = \begin{pmatrix} \HH_{\vG} & -\TT_{\vG} \end{pmatrix}$.
Hence, any vector $\xx$ satisfying $\BB_{\vG^\uparrow}^\top \xx = \vzero_V$ must have
$\HH_{\vG}^\top \xx = \vzero_V, \TT_{\vG}^\top \xx = \vzero_V$, giving us
preservation of both the difference between in and out degrees and the
sum of in and out degrees $\vG$, i.e.,
\[
    \BB_{\vG}^\top \xx = \vzero_V,\; |\BB_{\vG}|^\top \xx = \vzero_V,
\]
where we used that $\BB_{\vG} = \HH_{\vG} - \TT_{\vG}$ and $|\BB_{\vG}| = \HH_{\vG} + \TT_{\vG}$.
Taking $\xx \defeq \ww' - \ww$ then gives the the first two claims.
We remark that the directed graph $\vG$ need not be Eulerian.

We proceed to prove the third claim. For ease of notation, we omit $\vG$ in the subscripts of the matrices and
denote $\AA \defeq \AA_{\vG}$, $\AA_\uparrow \defeq \AA_{\vG^\uparrow}$ for all
matrices $\AA$.
We also use the following equivalent definition of operator norms with the convention
that the fraction is 0 if the numerator is 0:
\[
    \normop{\AA} = 
    \max_{\xx,\yy} \frac{|\xx^\top \AA \yy|}{\norm{\xx}_2
    \norm{\yy}_2},
\]
where the $\max$ is over $\xx, \yy$ of compatible dimensions.
Also, we let $\QQ \in \R^{(V \cup \bV) \times V}$ be defined by 
\[\QQ\ee_v = \begin{pmatrix} \ee_v \\ \ee_{\bv} \end{pmatrix} \]
for all $v \in V$, where $\bv$ is identified with $v$. 
Notice that $\HH = \HH_\uparrow\QQ$, $\TT = \TT_\uparrow\QQ$ and $\BB = \BB_\uparrow\QQ$.
Then,
\[
    \QQ^\top \LL_\uparrow \QQ
    =
    \QQ^\top \BB_\uparrow^\top \WW \BB_\uparrow \QQ
    =
    \BB^\top \WW \BB
    =
    \LL.
\]
Finally, for any non-trivial vectors $\xx,\yy \in \R^V$ satisfying $\xx,\yy \perp
\vone_V$, and defining 
\[\aa \defeq \LL^{\frac \dagger 2}\xx,\; 
\bb \defeq \LL^{\frac \dagger 2}\yy,\; \hxx \defeq \LL_\uparrow^{\frac \dagger 2}\QQ \aa,\; 
\hyy \defeq \LL_\uparrow^{\frac \dagger 2}\QQ \bb,\]
 we have
\begin{align*}
  \frac{|\xx^\top \LL^{\frac \dagger 2} \BB^\top (\WW'-\WW) \HH 
    \LL^{\frac \dagger 2} \yy|}{\norm{\xx}_2 \norm{\yy}_2} &= 
    \frac{|\aa^\top \BB^\top (\WW'-\WW) \HH \bb|}{\norm{\aa}_{\LL}
    \norm{\bb}_{\LL}}\\
    &=
    \frac{|\aa^\top \QQ^\top \BB_\uparrow^\top (\WW'-\WW) \HH_\uparrow \QQ
    \bb|}{\norm{\QQ \aa}_{\LL_\uparrow} \norm{\QQ \bb}_{\LL_\uparrow}} \\
    &=
    \frac{|\hxx^\top \LL_\uparrow^{\frac \dagger 2} \BB_\uparrow^\top 
        (\WW'-\WW) \HH_\uparrow \LL_\uparrow^{\frac \dagger 2} \hyy|}
        {\norm{\hxx}_2 \norm{\hyy}_2},
\end{align*}
giving us the desired operator norm bound.
We note that $\QQ \aa,\QQ \bb \perp \ker(\LL_\uparrow)$.

Finally, we record a consequence of this proof we will later use. Recall that we have shown
\[\HH = \HH_{\uparrow}\QQ,\quad
\TT = \TT_{\uparrow}\QQ,\quad
\BB = \BB_{\uparrow}\QQ,\quad
\QQ^\top \LL_{\uparrow} \QQ = \LL.\]
Therefore, suppose that for some $\eps > 0$ and fixed vectors $\xx,\yy \in \R^V$,  $\ww' \in \R^E_{>0}$ satisfies
\begin{equation}\label{eq:lift_imply_1}
\Abs{\xx^\top \QQ^\top \BB_{\uparrow}^\top (\WW' - \WW) \HH_{\uparrow} \QQ \yy} \le
\eps \cdot \norm{\QQ\xx}_{\LL_{\uparrow}} \norm{\QQ\yy}_{\LL_{\uparrow}}.
\end{equation}
Then, we also have the bound in the unlifted graph $G$,
\begin{equation}\label{eq:lift_imply_2}
 \Abs{\xx^\top \BB^\top (\WW' - \WW) \HH \yy} \le
\eps \cdot \norm{\xx}_{\LL} \norm{\yy}_{\LL}.
\end{equation}
\end{proof}
 
\subsection{Expander decomposition and sketching by degrees}

In this section, we provide guarantees on our earlier $\BFSalgo$ (Algorithm~\ref{alg:basic_fast_sparsify}) which hold when the algorithm is passed a expander graph, that is also a bipartite lift, as input. We first recall the definition of an expander graph, parameterized by a minimum conductance threshold $\phi$.

\begin{definition}[Expander graphs]\label{def:expander}
    Let $G = (V,E,\ww)$ be an undirected graph, and let $\dd \in \R_{\geq 0}^V$ be its weighted
    degrees.
    For a set $S \subseteq V$, let $\vol(S) = \sum_{v \in S} \dd_v$ be the sum of
    weighted degrees in $S$, and let $\partial S = \{e = (u, v)\in E \mid u \in S \text{ and }
    v \not\in S\}$ be the edge boundary of $S$.
    We define the cut value of $S$ by $\ww(\partial S) \defeq \sum_{e \in \partial S} \ww_e$, and
    the conductance of $S$ by
    \[
        \Phi(S) = \frac{\ww(\partial S)}{\min\{\vol(S), \vol(V\backslash S)\}}.
    \]
    Finally, we say $G$ is a \emph{$\phi$-expander} if $\Phi(S) \geq \phi$ for all 
    $S \subseteq V$.
\end{definition}

An important algorithmic primitive related to Definition~\ref{def:expander} is an expander decomposition.

\begin{definition}[Expander decomposition]\label{def:exp_partition}
We call $\{G_i\}_{i \in [I]}$ a \emph{$(\phi, r, J)$-expander decomposition} if $\{G_i\}_{i \in [I]}$ are edge-disjoint subgraphs of $G = (V, E, \ww)$, and the following hold.
\begin{enumerate}
    \item \label{item:exp:partition:weight} \emph{Bounded weight ratio}: For all $i \in [I]$, $\frac{\max_{e \in E(G_i)}\ww_e}{\min_{e \in E(G_i)}\ww_e} \le r$.
    \item \label{item:exp:partition:phi} \emph{Conductance}: For all $i \in [I]$, $\Phi(G_i) \ge \phi$.
    \item \label{item:exp:partition:cut} \emph{Edges cut}: $|E(G) \setminus ( \bigcup_{i \in [I]} E(G_i) )| \leq \frac m 2$.
    \item \label{item:exp:partition:vertex} \emph{Vertex coverage}: Every vertex $v \in V(G)$ appears in at most $J$ of the subgraphs.
\end{enumerate}
\end{definition}

We recall the state-of-the-art expander decomposition algorithm in the literature, which will later be used in conjunction with the subroutines developed in this section.
\begin{proposition}[Theorem 4.4, \cite{AgassyDK23}]\label{lemma:ex_partition}
    There is an algorithm $\EPalgo(G,r,\delta)$ that, given as input undirected $G =
    (V,E,\ww)$ with $\frac{\max_{e\in \supp(\ww)}{\ww_e}}{\min_{e\in
    \supp(\ww)}{\ww_e}} \le W$ and $r \geq 1$, computes in time 
\[O\Par{m \log^6(n)\log\Par{\frac n \delta}}\]
    a $(\cadk\log^{-2}(n), r, \log_r W + 3)$-expander decomposition of $G$ with
    probability $\ge 1-\delta$, for a universal constant $\cadk$.\footnote{The algorithm of \cite{AgassyDK23} is stated for $n^{-O(1)}$ failure probabilities, but examining Section 5.3, the only place where randomness is used in the argument, shows that we can obtain failure probability $\delta$ at the stated overhead. The vertex coverage parameter $\log_r(W) + 3$ is due to bucketing the edges by weight, analogously to the proof of Proposition~\ref{prop:er_partition}. We note that there is no $\log_r(W)$ overhead in the runtime, as the edges in each piece are disjoint.}
\end{proposition}

We further require two spectral inequalities based on the expansion.

\begin{lemma}[Cheeger's inequality] \label{lemma:cheeger}
If $G$ is a $\phi$-expander with Laplacian $\LL_G$ and degrees $\DD_G$, then
\[
    \lambda_2\Par{\DD_G^{\frac \dagger 2} \LL_G \DD_G^{\frac \dagger 2}} 
    \geq \frac {\phi^2} 2.
\]
\end{lemma}

\begin{lemma}[Lemma 6.6, \cite{ChuGPSSW18}] \label{lemma:ex_deg}
If $G = (V,E,\ww)$ is a $\phi$-expander that satisfies, for some $\bw > 0$, $\ww_e
\in [\bw,2\bw]$ for all $e \in E$, then for any $\xx \in \R^V$,
\[
    \norm{\xx}_{\LL_G}^2  \ge 
    \frac {\phi^2 \bw} 2 \cdot \sum_{v \in V} [\deg_{G}]_v (\xx_v - \hx)^2,
\]
where $\deg_G \in \N^V_{\ge 0}$ is the combinatorial (unweighted) degrees of $G$, and 
$\hx \defeq \frac{\dd_G^\top\xx}{\norms{\dd_G}_1}$.
\end{lemma}

Importantly, Lemma~\ref{lemma:ex_deg} allows us to obtain improved tradeoffs on how well spectral sketch guarantees are preserved in Lemma~\ref{lemma:ess_guarantee}, by first lower bounding degrees of vertices under consideration. Next, we show that expander graphs with small weight ratio form clusters (Definition~\ref{def:cluster}), which make them compatible with our algorithm $\BFSalgo$.

\begin{lemma} \label{lemma:ex_to_er}
Let $\vG = (V,E,\ww)$ and let $G = \und(\vG)$. Suppose $G$ is a $\phi$-expander and
that for all $e \in E$, $\ww_e \in [\bw,2\bw]$ for some $\bw > 0$.
Given any $\beta > 0$, let $U \subseteq V$ be the set of vertices with
$[\deg_G]_u \ge \beta$ for every $u \in U$.
Then, the subgraph $\vG[U]$ is a $(\bw,8\beta^{-1}\phi^{-2})$-cluster in $\vG$.
\end{lemma}
\begin{proof}
Let $\DD_G \in \R^{V \times V}$ be the diagonal matrix whose diagonal is the weighted degrees of $G$.
By Cheeger's inequality (\Cref{lemma:cheeger}), for any pair of distinct
vertices $a,b \in U$, we have the desired
\begin{align*}
    \ww_e \ER_G(a, b) &\le 2\bw \bb_{a,b}^\top \LL_G^{\dagger} \bb_{a,b}
    \le
    4\bw \phi^{-2} \bb_{(a,b)}^\top \DD_G^{-1} \bb_{(a,b)}
    \\
    &\le
    4 \phi^{-2} \Par{\frac{1}{[\deg_G]_a} + \frac{1}{[\deg_G]_b}}
    \le
    8 \phi^{-2} \beta^{-1}.
\end{align*}
\end{proof}

Finally, we state one additional sketching property enjoyed by $\BFSalgo$ in Lemma~\ref{lemma:bfs_fix}. We mention that this is the key step in our proof where we require that our input graph is a bipartite lift of a directed graph. We exploit this property by employing machinery from Section~\ref{ssec:bilift}.

\begin{lemma}\label{lemma:bfs_fix}
Suppose \BFSalgo~is given input $\vC$ instead of $\vG$ where $\vC$ is a subgraph of
$\vG=(V,E,\ww)$, a bipartite lift of a directed graph, satisfying $V = A \cup B$ and $E \subseteq A \times B$, 
$\vHs$ is a $(\bw,\rho)$-cluster of 
$\vCs = (V(\vC),E(\vC),[\wws]_{E(\vC)})$, and $T$ is a subgraph of $G \defeq
\und(\vG)$ and $E(T) \cap E(C) = \emptyset$, where $C \defeq \und(\vC)$.
Under the same assumptions as \Cref{lem:bfs_guarantee},
with probability $\ge 1-\delta$,
\Cref{item:bfs_1,item:bfs_2,item:bfs_3,item:bfs_4} and the runtime of
\Cref{lem:bfs_guarantee} still hold.

In addition, $|\BB_{\vG}| \ww = |\BB_{\vG}| \ww'$, and 
for fixed $\aa,\zz \in \R^V$,
\begin{equation}\label{eq:error_per_vecs}
    \Abs{\aa^\top \BB_{\vG}^\top (\WW'-\WW) \HH_{\vG} \zz}
    \le
    \cbfs \cdot \alpha \bw \sqrt{\log\Par{\frac 1 \delta}} 
    \norm{\aa}_{\PPi_{V(\vC)}} \norm{\zz}_{\PPi_V(\vC)}
    + \eps \cdot \norm{\aa}_{\LL_G} \norm{\zz}_{\LL_G}.
\end{equation}
\end{lemma}
\begin{proof}
\Cref{item:bfs_2,item:bfs_3} and the second claim in \Cref{item:bfs_1} of
\Cref{lem:bfs_guarantee} are not affected by the change in input.
Further, the first claim in \Cref{item:bfs_1} follows by \Cref{lemma:rounding}, where the relevant
edge-vertex incidence matrix remains $\BB_{\vG}$ (since $T$ is a subgraph of
$G$).
By the same argument in the proof of the second equation in
\Cref{lemma:blift_spectral}, the assumption that $\vG$ is a bipartite lift gives
$|\BB_{\vG}|\ww = |\BB_{\vG}|\ww'$.
Thus, it suffices to discuss Item~\ref{item:bfs_4} and \eqref{eq:error_per_vecs}.

We now prove Item~\ref{item:bfs_4}. Note that the key difference is that we assume $\vH_\star$ is a $(\bw, \rho)$-cluster, when ERs are measured through $C$ instead of $G$. Under the assumption that \PMRO~is exact, the same argument as in the proof of
\Cref{lem:bfs_guarantee} gives that
\begin{equation} \label{eq:bfs_fix_spec}
    \normop{\LL_C^{\frac \dagger 2} \BB_{\vC}^\top (\WW'-\WW) \HH_{\vC}
    \LL_C^{\frac \dagger 2}} 
    \le
    \cbfs \sqrt{\alpha \rho \log\Par{\frac m \delta}}.
\end{equation}
Since $\supp(\ww' - \ww) \subseteq E(\vC)$ and $\LL_C \pleq \LL_G$, we obtain the
first term in the inequality in \Cref{item:bfs_4} of \Cref{lem:bfs_guarantee} (i.e., without the additive $\eps$) under an exact $\PMROalgo$. The error due to inexactness is then handled the same way as in \Cref{item:bfs_4} of \Cref{lem:bfs_guarantee}, since the spectral error guarantees of Lemma~\ref{lemma:rounding} are measured with respect to $G$, not $C$.

In the remainder of the proof, we handle \eqref{eq:error_per_vecs}. We first consider the sketching error assuming \PMRO~is exact.
In iteration $t$, the difference to the directed Laplacian is:
\begin{equation}\label{eq:dir_lap_diff}
\begin{aligned}
    \BB_{\vC}^\top (\WW_{t+1} - \WW_t) \HH_{\vC}
    &=
    \eta \sum_{e \in E(\vH_t)} [\ww_t]_e \Brack{\PP_{\vH_t,\ww_t} \ss}_e \bb_e \hh_e^\top \\
    &=
    \eta \sum_{e \in E(\vH_t)} \ss_e \sum_{f \in E(\vH_t)} 
    [\PP_{\vH_t,\ww_t}]_{fe} [\ww_t]_f \bb_f \hh_f^\top.
\end{aligned}
\end{equation}
By the third equality of \Cref{lem:rankone_fix}, $\BB_{\vC}^\top \WW_t
[\PP_{\vH_t,\ww_t}]_{:e} = \vzero_E$ for any $e \in E(\vH_t)$, i.e.,
$\WW_t[\PP_{\vH_t,\ww_t}]_{:e}$ is a circulation on the graph $\vC$.
Again, by the same argument in \Cref{lemma:blift_spectral}, we have
\begin{equation}\label{eq:column_circ}|\BB_{\vC}|^\top\WW_t[\PP_{\vH_t,\ww_t}]_{:e} = \vzero_V,\end{equation}
i.e., both in-degrees and
out-degrees are preserved.

Next, because all edges in $\vC$ are from $A$ to $B$, we have
$[|\BB_{\vC}|]_{:A} = \HH_{\vC}$.
Then, we compute that
\begin{equation}\label{eq:heads_diag}
	\begin{aligned}
    \sum_{f \in E(\vH_t)} [\PP_{\vH_t,\ww_t}]_{fe} \ww_f \hh_f \hh_f^\top &=
    [|\BB_{\vC}|]_{:A}^\top \WW_t \diag{[\PP_{\vH_t,\ww_t}]_{:e}} \HH_{\vC} \\
    &= \diag{[|\BB_{\vC}|]_{:A}^\top \WW_t [\PP_{\vH_t,\ww_t}]_{:e}}.
    \end{aligned}
\end{equation}
Combining \eqref{eq:column_circ} and \eqref{eq:heads_diag} shows:
\begin{equation}\label{eq:cancel_heads}
    \sum_{f \in E(\vH_t)} [\PP_{\vH_t,\ww_t}]_{fe} \ww_f \bb_f \hh_f^\top =
    \sum_{f \in E(\vH_t)} [\PP_{\vH_t,\ww_t}]_{fe} \ww_f \tt_f \hh_f^\top.
\end{equation}
In addition, we have by \Cref{lemma:circeq} and that $\WW_t [\PP_{\vH_t,\ww_t}]_{:e}$ is a circulation that the following hold:
\begin{equation}\label{eq:each_edge_kernel}
\begin{aligned}
    \BB_{\vC}^\top \WW_t \diag{[\PP_{\vH_t,\ww_t}]_{:e}} \HH_{\vC} \vone_V &= \BB_{\vC}^\top \WW_t [\PP_{\vH_t,\ww_t}]_{:e} = \vzero_V,
    \\
   \vone_V^\top \BB_{\vC}^\top \WW_t \diag{[\PP_{\vH_t,\ww_t}]_{:e}} \HH_{\vC}
    &=
    -\vone_V^\top \TT_{\vC}^\top \WW_t \diag{[\PP_{\vH_t,\ww_t}]_{:e}}
    \BB_{\vC} \\
    &= -[\PP_{\vH_t,\ww_t}]_{:e}^\top \WW_t \BB_{\vC}= \vzero_V^\top.
\end{aligned}
\end{equation}
Recalling the formula \eqref{eq:dir_lap_diff}, and summing \eqref{eq:each_edge_kernel} over all $e \in E(\vH_t)$, gives
\begin{equation}\label{eq:ones_in_kernel}
    \BB_{\vC}^\top (\WW_{t+1}-\WW_t) \HH_{\vC} \vone_V = \sum_{e \in E(\vH_t)} \ss_e \vzero_V = \vzero_V, \quad
    \vone_V^\top \BB_{\vC}^\top (\WW_{t+1}-\WW_t) \HH_{\vC} = \vzero_V^\top.
\end{equation}
Define for each $e \in E(\vH_t)$ a scalar $\xx_e^{(t)}$ by:
\[
    \xx_e^{(t)} \defeq 
    [\ww_t]_e \aa^\top \PPi_{V(\vC)} \tt_e \hh_e^\top \PPi_{V(\vC)} \zz,
\]
where we recall that $\PPi_{V(\vC)} = \II_{V(\vC)} -
\frac{1}{|V(\vC)|}\vone_{V(\vC)}\vone^\top_{V(\vC)}$. We showed in \eqref{eq:ones_in_kernel} that $\vone_{V}$, and hence $\vone_{V(\vC)}$, is in the left and right kernel of $\BB_{\vC}^\top (\WW_{t+1}-\WW_t) \HH_{\vC} $. Combining \eqref{eq:dir_lap_diff} and \eqref{eq:cancel_heads} then yields
\begin{equation}\label{eq:difference_az}
\begin{aligned}
    \aa^\top \BB_{\vC}^\top (\WW_{t+1} - \WW_t) \HH_{\vC} \zz &=
    \aa^\top \PPi_{V(\vC)} \BB_{\vC}^\top (\WW_{t+1} - \WW_t) \HH_{\vC} \PPi_{V(\vC)} \zz \\
    &=
    \eta \sum_{e \in E(\vH_t)} \ss_e \sum_{f \in E(\vH_t)} 
    [\PP_{\vH_t,\ww_t}]_{fe} \xx_f^{(t)}.
\end{aligned}
\end{equation}
Each $\eta \sum_{f \in E(\vH_t)} [\PP_{\vH_t,\ww_t}]_{fe}
\xx_f \cdot \ss_e$ is sub-Gaussian with parameter $\sigma_e \defeq \eta |\inprods{[\PP_{\vH_t,\ww_t}]_{:e}}{\xx^{(t)}}|$.
Therefore, the left-hand side of \eqref{eq:difference_az} is sub-Gaussian with parameter
\[
    \eta^2 \sum_{e \in E(\vH_t)} \inprod{[\PP_{\vH_t,\ww_t}]_{:e}}{\xx^{(t)}}^2
    =
    \eta^2 \xx^\top \PP_{\vH_t,\ww_t} \xx^{(t)}
    \le
    \eta^2 \norm{\xx^{(t)}}_2^2,
\]
where $\norm{\xx^{(t)}}_2^2$ can be bounded, using the definition of $\alpha$
and \Cref{item:bfs_1}, by
\begin{align*}
    \norm{\xx^{(t)}}_2^2 
    &= 
    \sum_{e \in E(\vH_t)} \ww_e^2 [\PPi_{V(\vC)}\aa]_{t(e)}^2 [\PPi_{V(\vC)}\zz]_{h(e)}^2 \\
    &\le 8\alpha^2 \bw^2 \Par{\sum_{v \in V(\vC)} [\PPi_{V(\vC)}\aa]_{v}^2 } 
    \Par{\sum_{u \in V(\vC)} [\PPi_{V(\vC)}\zz]_{u}^2}
    \\
    &= 
    8\alpha^2 \bw^2 \norm{\aa}_{\PPi_{V(\vC)}}^2 \norm{\zz}_{\PPi_{V(\vC)}}^2.
\end{align*}
Summing over at most $\tau$ iterations, the total sub-Gaussian parameter of $\aa^\top \BB_{\vC}^\top (\WW' - \WW) \HH_{\vC} \zz$ is:
\[
    \tau \cdot 8\alpha^2 \bw^2 \eta^2 \cdot \norm{\aa}_{\PPi_{V(\vC)}}^2 \norm{\zz}_{\PPi_{V(\vC)}}^2
    \le
    6000 \alpha^2 \bw^2 \cdot \norm{\aa}_{\PPi_{V(\vC)}}^2 \norm{\zz}_{\PPi_{V(\vC)}}^2.
\]
Standard sub-Gaussian concentration finally yields, with probability $1-\frac{\delta}{4}$, the desired
\[
    \Abs{\aa^\top \BB_{\vG}^\top (\WW'-\WW) \HH_{\vG} \zz}
    \le
    \cbfs \cdot \alpha \bw \sqrt{\log\Par{\frac 1 \delta}} 
    \norm{\aa}_{\PPi_{V(\vC)}}^2 \norm{\zz}_{\PPi_{V(\vC)}}^2.
\]
Following the notation in Lemma~\ref{lem:bfs_guarantee}, conditioning on the events $\bigcup_{0 \le t \le \tau} \event_t$ does not affect the proof, for the same reason as outlined in Lemma~\ref{lem:bfs_guarantee}: if any $\event_t$ fails, we set all future weight updates to zero in the scalar martingale. 
Finally, as $T$ is edge-disjoint from $C$, the additive spectral error term due to the inexactness of \PMRO~and the final rounding in each iteration is measured with respect
to $\LL_G$, as is done in Lemma~\ref{lem:pmro}.
Applying this additive spectral error to vectors $\aa$ and $\zz$ gives an
additive term of $\eps \cdot \norm{\aa}_{\LL_G} \norm{\zz}_{\LL_G}$.
This completes our proof.
\end{proof}

\begin{corollary}\label{cor:multiple_cluster_fix}
Consider calling \BFS $I$ times, all with shared parameters $\wws, \ell, \delta,
\eps$, but on edge-disjoint subgraphs $\{\vH_i\}_{i \in [I]}$ and $\{\vC_i\}_{i
\in [I]}$ of $\vG=(V,E,\ww)$, a bipartite lift of a directed graph, with $V = A \cup B$ and $E \subseteq A \times
B$, so that each corresponding $[\vHs]_i$ is a
$(\bw_i, \rho)$-cluster in $\vC_i$ for some value of $\bw_i$.
Then with probability $\ge 1 - \delta I$, the runtime and spectral error
guarantee in \Cref{cor:multiple_cluster} still hold.
In addition, $|\BB_{\vG}|^\top\ww = |\BB_{\vG}|^\top \ww'$, and 
for fixed $\aa,\zz \in \R^V$, for all $i \in I$,
\begin{align*}
    \Abs{\aa^\top \BB_{\vG}^\top (\WW'-\WW) \HH_{\vG} \zz}
    \le&\
    \cbfs \cdot \alpha \sqrt{\log\Par{\frac 1 \delta}} 
    \sum_{i \in [I]} \bw_i \norm{\aa}_{\PPi_{V(\vC_i)}}
    \norm{\zz}_{\PPi_{V(\vC_i)}} + \eps I\cdot \norm{\aa}_{\LL_G} \norm{\zz}_{\LL_G},
\end{align*}
where $G \defeq \und(G)$ and for each $i \in [I]$, $C_i \defeq \und(C_i)$.
\end{corollary}
\begin{proof}
This claim follows by an analogous argument as in the proof of
\Cref{cor:multiple_cluster}, where we use the sketching error claim \eqref{eq:error_per_vecs} from
\Cref{lemma:bfs_fix} summed across each subgraph.
\end{proof}

We are now ready to give the main algorithm of this section, as well as its analysis. To clarify the role of the expander decomposition (and degree-based spectral bounds), we note that the sketching guarantee provided to a fixed pair of vectors in \eqref{eq:sketch_err_bess} scales as $\beta^{-1}$, as opposed to the $\beta^{-\half}$ bound one would na\"ively apply from our ER decomposition-based guarantee in Lemma~\ref{lem:ps_guarantee}.

\begin{algorithm2e}[ht!]\label{alg:expander_spectral_sketching}
\caption{$\ESSalgo(\{\vG_i\}_{i \in [I]}, \vG, T, \delta, \eps, W, \beta)$}
\DontPrintSemicolon
\codeInput $\{\vG^{(i)}\}_{i \in [I]}$, subgraphs of simple %
$\vG = (V, E, \ww)$, with $V = A\cup B$, $E \subseteq A \times B$, 
$\max_{e \in \supp(\ww)} \ww_e \le W$, and such that $\{G^{(i)} \defeq \und(\vG^{(i)})\}_{i \in [I]}$ are a $(\phi, 2, J)$-expander decomposition of $G \defeq \und(\vG)$, $T$ a tree subgraph of $G$ with $\min_{e \in E(T)} \ww_e \ge 1$, $\delta, \eps \in (0, \frac 1 {100})$, $\beta > 0$ \;
$m \gets E(\vG)$, $n \gets V(\vG)$, $R \gets \emptyset$\;
\For{$i \in [I]$}{
$V_i \subseteq V(\vG^{(i)})$ be vertices in $\vG^{(i)}$ with combinatorial (unweighted) degrees at least $\beta$\;
$\vH \gets \vC \defeq \vG^{(i)}[V_i]$, $\hm \gets |E(\vH)|$, $\hn \gets |V(\vH)|$, $\wws \gets \ww$\;
\If{$\hm \ge 40\hn$}{
$\ww_0 \gets \ww$, $\vC_0 \gets \vC$, $\vH_0 \gets \vH$, $\ell_1 \gets \frac{1}{2\log^{2}(\frac{nW}{\eps})}$, $\tau_1 \gets \log(\frac 2 {\ell_1})$\;
\For{$0 \le t < \tau_1$}{\label{line:while_ess_1_start}
$\ww_{t + 1} \gets \BFS(\vH_t, \vC_t, \wws, \ell_1, \frac{\delta}{2I\tau_1},
\frac{\eps}{4I\tau_1}, E(\vH), T)$\;
$\vC_{t + 1} \gets (V(\vC), E(\vC), [\ww_{t + 1}]_{E(\vC)})$, $\vH_{t + 1} \gets (V(\vH), E(\vH), [\ww_{t + 1}]_{E(\vH)})$\;
\label{line:while_ess_1_end}
}
$F \gets \{e \in E(\vH) \mid [\ww_t]_e \le \ell_1 [\wws]_e\}$\;\label{line:ess_F_def}
$\ww_0 \gets \ww_t$, $t \gets 0$, $\vG_0 \gets \vG_t$, $\vH_0 \gets \vH$, $\ell_2 \gets \frac{\eps}{4nmW}$, $\tau_2 \gets \log(\frac 2 {\ell_2})$\;
\For{$0 \le t < \tau_2$}{\label{line:while_ess_2_start}
$\ww_{t + 1} \gets \BFS(\vH_t, \vC_t, \wws, \ell_2, \frac{\delta}{2I\tau_2},
\frac{\eps}{4I\tau_2}, F, T)$\;
$\vC_{t + 1} \gets (V(\vC), E(\vC), [\ww_{t + 1}]_{E(\vC)})$, $\vH_{t + 1} \gets (V(\vH), E(\vH), [\ww_{t + 1}]_{E(\vH)})$;
\label{line:while_ess_2_end}
}
$R \gets R \cup \{e \in E(\vH) \mid [\ww_t]_e \le \frac \eps {4nm}\}$, $\ww \gets \ww_t$\;
}
}
\Return{$\vG' \gets (V, E, \ww_{E \setminus R} + \ROalgo(\vG, \ww_{R}, T))$}\;
\end{algorithm2e}

\begin{lemma}\label{lemma:ess_guarantee}
There is a universal constant $\cess$ such that 
if $\cess \cdot \beta^{-1}\phi^{-2}\log(\frac {nW} {\delta\eps})\log^2\log(\frac{nW}{\eps}) \le 1$, and $\vG = (V, E, \ww)$ is a bipartite lift of a directed graph,
\ESS (Algorithm~\ref{alg:expander_spectral_sketching}) returns $\vG' = (V, E, \ww')$
satisfying the following guarantees with probability $\ge 1 - \delta$.
\begin{enumerate}
\item $\BB_{\vG}^\top \ww' = \BB_{\vG}^\top \ww$,
    $|\BB_{\vG}|^\top \ww' = |\BB_{\vG}|^\top \ww$. \label{item:ess_1}
\item $\nnz(\ww') \le \frac{31}{32} \nnz(\ww)+ \cess \cdot nJ\beta$.
    \label{item:ess_2}
\item For $G \defeq \und(\vG)$,
\[
\normop{\LL_{G}^{\frac \dagger 2}\BB_{\vG}^\top\Par{\WW' - \WW}\HH_{\vG}\LL_G^{\frac \dagger 2}} 
\le \cess \beta^{-\frac 1 2}\phi^{-1}\sqrt{\log\Par{\frac {nW} {\delta\eps}}}\log\log\Par{\frac{nW}{\eps}} + \eps.
\]
\label{item:ess_3}
\item For any fixed $\aa,\zz \in \R^V$, 
    \begin{equation}\label{eq:sketch_err_bess}
    \Abs{\aa^\top \BB_{\vG}^\top\Par{\WW' - \WW}\HH_{\vG} \zz} \le
    \Par{\cess \beta^{-1}\phi^{-2}\sqrt{\log\Par{\frac {nW}
    {\delta\eps}}}\log\log\Par{\frac{nW}{\eps}} + \eps} \cdot
    \norm{\aa}_{\LL_G} \norm{\zz}_{\LL_G}.
    \end{equation}\label{item:ess_4}
\end{enumerate}
Moreover, $\max_{e \in E} \frac{\ww'_e}{\ww_e} \le \cess$. 
The runtime of $\ESSalgo$ is
\[\bO\Par{|E|\log^2\Par{\frac{nW}{\delta\eps}}\log\Par{\frac{nW}{\eps}}}.\]
\end{lemma}
\begin{proof}
    We closely follow the arguments in the proof of \Cref{lem:ps_guarantee}.
    In light of \Cref{lemma:ex_to_er}, we let $\rho \defeq 8\beta^{-1}\phi^{-2}$.
    Further, throughout the proof, we condition on the success of all calls to
    \BFSalgo, assuming their input conditions are met, which gives the failure probability.
    We claim that every $\vG_t$ satisfies $0.9 \LL_{G_t} \pleq \LL_G \pleq 1.1 \LL_{G_t}$, where $G \defeq \und(G)$ and $G_t \defeq \und(\vG_t)$.
    Again, we defer proving this statement to the end of the proof.

    For a fixed $i \in [I]$, consider the first $\tau_1$ loops from
    \cref{line:while_ess_1_start} to \cref{line:while_ess_1_end}.
    Since we claimed $0.9 \LL_{G_t} \pleq \LL_G \pleq 1.1 \LL_{G_t}$ for all
    $t$, \Cref{lemma:ex_to_er} gives that each $\vG^{(i)}[V_i]$ is
    an ER decomposition piece with parameters $(1.2\rho, 2)$.
    Then, in all calls to \BFSalgo, the conditions on $\wws$ are met by
    assumption.
    Moreover, \BFSalgo~is only called  if $\hm \ge 40\hn$, and the conditions in
    \Cref{eq:wws_reqs} are preserved inductively by
    \Cref{lem:bfs_guarantee}, since the $\ell_1$ norm of the weights does not
    change by more than a $\frac{\eps}{4\tau_1}$ factor in each iteration.
    Thus, the $\tau_1$ loops all satisfy their input conditions
    and we may assume they succeed.
    We then show that $F$ on \Cref{line:ess_F_def} must have $|F| \ge \frac \hm
    4$.
    Suppose for contradiction that $|F| < \frac \hm 4$, which means the second
    part of \Cref{item:bfs_3} in \Cref{lem:bfs_guarantee} holds for all
    iterations $0\le t<\tau_1$.
    However, since \Cref{lem:bfs_guarantee} also guarantees
    \[
        \sum_{e \in E(\vH)} \log\Par{\frac{[\ww_\tau]_e}{\ww_e}} > 
        -\hm\log\Par{\frac 2 \ell} = -\hm \tau_1,
    \]
    we arrive at a contradiction after $\tau_1$ iterations.
    By using a similar argument, we also show that after $\tau_2$ loops from
    \cref{line:while_ess_2_start} to \cref{line:while_ess_2_end} have finished,
    at least $\frac \hm {16}$ edges are added to $R$.
    Notice that for each $i \in [I]$, at most $n_i \beta$ edges are not included
    in $E(\vH^{(i)})$, we then have the total number of remaining edges (i.e.,
    $\nnz(\ww')$) is bounded by
    \[
        \Par{m - \sum_{i \in I} \hm_i} + \sum_{i \in I} \Par{\hn_i \beta + \frac{15}{16}\hm_i}
        \leq
        m - \frac{1}{16}\cdot \frac{1}{2}m + nJ\beta
        =
        \frac{31}{32}m + nJ\beta,
    \]
    where the inequality follows by \Cref{item:exp:partition:cut} of
    \Cref{def:exp_partition}. Similarly, conditioned on all calls to \BFSalgo~succeeding, 
    by \Cref{item:bfs_1} of \Cref{lem:bfs_guarantee} and the first additional
    guarantee in \Cref{lemma:bfs_fix}, we obtain \Cref{item:ess_1}.

    Now, consider both error bounds in \Cref{item:ess_3,item:ess_4}.
    Note that we have $\alpha=2$ in the first $\tau_1$ calls to \BFSalgo~for
    each cluster and $\alpha = \frac{1}{\log^2(\frac{nW}{\eps})}$
    in the last $\tau_2$ calls.
    Since $I \leq m$, we have by
    \Cref{cor:multiple_cluster,cor:multiple_cluster_fix} and our decomposition parameters that
    the total spectral error in all intermediate iterations across all decomposition
    pieces is bounded by
    \begin{equation}\label{eq:two_alphas}
        O\Par{
            \tau_1 \cdot \sqrt{\rho \log\Par{\frac{\tau_1}{\delta}}}
            + 
            \tau_2 \cdot \sqrt{\rho \frac{\log\Par{\frac{\tau_1}{\delta}}}{\log^2\Par{\frac{nW}{\eps}}}}
        } 
        =
        O\Par{\beta^{-\frac 1 2}\phi^{-1}} \cdot
        \sqrt{\log\Par{\frac{mW}{\delta\eps}}} \cdot
        \log\log\Par{\frac{mW}{\eps}},
    \end{equation}
    where we used $\rho = 8\beta^{-1}\phi^{-2}$.
    Additionally, there is an $\frac{\eps}{4\tau_1 I}\cdot \tau_1 I +
    \frac{\eps}{4\tau_2 I} \cdot \tau_2 I$ additive spectral error
    term, which is $\le \frac{2\eps}{3}$ after accounting for the change
    in the graph Laplacian by \Cref{fact:opnorm_precondition}.
    For appropriate $\cess$, this both proves the desired spectral error bound
    by the triangle inequality, as well as the claimed $0.9 \LL_{G_t} \pleq
    \LL_G \pleq 1.1 \LL_{G_t}$ throughout the algorithm by
    \Cref{fact:dirclose_undirclose}.

    Consider now the sketching error bound \eqref{eq:sketch_err_bess}.
    For each cluster $\vG^{(i)}$ with $G^{(i)} \defeq \und(\vG^{(i)})$, 
    let $\LL_i,\DD_i$ and $\hLL_i,\hDD_i$ be the corresponding undirected
    Laplacians and weighted degrees of $G^{(i)}$ and $\und(\vG^{(i)}[V_i])$
    respectively.
    Following the notation of Corollary~\ref{cor:multiple_cluster_fix}, and using that $\norm{\xx}_{\PPi_{V(\vC_i)}} \le \norms{[\xx - \hx \1]_{V(\vC_i)}}_2$ for any $\hx \in \R$, we have that if all calls to $\BFSalgo$ succeed,
    \begin{align*}\Abs{\aa^\top \BB_{\vG}^\top (\WW' - \WW) \TT_{\vG} \zz}
    &\le
    O\Par{\sqrt{\log\Par{\frac{mW}{\delta\eps}}} \cdot
    	\log\log\Par{\frac{mW}{\eps}}} \cdot
    \sum_{i \in [I]} \bw_i \norm{[\aa - \ha_i \vone]_{V_i}}_2 \norm{[\zz - \hz_i \vone]_{V_i}}_2 \\ 
    &+ \frac \eps 4 \norm{\aa}_{\LL_G} \norm{\zz}_{\LL_G},
    \end{align*}
    for any choices of scalars $\{\ha_i\}_{i \in [I]}$, $\{\hz_i\}_{i \in [I]}$. In the above, we used a calculation analogous to \eqref{eq:two_alphas} to bound the first term on the right-hand side. Now, we choose each $\ha_i$ and $\hz_i$ as defined in \Cref{lemma:ex_deg}, so that for all $i \in [I]$,
    \begin{equation}\label{eq:cluster_bound}
    \begin{aligned}
    \bw_i  \norm{[\aa - \ha_i \vone]_{V_i}}_2^2 &\le \bw_i \beta^{-1} \sum_{v \in V_i} [\deg_{G^{(i)}}]_v (\aa_v - \ha_i)^2 \le 2\beta^{-1} \phi^{-2} \aa^\top \LL_i \aa, \\
     \bw_i  \norm{[\zz - \hz_i \vone]_{V_i}}_2^2 &\le \bw_i \beta^{-1} \sum_{v \in V_i} [\deg_{G^{(i)}}]_v (\zz_v - \hz_i)^2 \le 2\beta^{-1} \phi^{-2} \zz^\top \LL_i \zz.
    \end{aligned}
    \end{equation}
    By the Cauchy-Schwarz inequality and \eqref{eq:cluster_bound}, we then obtain the first term in \eqref{eq:sketch_err_bess}:
    \begin{align*}
        \sum_{i \in [I]} \bw_i \norm{[\aa - \ha_i \vone]_{V_i}}_2 \norm{[\zz - \hz_i \vone]_{V_i}}_2 &\le \sqrt{\sum_{i \in [I]} \bw_i  \norm{[\aa - \ha_i \vone]_{V_i}}_2^2} \cdot \sqrt{\sum_{i \in [I]} \bw_i \norm{[\zz - \hz_i \vone]_{V_i}}_2^2}
        \\
        &\le
        2\beta^{-1} \phi^{-2} \cdot
        \sqrt{\aa^\top \Par{\sum_{i \in [I]} \LL_i} \aa }
        \cdot
        \sqrt{\zz^\top \Par{\sum_{i \in [I]} \LL_i} \zz} \\
        &\le
        2\beta^{-1} \phi^{-2} \cdot
        \norm{\aa}_{\LL_G} \norm{\zz}_{\LL_G}.
    \end{align*}
    The second term in \eqref{eq:sketch_err_bess} comes from our earlier application of \Cref{cor:multiple_cluster_fix}.

    Finally, \ROalgo~incurs at most $\frac \eps {4}$ spectral error through the
    final graph by \Cref{lemma:rounding}, which is at most $\frac \eps {3}$
    spectral error through the original graph by
    \Cref{fact:opnorm_precondition}. Here we again used that $G$ is connected (Section~\ref{sec:prelims}), which implies each $G_t$ is connected via our earlier bound $0.9\LL_{G_t} \preceq \LL_{G_t} \preceq 1.1 \LL_{G_t}$.
    Using an analogous argument from above, this also gives an additive
    sketching error of at most $\frac{\eps}{3} \cdot \norm{\aa}_{\LL_G}
    \norm{\zz}_{\LL_G}$.
    By the first claim in \Cref{lemma:rounding}, \Cref{item:ess_1} remains true.
    \Cref{item:ess_2} in the lemma statement is clear as we
    only modify weights within clusters, and \Cref{item:bfs_2} of
    \Cref{lem:bfs_guarantee} shows no edge weight grows by more than a factor
    of 60.
    The runtime follows by applying
    \Cref{cor:multiple_cluster_fix} to each of the
    $\tau_1 + \tau_2 = O(\log \frac{nW}{\eps})$ times we call \BFSalgo~on each
    expander.
\end{proof}

\subsection{Complete spectral sketching algorithm}

We are now ready to give our main guarantee on improved constructions of graphical spectral sketches (Definition~\ref{def:sketch_undir}), as well as their Eulerian generalization (Definition~\ref{def:sketch_directed}).

\begin{algorithm2e}[ht!]
\caption{$\SSalgo(\vG, \epsilon, \delta)$}
\label{alg:spectral_sketching}
\DontPrintSemicolon
\codeInput $\vG = (V, E, \ww)$ with $\ww_e \in [1, U]$ for all $e \in
E$, $\eps, \delta \in (0, \frac{1}{100})$  \;
$n \gets |V|$, $m \gets |E|$\;
$\vG^\uparrow \defeq (V\cup V',E^\uparrow,\ww) \gets \blift(\vG)$\;
$T \gets $ arbitrary spanning tree of $G^\uparrow \defeq \und(\blift(\vG))$, $\hE^\uparrow \gets E^\uparrow \setminus E(T)$\;
$R \gets 6\log n$, $U_{\max} \gets U \cdot \cess^R$ for $\cess$ in \Cref{lemma:ess_guarantee} \;
$\beta \gets \frac{400000\cess^2}{\cadk^2 \cdot \eps} \cdot \log^6 n
\log\Par{\frac{nU_{\max}}{\delta}} \log^2\log\Par{\frac{nU_{\max}}{\delta}}$,
$t \gets 0$, $\ww_0 \gets \ww$\; 
\While{$t < R$ $\mathbf{and}$ $\nnz([\ww_t]_{\hE^\uparrow}) > 4\cess n\beta \log_2(\frac{32mnU_{\max}}{\eps})$}{\label{line:while_start_ss}
$\vG_t^\uparrow \gets (V\cup V', E^\uparrow, \ww_t)$, $G_t^\uparrow \gets \und(\vG_t^\uparrow)$\;
$S \gets \EPalgo([G_t^\uparrow]_{\hE^\uparrow}, 2, \frac \delta {4R})$ \Comment*{See \Cref{lemma:ex_partition}.}
$(\vG_t^\uparrow)' \defeq (V\cup V', E^\uparrow, \ww'_t) \gets \ESSalgo(S, \vG_t^\uparrow, T, \frac{\delta}{4R}, \frac{\eps}{4R}, U_{\max}, \beta)$ \;
$D \gets \{e \in E^\uparrow \mid [\ww'_t]_e \le \frac{\eps}{4mn}\}$\;\label{line:d_define_ss}
$\ww_{t + 1} \gets [\ww'_{t}]_{E^\uparrow \setminus D} + \ROalgo(\vG^\uparrow, [\ww'_{t}]_D, T)$\; \label{line:round_ss}
$t \gets t + 1$\;\label{line:while_end_ss}
}
\Return{$\vH \gets (V, E, \ww_t)$}\;
\end{algorithm2e}

\begin{theorem}\label{thm:fastsketch}
Given Eulerian $\vG = (V, E, \ww)$ with $|V| = n$, $|E| = m$, $\ww \in [1, U]^E$
and $\eps, \delta \in (0, \frac{1}{100})$, $\SSalgo$
(Algorithm~\ref{alg:spectral_sketching}) returns a distribution over $\vH$
that is an $(\eps, \delta)$-Eulerian graphical sketch. Moreover, with probability $\ge 1 - \delta$, $\vH$ is a $\sqrt{\eps}$-approximate Eulerian sparsifier of $\vG$, and
\[
\begin{gathered}
|E(\vH)| =
O\Par{\frac{n}{\eps}\log^7(n)\log^2\Par{\frac{nU}{\delta}}
\log^2\log\Par{\frac{nU}{\delta}}},
\\
\log\Par{\frac{\max_{e \in \supp(\ww')} \ww_e'}{\min_{e \in \supp(\ww')} \ww_e}}
= O(\log(nU)).
\end{gathered}
\]
The runtime of $\SSalgo$ is 
\[\bO\Par{m\log^2\Par{\frac{nU}{\delta}}\log\Par{nU}
+m\log^8(n)\log\Par{\frac{n}{\delta}}}.\]
\end{theorem}

\begin{proof}
Throughout, we condition on the event that all of the $R$ calls to \EPalgo~and
\ESSalgo~succeed, which happens with probability $\ge 1-\delta$.
Since \ESSalgo~guarantees no weight grows by more than a $\cess$ factor in each
call, $U_{\max}$ is an upper bound for the maximum possible weight throughout
the algorithm.
As we remove any edge with weight below $\frac{\eps}{4mn}$ on
\cref{line:round_ss}, the number of expander pieces is upper bounded by 
$J_{\max} \defeq \log_2(\frac{32mnU_{\max}}{\eps})$ by \Cref{lemma:ex_partition}.

When $\nnz(\ww_t) \ge 4\cess \cdot n\beta J_{\max}$,
\Cref{item:exp:partition:cut} of \Cref{def:exp_partition} and \Cref{item:ess_2}
of \Cref{lemma:ess_guarantee} guarantees that the number of edges in each
iteration decreases by at least a $\frac{1}{64}$ factor.
Since $m \leq n^2$, we may assume for the rest of the proof that $\eps > \frac{1}{n^2}$. 
Therefore, after $R$ iterations, we are guaranteed that the number of edges at
termination is at most $O(n\beta J_{\max})$.
Plugging in the definition of $\beta$ and $\phi = \cadk\log^{-2}(m) \ge
\frac{1}{4}\cadk\log^{-2}(n)$ and noting that $\log(U_{\max}) = O(\log(nU))$ gives the
desired sparsity bound.
The runtime follows from combining \Cref{lemma:ex_partition},
\Cref{lemma:ess_guarantee} and noting that the number of edges decreases
geometrically until the lower bound on \cref{line:while_start_ss}.

The degree-preserving property follows from \Cref{item:ess_1} of
\Cref{lemma:ess_guarantee} and the first claim of \Cref{lemma:rounding}.
This guarantees that if $\vG$ is Eulerian, then $\vH$ is also Eulerian.

Next, consider the spectral error bound.
By \Cref{item:ess_3} of \Cref{lemma:ess_guarantee}, the total spectral error
incurred within each iteration of the while loop from \cref{line:while_start_ss}
to \cref{line:while_end_ss} with respect to the current Laplacian is bounded by
\[
    5\cess \cdot
    \beta^{-1/2}\phi^{-1}\sqrt{\log\Par{\frac{nU_{\max}}{\delta}}}
    \log\log(nU_{\max}) + \frac{\eps}{4R}
    \le
    \frac{1}{4}\frac{\sqrt{\eps}}{R} + \frac{1}{4}\frac{\eps}{R}
    \leq
    \frac{1}{2}\frac{\sqrt{\eps}}{R}.
\]
We condition on $0.9 \LL_G \pleq \LL_{G_t} \pleq 1.1\LL_G$ for all $t < R$,
which shows a total spectral error of at most $\frac{2}{3}\sqrt{\eps}$ over all
iterations due to \Cref{fact:dirclose_undirclose}.
Similarly, the rounding on \cref{line:round_ss} also contributes at most $\frac \eps 4$
by \Cref{lemma:rounding} and \Cref{fact:dirclose_undirclose}.
This also shows our assumption $0.9 \LL_G \pleq \LL_{G_t} \pleq 1.1\LL_G$ holds, as $\eps < \frac{1}{100}$. As before, we achieve the sparsity bound by dropping edges with zero weight in $\ww_t$.

Finally, consider the sketching error bound.
We take the same definition of $\QQ$ as in \Cref{lemma:blift_spectral}.
Let $\aa,\zz \in \R^V$ be arbitrary fixed vectors.
We have, by \Cref{item:ess_4} of \Cref{lemma:ess_guarantee}, the sketching error
for $\QQ\aa$ and $\QQ\zz$ in $\vG^\uparrow$ is bounded by
$\norm{\QQ\aa}_{\LL_{G^\uparrow}} \norm{\QQ\zz}_{\LL_{G^\uparrow}}$ times
\[
    \frac{4}{3} \cdot R \cdot \Par{5\cess \cdot \beta^{-1}\phi^{-2}
        \sqrt{\log\Par{\frac{nU_{\max}}{\delta}}}\log\log(nU_{\max}) +
    \frac{\eps}{4R}}
    \le
    \frac{2}{3} \eps,
\]
where the factor of $\frac{4}{3}$ again comes from the valid assumption of $0.9 \LL_G
\pleq \LL_{G_t} \pleq 1.1\LL_G$ and \Cref{fact:dirclose_undirclose} for each
factor of the form $\norm{\xx}_{\LL_{G_t}}$.
By an analogous argument in the proof of \Cref{lemma:ess_guarantee}, the
additive error by \ROalgo~on \cref{line:round_ss} is bounded by 
$\frac{4}{3} \cdot \frac{\eps}{4} \norm{\QQ\aa}_{\LL_{G^\uparrow}}
\norm{\QQ\zz}_{\LL_{G^\uparrow}}$.
Now, the fact that \eqref{eq:lift_imply_1} implies \eqref{eq:lift_imply_2} gives the desired sketching error bound in the
unlifted graph.
\end{proof}

Our spectral sketch algorithm has additional desirable properties in the undirected graph setting, where we can ensure that the sketched graph is also undirected via the following reduction.

\begin{lemma} \label{lemma:undir_to_dir}
For an undirected graph $G=(V,E,\ww)$, let $\vH=(V,E',\ww)$ be a directed graph
where each edge $e' \in E'$ has the same endpoints as an undirected edge $e \in E$ with an
arbitrary orientation.
Let $\ww' \in \R^E$ satisfy $\BB_{\vH}^\top \ww' = \BB_{\vH}^\top \ww$.
Then for any $\xx \in \R^V$,
\[
\begin{gathered}
    \normop{\LL_G^{\frac \dagger 2} \BB_G^\top (\WW' - \WW) \BB_G \LL_G^{\frac \dagger 2}}
    \le
    4 \normop{\LL_H^{\frac \dagger 2} \BB_{\vH}^\top (\WW' - \WW) \HH_{\vH}
    \LL_H^{\frac \dagger 2}}, 
    \\
    \Abs{\xx^\top \LL_G^{\frac \dagger 2} \BB_G^\top (\WW' - \WW) \BB_G \LL_G^{\frac
    \dagger 2} \xx} \le
    4 \Abs{\xx^\top \LL_H^{\frac \dagger 2} \BB_{\vH}^\top (\WW' - \WW)
    \HH_{\vH} \LL_H^{\frac \dagger 2} \xx},
\end{gathered}
\]
where $H \defeq \und(\vH) = (V,E,2\ww)$.
\end{lemma}
\begin{proof}
Without loss of generality, we assume that orientations are chosen so that $\BB_G \defeq \BB_{\vH}$.
Since $\ww'-\ww$ is a circulation, we have by \Cref{lemma:circeq},
$\BB_{\vH}^\top (\WW'-\WW)\HH_{\vH} = -\TT_{\vH}^\top (\WW'-\WW) \BB_{\vH}$.
Further, as
\[
    \BB_G^\top (\WW'-\WW) \BB_G = \BB_{\vH}^\top (\WW'-\WW) \HH_{\vH} -
    \BB_{\vH}^\top (\WW'-\WW) \TT_{\vH},
\]
applying the triangle inequalities on operator norms and absolute values, combined with
the fact $\LL_H = 2\LL_G$, gives both desired inequalities.
\end{proof}

Moreover, we can use the following claim from \cite{ChuGPSSW18} to show that the output of our algorithm in the undirected case is an approximate \emph{inverse sketch} of $G$, i.e., it preserves quadratic forms with the Laplacian pseudoinverse. This is useful for approximating effective resistances.

\begin{lemma}[Lemma 6.8, \cite{ChuGPSSW18}] \label{lemma:sketch_inverse}
Let $\MM,\NN$ be symmetric PSD matrices of the same dimension, and let $\xx$ be a vector
of the same dimension such that for some $\eps \in (0,0.1)$,
\[
    \normop{\MM^{\frac \dagger 2}(\MM - \NN)\MM^{\frac \dagger 2}} \le \sqrt{\eps}, \quad
    |\xx^\top \MM^\dagger (\MM - \NN) \MM^\dagger \xx| \le \eps \cdot 
    \xx^\top \MM^\dagger \xx.
\]
Then,
\[
    |\xx^\top (\MM^\dagger-\NN^\dagger) \xx| \le 7 \eps \cdot 
    \xx^\top \MM^\dagger \xx.\footnotemark  
\]
\footnotetext{The original Lemma 6.8 in \cite{ChuGPSSW18} uses a different, more symmetric definition of $\eps$-approximation, involving multiplicative factors of $e^\eps$, but it is straightforward to check that the same constant factors hold for our definition.}
\end{lemma}

The following theorem is a refined version of \Cref{thm:fastsketch_undir}. To obtain these results, we crucially use the fact that the guarantees of \Cref{lemma:ess_guarantee} continue to hold even if the input directed graph is not Eulerian (as the signed variant of an undirected graph, as in \Cref{lemma:undir_to_dir}, need not be Eulerian). 
\begin{corollary} \label{thm:fastsketch_undir_detailed}
There is an algorithm that, given undirected graph $G = (V,E,\ww)$ with $|V| =
n$, $|E| = m$, $\ww \in [1, U]^E$ and $\eps, \delta \in (0, \frac{1}{100})$,
returns a distribution over graphs $H$ which is an $(\eps, \delta)$-graphical spectral sketch, and
\begin{gather*}
 |E(H)| =
	O\Par{\frac{n}{\eps}\log^7(n)\log^2\Par{\frac{nU}{\delta}}
		\log^2\log\Par{\frac{nU}{\delta}}},
	\\
	\log\Par{\frac{\max_{e \in \supp(\ww')} \ww_e'}{\min_{e \in \supp(\ww')} \ww_e}}
	= O(\log(nU)).
\end{gather*}
Moreover, with probability $\ge 1 - \delta$, $H$ is a $\sqrt{\eps}$-approximate
spectral sparsifier of $G$, and for an arbitrary fixed $\xx \in \R^V$,
$|\xx^\top (\LL_H^\dagger - \LL_G^\dagger) \xx|
\le \eps \cdot \xx^\top \LL_H^\dagger \xx$.
The runtime of the algorithm is $\bO\Par{m\log^2\Par{\frac{nU}{\delta}}\log\Par{nU}
+m\log^8(n)\log(\frac{n}{\delta})}$.
\end{corollary}
\begin{proof}
This is a direct consequence of \Cref{lemma:undir_to_dir,lemma:sketch_inverse}
and \Cref{thm:fastsketch}.
Here, instead of the standard transformation of doubling the edges and taking
both directions of each edge for $G$, we keep one edge each and set an arbitrary
direction as in \Cref{lemma:undir_to_dir}.
We remark that the input directed graph, say $\vG = (V,E',\ww)$, to $\SSalgo$
need not be Eulerian.
Let $\vG' = (V,E',\ww')$ be the resulting directed graph, then $\BB_{\vG}^\top \ww'=
\BB_{\vH}^\top \ww$ by \Cref{item:ess_1} of \Cref{lemma:ess_guarantee} and the first
claim of \Cref{lemma:rounding}.
Scaling $\eps$ by a factor of $\frac{1}{30}$ then guarantees our desired
approximation factors.
\end{proof}

\section*{Acknowledgments}

We thank Richard Peng \cite{Peng23} for clarifying the dependence in Proposition~\ref{prop:ps22} on $\delta$ and $U$. 
We thank the authors of \cite{CohenKPPRSV17} for helpful discussions.
Sushant Sachdeva’s research is supported by an Natural Sciences and Engineering Research Council of Canada
(NSERC) Discovery Grant RGPIN-2018-06398, an Ontario Early Researcher Award (ERA) ER21-16-283, and a Sloan
Research Fellowship. Aaron Sidford was supported in part by a Microsoft Research Faculty Fellowship, NSF CAREER Grant CCF-1844855, NSF Grant CCF-1955039, and a PayPal research award. 
Part of this work was conducted while authors were visiting the Simons Institute for the Theory of Computing Fall 2023 Program \emph{Data Structures and Optimization for Fast Algorithms}. The authors are grateful to the Simons Institute for its support.

\printbibliography

\begin{appendix}
\section{Deferred proofs from Section~\ref{sec:prelims}}\label{app:deferred}

\restatecirceq*
\begin{proof}
	We observe that $\HH^\top \XX \HH = \diag{\HH^\top \xx}$ and $\TT^\top \XX \TT = \diag{\TT^\top \xx}$. The first claim then follows from $\HH^\top \xx = \TT^\top \xx$ as $\xx$ is a circulation. The second claim then follows from
	\[\BB^\top \XX \HH = \HH^\top \XX \HH - \TT^\top \XX \HH = \TT^\top \XX \TT - \TT^\top \XX \HH = -\TT^\top \XX \BB.\] 
\end{proof}

\restatedirundir*
\begin{proof}
	Throughout the proof, let
	\[\BB \defeq \begin{pmatrix} \BB_{\vG}\\ \BB_{\vH}\end{pmatrix} \in \{0, 1\}^{(E \cup F) \times V},\;\ww \defeq \begin{pmatrix} \ww_{\vG} \\ -\ww_{\vH} \end{pmatrix} \in \R^{E \cup F},\]
	and define $\HH,\TT \in \{0,1\}^{(E \cup F) \times V}$ to be appropriate concatenations such that $\BB = \HH - \TT$. Observe that $\BB^\top \ww = \BB_{\vG}^\top \ww_{\vG} - \BB_{\vH}^\top \ww_{\vH} = \vzero_V$. By \Cref{lemma:circeq}, we have
	\[
	\vLL_{\vG} - \vLL_{\vH} = \BB^\top \WW \HH, \quad
	\vLL_{\rev(\vG)} - \vLL_{\rev(\vH)} = -\BB^\top \WW \TT = \HH^\top \WW \BB = \vLL_{\vG}^\top - \vLL_{\vH}^\top.
	\]
	It then suffices to apply the triangle inequality, that transposition preserves the operator norm, and the characterization $\LL_G = \vLL_{\vG} + \vLL_{\rev(\vG)}$ (with a similar equality for $\vH$ and $H$).
\end{proof}

\restatelinalgthree*
\begin{proof}
	Since $\LL_G$ and $\LL_H$ share a kernel, the given condition implies
	$(1 - \eps)\LL_G \preceq \LL_H \preceq (1 + \eps) \LL_G$. 
	Hence, $\|v\|_{\LL_G} \le 1$ implies $\|v\|_{\LL_H} \le \sqrt{1 + \eps}$, and so the conclusion follows from 
	\begin{align*}
		\normop{\LL_{G}^{\frac \dagger 2} \MM \LL_{G}^{\frac \dagger 2}} = \sup_{\substack{u, v \perp \vone_V \\ \|u\|_{\LL_G}, \|v\|_{\LL_G} \le 1}} u^\top \MM v \le (1 + \eps)\sup_{\substack{u, v \perp \vone_V \\ \|u\|_{\LL_H}, \|v\|_{\LL_H} \le 1}} u^\top \MM v = (1 + \eps)\normop{\LL_{H}^{\frac \dagger 2} \MM \LL_{H}^{\frac \dagger 2}}.
	\end{align*}
\end{proof}

\section{Rounding}\label{app:rounding}

In this section, we prove \Cref{lemma:rounding}, our guarantee on $\ROalgo$.

\restaterounding*
\begin{proof}
Throughout the proof we drop the subscripts $\vG$, $G$ from $\BB$, $\HH$, $\LL$ for simplicity.
The algorithm sets $\yy$ to be the unique flow on the edges of tree $T$ that satisfies $\BB^{\top}\yy = \dd.$ Such a vector $\yy$ can be constructed in $O(n)$ time by recursively computing the flow required at each leaf, and then removing the leaf.
By construction, $\supp(\yy) \subseteq T$.  Since $\dd \perp \vone_V$, we also have $\norm{\yy}_{\infty} \le \frac{1}{2}\norm{\dd}_1.$

Next, recall $\BB^\top \zz = \dd$, so $\norm{\dd}_1 = \norms{\BB^{\top}\zz}_1 \le 2 \norm{\zz}_1$, and $\yy - \zz$ is a circulation on $G$.
We now show that spectral error induced by this circulation $\yy-\zz$ is not significant in the directed Laplacians.
For every edge $e \notin T,$ we let $\cc^{(T,e)} \in \{0, 1\}^E$ denote the (signed) incidence vector of the unique cycle in $T \cup {e}.$
We observe that $\zz-\yy$ can be expressed uniquely as $\sum_{e \notin T} \zz_e \cc^{(T,e)}$, so
\begin{align*}
 \normop{\LL^{\frac \dagger 2} \BB^\top (\YY-\ZZ) \HH \LL^{\frac \dagger 2}} & \le \sum_{e \notin T}  |\zz_e| \normop{\LL^{\frac \dagger 2} \BB^\top \CC^{(T,e)} \HH \LL^{\frac \dagger 2}}.
\end{align*}
It suffices to show that each operator norm in the right-hand side is bounded by $n.$ Note that 
\begin{equation}
    \begin{aligned} \label{eq:cyclenorm}
\normop{\LL^{\frac \dagger 2} \BB^\top \CC^{(T,e)} \HH \LL^{\frac \dagger 2}} & = \sqrt{\normop{(\LL^{\frac \dagger 2} \BB^\top \CC^{(T,e)} \HH \LL^{\frac \dagger 2})(\LL^{\frac \dagger 2} \BB^\top \CC^{(T,e)} \HH \LL^{\frac \dagger 2})^{\top}}} \\
     & = \sqrt{\normop{\LL^{\frac \dagger 2} (\BB^\top \CC^{(T,e)} \HH) \LL^{\dagger} (\BB^\top \CC^{(T,e)} \HH)^{\top} \LL^{\frac \dagger 2}}}.
\end{aligned}
\end{equation}
We will bound the norm of the last matrix in the above expression.
Observe that $\BB^\top \CC^{(T,e)} \HH$ is just the directed Laplacian of the cycle with unit weights. Denote it $\MM$ for brevity. 
We further observe that $\MM^{\top}\MM$ is twice the undirected Laplacian of the cycle with unit weights. Since the cycle with unit weights is a downweighted subgraph of (the undirected graph) $G,$ we have $\MM^{\top} \MM \preceq 2\LL.$
Thus,
\[\MM \LL^{\dagger} \MM^{\top} \preceq 2\MM (\MM^{\top} \MM)^{\dagger} \MM^{\top} \preceq 2\II_V.\]
This implies
\[\LL^{\frac \dagger 2} (\BB^\top \CC^{(T,e)} \HH) \LL^{\dagger} (\BB^\top \CC^{(T,e)} \HH)^{\top} \LL^{\frac \dagger 2}
\preceq 2\LL^{\dagger} \preceq 2\LL_T^{\dagger}.\]
Since $T$ has edge weights $\ge 1$ and diameter $\le n,$
$\normsop{\LL_T^\dagger} \le \frac {n^2} 4$ \cite{Mohar91}. 
By using this bound in \eqref{eq:cyclenorm} and taking square roots, we obtain the third result.

To see the last result, we bound using the triangle inequality:
\[\normop{\LL^{\frac \dagger 2}\BB^\top \YY \HH\LL^{\frac \dagger 2}} \le \sum_{e \in T} |\yy|_e \normop{\LL^{\frac \dagger 2}\bb_e \ee_{h(e)}^\top \LL^{\frac \dagger 2}}.\]
Note that $\bb_e \ee_{h(e)}^\top \ee_{h(e)} \bb_e^\top = \bb_e\bb_e^\top \preceq \LL$. Therefore, using $\normsop{\LL_T^\dagger} \le \frac{n^2}{4} \le n^2$, we have the claim:
\begin{align*}
\normop{\LL^{\frac \dagger 2} \bb_e \ee_{h(e)}^\top \LL^{\frac \dagger 2}} &= \sqrt{\normop{\LL^{\frac \dagger 2} \bb_e \ee_{h(e)}^\top \LL^\dagger \ee_{h(e)} \bb_e^\top \LL^{\frac \dagger 2}}} \\
&\le \sqrt{\normop{\LL^{\frac \dagger 2} \bb_e \ee_{h(e)}^\top \LL_T^\dagger \ee_{h(e)} \bb_e^\top \LL^{\frac \dagger 2}}}\\
&\le n \sqrt{\normop{\LL^{\frac \dagger 2} \bb_e \ee_{h(e)}^\top \ee_{h(e)} \bb_e^\top \LL^{\frac \dagger 2}}} \le n.
\end{align*}
\end{proof}

\section{Potential improvements to Theorem~\ref{thm:existential}}
\label{sec:conjectures}

In this section, we discuss two natural avenues to improve the sparsity of our sparsifier construction in Theorem~\ref{thm:existential}: improving the matrix discrepancy result in Proposition~\ref{prop:bjm_measure}, and obtaining a graph decomposition with stronger guarantees than Proposition~\ref{prop:er_partition}.

\paragraph{Partial coloring matrix Spencer.} Consider the following conjecture.

\begin{conjecture}[Partial coloring matrix Spencer]\label{conj:partial_coloring}
	There is a constant $\gamma \in (0, 1)$ such that for $\{\AA_i\}_{i \in [m]} \subset \Sym^n$ with $\normsop{\sum_{i \in [m]} \AA_i^2} \le 1$, there exists $\xx \in [-1, 1]^m$ such that
	\[|\{i \in [m] \mid |\xx_i| = 1 \}| \ge \gamma m,\text{ and } \normop{\sum_{i \in [m]} \xx_i \AA_i} \le \frac 1 \gamma.\]
\end{conjecture}

By observation, applying the posited coloring in
Conjecture~\ref{conj:partial_coloring} in place of
Proposition~\ref{prop:bjm_measure} and Corollary~\ref{cor:partialcolor_variance}
when designing our $\ESOalgo$ (see the proof of \Cref{lemma:eps})%
would remove the last low-order term in Theorem~\ref{thm:existential}, giving a sparsity bound of $O(n\log U + n\log(n) \cdot \eps^{-2})$, which is $O(n\log(n)\cdot \eps^{-2})$ for $U = \poly(n)$. Conjecture~\ref{conj:partial_coloring} has already been stated implicitly or explicitly in the literature in several forms (see e.g., Conjecture 3 in \cite{ReisR23} with $p = \infty$). Notably, it is stronger than the matrix Spencer conjecture, which asserts (in the most prominent special case) that for a set of matrices $\{\AA_i\}_{i \in [n]} \in \Sym^{n}$ with $\normsop{\AA_i} \le 1$ for all $i \in [n]$, there exists $\xx \in \{\pm 1\}^n$ such that $\normsop{\sum_{i \in [n]} \xx_i \AA_i} = O(\sqrt n)$. In the context of Conjecture~\ref{conj:partial_coloring}, considering the matrices $\frac{1}{\sqrt{n}}\AA_i,$ the assumption is satisfied since $\frac{1}{n} \sum_i \AA_i^2 \preceq \II,$ and hence Conjecture~\ref{conj:partial_coloring} implies a partial coloring with spectral discrepancy $O(1)$ (i.e., $\xx \in [-1, 1]^n$ with a constant fraction of coordinate magnitudes equal to $1$). Standard boosting techniques (see, e.g., \cite{Giannopoulos97} or Section 4 of \cite{Rothvoss17}) show that we can recurse upon this partial coloring scheme to obtain a full coloring in $\{\pm 1\}^n$, since the matrix variance decreases by a constant factor in each iteration.

We also note that Conjecture~\ref{conj:partial_coloring} has already been established in prominent settings, when the matrices $\{\AA_i\}_{i \in [m]} \subset \Sym^n$ are all low-rank. For example, Theorem 1.4 of \cite{KyngLS20} proves Conjecture~\ref{conj:partial_coloring} for rank-$1$ matrices (with a precise constant $\gamma = \frac 1 4$), and if all $\{\AA_i\}_{i \in [m]}$ have images supported in the same $O(\sqrt{n})$-dimensional subspace, Theorem 3.5 of \cite{HopkinsRS22} also proves the claim. For completeness, using tools recently developed in \cite{BansalJM23}, we provide a proof of Conjecture~\ref{conj:partial_coloring} in one of the strongest settings we are aware of known in the literature.

\begin{proposition}[Lemma 3.1, \cite{BansalJM23}]
	There is a constant $\gamma \in (0, 1)$ such that for $\{\AA_i\}_{i \in [m]} \subset \Sym^n$ with $\normsop{\sum_{i \in [m]} \AA_i^2} \le \sigma^2$ and with $\sum_{i \in [m]} \normsf{\AA_i}^2 \le mf^2$, there exists $\xx \in [-1, 1]^m$ such that
	\[|\{i \in [m] \mid |\xx_i| = 1 \}| \ge \gamma m,\text{ and } \normop{\sum_{i \in [m]} \xx_i \AA_i} \le \frac 1 \gamma\Par{\sigma + \sqrt{\sigma f}\log^{\frac 3 4}(n)}.\]
\end{proposition}

\begin{corollary}\label{cor:lowrank_matrix_spencer}
	If the images of all $\AA_i$ are supported in the same $r$-dimensional subspace and $m \ge r \cdot \log^3 n$, Conjecture~\ref{conj:partial_coloring} is true. 
\end{corollary}
\begin{proof}
	By linearity of trace, we can choose $f$ such that
	\[f^2 = \frac 1 m \sum_{i \in [m]} \Tr\Par{\AA_i^2} = \frac 1 m \Tr\Par{\sum_{i \in [m]} \AA_i^2} \le \frac r m \normop{\sum_{i \in [m]} \AA_i^2} \le \frac {r} m,\]
	where we use that the rank of $\sum_{i \in [m]} \AA_i^2$ is at most $r$. The resulting discrepancy bound is
	\[O\Par{1 + \sqrt[4]{\frac r m} \cdot \log^{\frac 3 4}(n)}\]
	which proves the claim for sufficiently small $\gamma$, under the assumed parameter bounds.
\end{proof}

For example, while Corollary~\ref{cor:lowrank_matrix_spencer} does not establish Conjecture~\ref{conj:partial_coloring} in full generality, it does establish it when $m$ is larger than $n$ by a polylogarithmic factor, as we may take $r = n$.

\paragraph{Stronger effective resistance decomposition.} We further observe that another avenue to improving Theorem~\ref{thm:existential} is via strengthening Proposition~\ref{prop:er_partition}, the graph decomposition result it is based on. We present one source of optimism that the parameters in Proposition~\ref{prop:er_partition}, which gives an $(O(\frac{n\log n}{m}), O(1), O(\log U))$-ER decomposition, can be directly improved, though this remains an open question suggested by our work. In particular, we use the following claim in \cite{AlevALG18}.

\begin{proposition}[Theorem 3, \cite{AlevALG18}]\label{prop:aalg18}
	Given $G = (V, E, \ww)$ with $n = |V|$ and sufficiently large $C > 1$, there is a constant $\alpha \in (0, 1)$ and a polynomial-time algorithm which finds a partition $V = \{V_j\}_{j \in [J]}$ such that if $\{G_j \defeq G[V_j]\}_{j \in [J]}$ are the corresponding induced subgraphs, we have
	\begin{equation}\label{eq:aalg18_edges_cut}\sum_{e \in E \setminus \bigcup_{j \in [J]} E(G_j)} \ww_e \le \frac{\sum_{e \in E} \ww_e}{C\alpha},\end{equation}
	and
	\begin{equation}\label{eq:aalg18_er_diameter}\max_{u, v \in V_j} \ER_G(u, v) \le \frac{n}{C^3\alpha\sum_{e \in E} \ww_e} \text{ for all } j \in [J].\end{equation}
\end{proposition}

Proposition~\ref{prop:aalg18} immediately implies an improvement of Proposition~\ref{prop:er_partition} when $\ww$ is well-behaved.
\begin{corollary}\label{cor:strong_erd_constant_wratio}
	If $G = (V, E, \ww)$ has $\ww \in [1, U]^E$ for $U = \Theta(1)$, there exists a $(\frac n {\beta m}, \infty, 1)$-ER decomposition of $G$, for a constant $\beta \in (0, 1)$.
\end{corollary}
\begin{proof}
	Let $m \defeq |E|$. Apply Proposition~\ref{prop:aalg18} to $G$ with parameter $C \gets \frac{2U}{\alpha} = \Theta(1)$. The guarantee \eqref{eq:aalg18_edges_cut} implies that the total cut weight is at most $\frac m 2$, so less than half the edges are cut as $\min_{e \in E} \ww_e \ge 1$. Further, \eqref{eq:aalg18_er_diameter} shows that the $\rho$ parameter in Definition~\ref{def:er_partition} is bounded by $U \cdot \frac{n}{C^3 \alpha m} = \Theta\Par{\frac{n}{m}}$,
	as desired. Each vertex appears in at most one decomposition piece by definition.
\end{proof}

The main difference between the statement of Proposition~\ref{prop:aalg18} and that needed to generalize Corollary~\ref{cor:strong_erd_constant_wratio} beyond the bounded weight ratio case is that Proposition~\ref{prop:aalg18} measures the cut edges by the amount of total weight cut, rather than the number of edges cut. Indeed, for a general $n$-vertex, $m$-edge graph $G = (V, E, \ww)$ with $\ww \in [1, U]^E$ but where $U$ may be superconstant, let $W \defeq \sum_{e \in E} \ww_e$, and let $G' = (V, E', \ww_{E'})$ where $E' \subseteq E$ removes any edge in $E$ with weight larger than $\frac{4W}{m}$ (so $|E'| \ge \frac{3m} 4$). Applying Proposition~\ref{prop:aalg18} with any constant $C$ on $G'$ yields
\[\Par{\max_{e \in E(G_j)} \ww_e}\Par{\max_{u, v \in V_j} \ER_G(u, v)} = O\Par{\frac W m} \cdot O\Par{\frac n W} = O\Par{\frac n m},\]
as desired. Unfortunately, the claim \eqref{eq:aalg18_edges_cut} does not imply few edges are cut in this case, though for sufficiently large $C$, it does imply only a small fraction of total weight is cut. 

We conclude this section by mentioning one barrier to improving the guarantees of \cite{AlevALG18}, towards obtaining a variant of Corollary~\ref{cor:strong_erd_constant_wratio} which holds for superconstant weight ratios $U$. 
In particular, no single decomposition of $G$'s vertices can simultaneously guarantee a bounded effective resistance diameter while cutting a small number of edges, as the following example demonstrates.

Let $H$ be a path graph with all edge weights $1$, and let $G$ equal $H$ plus a clique with edge weights $n^{-4}$. Since $\LL_H \preceq \LL_G \preceq 2 \LL_H$, we have $\frac 12 \ER_H(u,v) \leq \ER_G(u,v) \leq \ER_H(u,v)$ for any vertices $u,v$. We claim that any vertex-disjoint partition of $G$ which cuts at most $\frac m 2$ edges must have one component with resistance diameter $\Omega(n)$. Indeed, any partition $P_1, P_2, \ldots P_k$ with $|P_i| = n_i$ does not cut exactly $\sum_{i=1}^k \frac{n_i (n_i - 1)}{2}$ edges: as this must be more than $\frac m 2$, we have 
\[
\frac{m}{2} \leq \sum_{i\in[k]} \frac{n_i (n_i - 1)}{2} \leq \frac{1}{2} \left( \max_{i \in [k]} n_i - 1 \right) \sum_{i\in[k]} n_i = \frac{n}{2} \left( \max_{i \in [k]}n_i - 1 \right). 
\]
Since $m = \frac{n(n-1)}{2}$, the largest partition piece has $\ge \frac{n}{2}$ vertices, and since any path of length $k$ has resistance diameter $k$, this piece has diameter $\Omega(n)$. Thus any potential application of the techniques of \cite{AlevALG18} towards improving Proposition~\ref{prop:er_partition} must partition its input by both vertices and edges; extending \cite{AlevALG18}'s approach to subsets of edges is an intriguing open question.

Finally, we mention that the definition of a graph decomposition highlighted in this work, the ER decomposition of Definition~\ref{def:er_partition}, may not be the only useful notion of decomposition for constructing Eulerian sparsifiers. A potentially fruitful open direction is to explore other related decomposition notions, for which there may be better bounds bypassing difficulties with ER decompositions.

\section{Proof of Proposition~\ref{prop:bjm_measure}} \label{app:concentration}

In this section, we show how to modify the proof of Lemma 3.1 in \cite{BansalJM23} to yield the tighter concentration bound claimed in Proposition~\ref{prop:bjm_measure}. In particular, we show how to obtain the second argument in the minimum, since the first was already shown by \cite{BansalJM23}. To do so, we recall the following known concentration bounds from \cite{Tropp18, BandeiraBvH21}.

\begin{proposition}[Corollary 3.6, \cite{Tropp18}]\label{prop:tropp_second}
    Let $n \geq 8$ and $\{\AA_i\}_{i\in [m]} \in \Sym^n$ satisfy
    $\normsop{\sum_{i \in [m]} \AA_i^2} \leq \sigma^2$ and 
    $\max_{\UU,\VV,\WW \in \Uni^n}\normsop{\sum_{i,j \in [m]} \AA_i \UU \AA_j \VV
    \AA_i \WW \AA_j} \leq w$.
    Then, for $\gg \sim \Nor(\vzero_m,\II_m)$, there is a universal constant
    $\ctro$
    such that
    \[
        \E \normop{\sum_{i \in [m]} \gg_i \AA_i} \leq \ctro \cdot \Par{\sigma
        \log^{\frac 1 4} n + w \log^{\frac 1 2} n}
    \]
\end{proposition}

\begin{lemma}[Proposition 4.6, \cite{BandeiraBvH21}]\label{lem:wbound_bbh21}
    For $\{\AA_i\}_{i \in [m]} \in \Sym^n$,
    \begin{equation}
        \max_{\UU,\VV,\WW \in \Uni^n}\normop{\sum_{i,j \in [m]} \AA_i \UU \AA_j \VV
        \AA_i \WW \AA_j} 
        \leq
        \normop{\sum_{i \in [m]} \AA_i^2}^{\frac 1 2} \cdot
        \normop{\sum_{i \in [m]}
        \textup{vec}(\AA_i)\textup{vec}(\AA_i)^\top}^{\frac 1 2},
    \end{equation}
    where $\textup{vec}(\AA) \in \R^{n^2}$ is the vectorization of $\AA$.
\end{lemma}

By combining Proposition~\ref{prop:tropp_second} and Lemma~\ref{lem:wbound_bbh21}, we have the following corollary.

\begin{corollary}\label{cor:concentration}
    Let $n \geq 8$ and $\{\AA_i\}_{i\in [m]} \in \Sym^n$ satisfy $\normsop{\sum_{i \in [m]} \AA_i^2} \leq \sigma^2$, $ \sum_{i \in [m]} \normsf{\ma_i}^2
    \leq mf^2$.
    Then, for $\gg \sim \Nor(\vzero_m,\II_m)$, there is a universal constant
    $\ctro$
    such that
    \[
        \E \normop{\sum_{i \in [m]} \gg_i \AA_i} \leq \ctro \cdot \Par{\sigma
        \log^{\frac 1 4} n + \sqrt{\sigma f} \log^{\frac 1 2} n}.
    \]
\end{corollary}
\begin{proof}
It suffices to combine Proposition~\ref{prop:tropp_second} and Lemma~\ref{lem:wbound_bbh21}, where we use
\begin{align*}
\normop{\sum_{i \in [m]}
	\textup{vec}(\AA_i)\textup{vec}(\AA_i)^\top} \le \frac 1 m \Tr\Par{\sum_{i \in [m]} \textup{vec}(\ma_i)\textup{vec}(\ma_i)^\top} = \frac 1 m \sum_{i \in [m]}  \normf{\ma_i}^2.
\end{align*}
The first inequality used that the summed vectorized outer products has rank at most $m$.
\end{proof}

By replacing Theorem 1.2 of \cite{BandeiraBvH21} with \Cref{cor:concentration} in the proof of Lemma 3.1 in \cite{BansalJM23}, we obtain the second term in Proposition~\ref{prop:bjm_measure}; we may use the better of the two bounds. To handle the $n \ge 8$ constraint, for any smaller $n$, we can pad with zeroes up to dimension $n = 8$, which does not affect any operator norms and only changes constants in the claim.

\end{appendix}

\end{document}